\documentclass[a4paper,UKenglish,cleveref,autoref, thm-restate]{lipics-v2021}

\usepackage{environ}
\makeatletter
\let\@frontnotes\@empty
\newcommand\frontnote[2]{%
  \g@addto@macro\@frontnotes{%
    \vskip\topmattervskip\baselineskip\noindent
    \textcolor{lipicsGray}{\fontsize{9}{11}\sffamily\bfseries #1\enskip}#2}}
\long\def\frontnote@swap#1#2{\frontnote{#2}{#1}}
\NewEnviron{frontnotes}[1]{\expandafter\frontnote@swap\expandafter{\BODY}{#1}}
\g@addto@macro\endabstract{\@frontnotes}
\makeatother

 \hideLIPIcs  

\title{Graph Distortion Distances, Gromov--Hausdorff and the Hardness of Gromov--Wasserstein} 

\titlerunning{Graph Distortion Distances, Gromov--Hausdorff and Gromov--Wasserstein}

\author{Jens Breitung}{RWTH Aachen University, Germany}{breitung@informatik.rwth-aachen.de}{https://orcid.org/0009-0006-6697-5275}{European Union (ERC, SymSim, 101054974)}

\author{Timo Gervens}{RWTH Aachen University, Germany}{gervens@informatik.rwth-aachen.de}{https://orcid.org/0000-0002-1224-9853}{European Union (ERC, SymSim, 101054974)}

\author{Julius Pörzgen}{TU Dortmund University, Germany}{julius.poerzgen@tu-dortmund.de}{https://orcid.org/0009-0007-7077-1184}{}

\authorrunning{J. Breitung and T. Gervens and J. Pörzgen} 

\Copyright{Jens Breitung and Timo Gervens and Julius Pörzgen} 



\category{} 

\relatedversion{} 

\acknowledgements{We thank Robin Westermann for the proof of an earlier version of 
\cref{subclaim:bounding_inter_layer_mapping} and Martin Grohe for his valuable advice and for helpful feedback on an earlier draft of this paper.}

\nolinenumbers 

\EventEditors{John Q. Open and Joan R. Access}
\EventNoEds{2}
\EventLongTitle{42nd Conference on Very Important Topics (CVIT 2016)}
\EventShortTitle{CVIT 2016}
\EventAcronym{CVIT}
\EventYear{2016}
\EventDate{December 24--27, 2016}
\EventLocation{Little Whinging, United Kingdom}
\EventLogo{}
\SeriesVolume{42}
\ArticleNo{23}

\usepackage[textsize=tiny]{todonotes}
\usepackage{csquotes}

\usepackage{refcount}

\usepackage{bm}
\usepackage{bbm}
\usepackage{complexity}
\usepackage{mathtools}
\usepackage{nccmath}
\usepackage{xfrac}

\usepackage{environ}   
\usepackage{hyperref}
\usepackage[capitalise]{cleveref}
\AddToHook{cmd/appendix/before}{\crefalias{section}{appendix}}

\usepackage{adjustbox}

\usepackage[section]{algorithm}
\usepackage{algorithmicx}
\usepackage{algpseudocode}

\usetikzlibrary{
    patterns,
    patterns.meta,
    calc,
    math,
    arrows,
    arrows.meta,
    automata,
    positioning,
    shapes.symbols,
    shapes.geometric,
    decorations.pathmorphing,
    decorations.pathreplacing,
    shapes,
    nfold,
    fit
}
\usepackage{etoolbox}

\newcommand{\WLOG}{\text{w.l.o.g.\@}}
\newcommand{\WLOGcap}{\text{W.l.o.g.\@}}
\newcommand{\ie}{\text{i.e.\@}}

\newcommand{\multiset}[1]{\{\!\!\{#1\}\!\!\}}

\newcommand{\N}{\ensuremath{\mathbb{N}}}

\renewcommand{\R}{\ensuremath{\mathbb{R}}}
\newcommand{\Rnn}{\ensuremath{\mathbb{R}_{\geq 0}}}

\newcommand{\Bij}{\ensuremath{\textup{Bij}}}
\newcommand{\Perm}{\ensuremath{\textup{Perm}}}
\newcommand{\Corr}{\ensuremath{\textup{Corr}}}

\newcommand{\DS}{\ensuremath{\textup{DS}}}
\newcommand{\supp}{\ensuremath{\textup{supp}}}

\renewcommand{\O}{\ensuremath{\textup{O}}}
\newcommand{\GSim}[1]{\textup{\textsc{GSim}$\big(#1\big)$}}

\newcommand{\GW}{\textup{\textsf{GW}}}
\newcommand{\OPT}{\ensuremath{\textup{\textsf{OPT}}}}
\newcommand{\threepartition}{\textsc{Three-Partition}}
\newcommand{\hamcycle}{\textsc{Hamiltonian Cycle}}
\newcommand{\hampath}{\textsc{Hamiltonian Path}}
\newcommand{\bandwidth}{\textsc{Bandwidth}}

\newcommand{\norm}[1]{\ensuremath{\left \lVert #1 \right \rVert}}
\newcommand{\dis}{\ensuremath{\Delta}}
\newcommand{\dist}{\ensuremath{\delta}}
\newcommand{\Ddist}[1]{\ensuremath{\dist_{\textup{\textsf{D}}, #1}}}
\newcommand{\Adist}[1]{\ensuremath{\dist_{\textup{\textsf{A}}, #1}}}
\newcommand{\Mdist}[1]{\ensuremath{\textup{\textsf{d}}_{#1}}}
\newcommand{\GH}{\ensuremath{\textup{\textsf{GH}}}}
\newcommand{\Gdist}{\ensuremath{\textup{dist}}}     

\newcommand{\etal}{et al.\@}
\newcommand{\memoli}{\text{Mémoli}}
\newcommand{\hastad}{\text{Håstad}}

\renewcommand{\G}{\ensuremath{\mathcal G}}
\newcommand{\diam}{\ensuremath{\textup{diam}}}
\newcommand{\ecc}{\ensuremath{\textup{ecc}}}
\newcommand{\X}{\ensuremath{\mathcal{X}}}
\newcommand{\Y}{\ensuremath{\mathcal{Y}}}
\newcommand{\len}{\ensuremath{\textsf{len}}}
\newcommand{\wz}{
    \ensuremath{
        \mathrlap{\raisebox{-0.1ex}{$\widetilde{\phantom{z}}$}}%
        z%
    }
}
\newcommand{\arr}[1]{
    \ensuremath{
        \xrightarrow{\smash{\raisebox{-0.3ex}{\scriptsize$#1$}}}
    }
}
\newcommand{\x}{\ensuremath{\bm{x}}}
\newcommand{\eps}{\ensuremath{\varepsilon}}
\newcommand{\yescase}{\proofsubparagraph*{Completeness}}
\newcommand{\nocase}{\proofsubparagraph*{Soundness}}
\newcommand{\ceq}{\ensuremath{\coloneqq}}
\newcommand{\midd}{\ensuremath{\, \middle| \, }}
\DeclarePairedDelimiter{\abs}{\lvert}{\rvert}
\DeclareMathOperator*{\argmin}{arg\,min}
\newcommand{\indicator}{\ensuremath{\mathbbm{1}}}

\makeatletter
\@ifclasswith{lipics-v2021}{anonymous}
  {\newcommand{\pc}[1]{\textcolor{red}{{Anonymous contributor}}}}   
  {\newcommand{\pc}[1]{#1}}          
\makeatother

\AddToHook{cmd/appendix/after}{%
  \crefalias{subsection}{appendix}%
  \crefalias{subsubsection}{appendix}%
}

\begin{document}

\maketitle

\allowdisplaybreaks[4]

\begin{abstract}
    We study the computational complexity of 
    graph distortion distances, a family of graph metrics 
    measuring how vertex distances differ between two graphs.
    Formally, these distances are defined in terms of standard matrix norms on the difference of distance matrices under an optimal vertex alignment.
    The resulting metrics have natural interpretations, corresponding to three ways of measuring the distortion of shortest-path distances:
    the \emph{maximum} distortion among all vertex pairs (entrywise maximum norm),  \emph{global} aggregation (entrywise $p$-norms), and \emph{local} aggregation (operator $p$-norms) of distortions.
    
    For maximum distortion, we show that approximation below a factor of two is \NP-hard even for highly restricted trees. 
    Due to the close connection to the Gromov--Hausdorff distance, we obtain a similar hardness result for the latter.
    In both cases, we show essentially optimal additive inapproximability.
    On the algorithmic side, we give a subcubic-time algorithm for optimal embeddings of connected graphs into paths, based on a structural decomposition along a diametral path. 
    Additionally, we give a simple approximation algorithm when one input graph has a \enquote{relatively path-like} metric structure, \ie{} has bounded pathlength.

    For local and global distortion, we derive various hardness results even for restricted inputs.
    Making use of a close connection between the global distortion distances and Gromov--Wasserstein distances, we give the first formal proof of \NP-hardness of the latter, which has been claimed in the literature many times but 
    has, to the best of our knowledge, never been correctly proven.
\end{abstract}

\newpage

\section{Introduction}
Graphs are a fundamental model throughout computer science, mathematics, and numerous application areas. 
They are used to describe objects and processes ranging from chemical compounds and biological systems to social networks, transportation systems, and knowledge graphs, making graph comparison an important task.
Naturally, there are many different ways to determine how (dis)similar two graphs are.
For instance, there are approaches based on homomorphism counting \cite{boker21_homomorphism_counts, borgs_cut_distance_bound, seppelt_2024_dissertation}, optimal alignment of adjacency matrices \cite{ arvind_approx_graph_isomorphism, bunke97_edit_distance,gervens2022graphsimilaritybasedmatrix, grohe_rattan_woeginger_edit_distance, odonnell14_edit_distance}, spectral similarity \cite{kolla18_spectral_similarity, tsitsulin18_spectral_similarity, umeyama99_spectral_similarity}, optimal transport \cite{barbe20_optimal_transport, maretic19_optimal_transport, Memoli2011}, and behavioural equivalence \cite{sangiorgi11_behavioral_equivalence, glabbeek01_behavioral_equivalence}.
See \cite{grohe2025thoughtsgraphsimilarity} for a survey highlighting connections between these approaches.

\subparagraph{Adjacency-Based Graph Metrics} In the most basic setting, graphs are unweighted and simple. 
Here, a well-studied measure of graph similarity,
or more precisely a \emph{graph metric} (see \cref{section:preliminaries}), 
is the \emph{edit distance} \cite{arvind_approx_graph_isomorphism,bunke97_edit_distance,grohe_rattan_woeginger_edit_distance,odonnell14_edit_distance}, the minimum number of edges that need to be added or deleted to transform one graph into an isomorphic copy of the other.
An equivalent definition can be given in terms of the Frobenius norm of the difference of adjacency matrices under an optimal bijective alignment between the vertex sets of the graphs \cite{grohe_rattan_woeginger_edit_distance}.
A generalisation of this was studied by Gervens and Grohe \cite{gervens2022graphsimilaritybasedmatrix}:
Each matrix norm $\norm{\cdot}$ gives rise to a graph metric\footnote{Strictly speaking, this only holds for \emph{permutation-invariant} matrix norms, see \cite{gervens2022graphsimilaritybasedmatrix}.}
(on graphs of the same order), that is
\begin{equation}
    \label{equation:adjacency_distance_template}
    \Adist{\norm{\cdot}}(G, H) \ceq \min_{\pi} \norm{A_G^\pi - A_H}
\end{equation}
where $A_G, A_H$ denote the adjacency matrices of $G, H$, respectively. Here, $\pi$ ranges over all bijections from the vertex set of $G$ to that of $H$, and $A_G^\pi$ denotes the matrix obtained from $A_G$ by permuting the rows and columns according to $\pi$.

These graph metrics capture the structure of \emph{edge mismatches} between graphs, with standard matrix norms yielding very natural interpretations for the corresponding graph metrics (compare \cite{gervens2022graphsimilaritybasedmatrix,grohe2025thoughtsgraphsimilarity}).
For the entrywise maximum norm, this corresponds to the graph isomorphism distance (isomorphic graphs have distance zero while non-isomorphic graphs have distance one). 
All other entrywise $p$-norms are equivalent to the edit distance and essentially count the number of edge mismatches \emph{globally}.
The operator $p$-norms adopt a more \emph{local} view, with the operator $1$-norm being the maximum number of edge mismatches incident to a single vertex. Meanwhile, other $p$-norms also account for the mismatch structure in the connected component of the vertex (see \cite[Section 4]{gervens2022graphsimilaritybasedmatrix}).

\subparagraph{From Adjacency to Distortion}
The above graph metrics do not use the 
underlying distance structure of the input graphs (although in principle distances between vertices could easily be computed from the adjacency matrices).
However, in many application areas 
this distance information is of great relevance. 
In chemistry, the Wiener index, defined as the sum of shortest-path distances between all vertex pairs, 
is an important molecular descriptor \cite{wiener_in_chemistry, wiener_og}.
In network science, closeness centrality assesses the importance of a vertex by its average 
shortest-path distance to all other vertices \cite{network_science_standard}. 

An approach making explicit use of the shortest-path distances is to view (connected) graphs as metric spaces.
This allows the use of well-known similarity measures\footnote{To avoid ambiguity, here we refer to metrics defined on the space of metric spaces as similarity measures.} on metric spaces such as the Gromov--Hausdorff \cite{gromov_original_paper,memoli_properties_of_gromov_hausdorff} and Gromov--Wasserstein distances \cite{memoli_gw_original, Memoli2011} to compare graphs.
While these similarity measures have received considerable attention in the context of metric spaces in general, they have received less attention in the context of graphs.
In this work, we address this gap by combining previous work on graph similarity and metric space theory, making explicit use of distance information in the graphs.

\subparagraph{Distortion-Based Graph Metrics}
We replace
the adjacency matrices $A_G, A_H$ in \cref{equation:adjacency_distance_template} by the distance matrices $D_G, D_H$, respectively.
(The distance matrix of a graph records the shortest-path distance between every pair of vertices.)
Specifically, for a matrix norm $\norm{\cdot}$, we consider the \emph{graph distortion distance} defined by
\begin{equation}
    \label{equation:distortion_distance_template}
    \Ddist{\norm{\cdot}}(G, H) \ceq \min_{\pi} \norm{D_G^\pi - D_H}
\end{equation}
where $G, H$ are connected graphs of the same order. Here, connectivity is crucial to ensure that $\norm{D_G^\pi - D_H}$ is well-defined.
Intuitively, we no longer consider edge mismatches, but rather the \emph{(additive) distortion} of shortest-path distances, that is, the absolute change in shortest-path distance between pairs of vertices under a bijective alignment of the vertex sets.
This yields a natural comparison between graphs when viewed as metric spaces.

In contrast to the adjacency-based graph metrics, this allows for a meaningful comparison between sparse and dense graphs: Insights from graph sparsification, specifically the theory of spanners and emulators, show that the shortest-path distance structure of dense graphs can be approximately preserved by sparse graphs \cite{ahmed_graph_spanners_tutorial}. 

Similar to the adjacency-based graph metrics, our graph metrics have natural interpretations.
For the entrywise maximum norm, we obtain the \emph{maximum distortion distance} $\Ddist{(\infty)}$: the largest value $d$ such that every vertex alignment necessarily distorts the distance of some pair of vertices by at least $d$.
Note that the Gromov--Hausdorff distance can be recovered from $\Ddist{(\infty)}$
by relaxing the optimisation domain from bijections to \emph{correspondences}, \ie{} only requiring
every vertex in each graph to be paired with at least one vertex of the other graph.

The remaining entrywise $p$-norms yield graph metrics, denoted $\Ddist{(p)}$, that aggregate the distortions of all vertex pairs and may therefore be viewed as \enquote{global} measures of metric distortion, with larger values of $p$ placing increasing emphasis on the most severely distorted pairs and smaller values averaging distortion more uniformly across all pairs. 
Closely related are the Gromov--Wasserstein distances, which compare two weighted metric spaces in terms of their geometric discrepancy under an optimal soft assignment between the point sets.
For connected graphs of the same order with uniform distribution on the vertex sets, the $p$-Gromov--Wasserstein distance
coincides, up to a factor, with
$\Ddist{(p)}$
when replacing bijections with
fractional vertex alignments, see \cref{subsection:gw_derivation}.
In contrast, the operator $p$-norms yield metrics denoted $\Ddist{p}$ that provide a more \enquote{local} perspective. 
For example, $\Ddist{1}$ measures the maximum aggregate distortion of distances from a single vertex to all other vertices.

\subparagraph{Our Results}
For a graph metric $\dist$, we study the computational complexity of $\GSim{\dist}$, the problem of computing $\dist(G, H)$ for input graphs $G, H$, or, depending on the context, its decision version (decide whether $\dist(G, H) \leq d$ for a given $d$).
While computing $\Adist{(\infty)}$ is equivalent to \textsc{Graph Isomorphism} and can be solved
in quasipolynomial time \cite{babai_graph_iso_quasipolynomial}, computing $\Ddist{(\infty)}$ is much harder (under 
standard complexity-theoretic assumptions) even for highly restricted graph classes.
Unless stated otherwise, $n$ denotes the order of the input graphs.

\begin{theorem}[label=dummy1,restate=thmMaxDistortionHardness]
    \label{theorem:max_distortion_hardness}
    For any $\eps > 0$, unless $\P = \NP$, there is no polynomial-time approximation algorithm for $\GSim{\Ddist{(\infty)}}$ with multiplicative error guarantee $2-\eps$ or additive error guarantee $\O(n^{1-\eps})$, even if both input graphs are trees of pathwidth two, or if both input graphs are trees of maximum degree three.
\end{theorem}

\noindent A simple modification also leads to a similar result for the Gromov--Hausdorff distance (\cref{theorem:gh_inapprox_result_julius}).
In both cases, our results regarding additive inapproximability are essentially asymptotically tight
because both metrics admit a trivial approximation algorithm with additive error guarantee $\O(n)$.

Contrasting the hardness of \GSim{\Ddist{(\infty)}} for graphs
of almost minimal pathwidth, we find tractability when restricting one input graph to be a path, even when leaving the other input graph completely unrestricted. 

\begin{theorem}[label=dummy2,restate=thmMaxDistortionIntoPath]
    \label{theorem:max_distortion_graph_into_path}
    Let $G$ be a connected $n$-vertex graph and $P_n$ be the $n$-vertex path. Then
    $\Ddist{(\infty)}(G, P_n) = n - 1 - \diam(G)$.
    Moreover, a minimising bijection can be computed in time $\O(n^\omega \log n)$ where $\omega$ is 
    the best known exponent for multiplication of two $n\times n$ matrices.
\end{theorem}
This is surprising given the close resemblance
to the notoriously hard \bandwidth{} problem, which minimises the maximum stretch of an edge under a bijective embedding into a path.
Relaxing the extremal distance structure of path graphs, we obtain a simple approximation algorithm for graphs with a \enquote{path-like} metric structure, that is, graphs of bounded pathlength.

\begin{theorem}
    \label{theorem:pathlength_approx}
    There is a cubic-time approximation algorithm for $\GSim{\Ddist{(\infty)}}$ with a combined guarantee of multiplicative error $5$ and additive error $12k+12$ if one input graph has pathlength at most $k$.
    If both input graphs have pathlength at most $k$, the error guarantees can be improved to $\frac{7}{4}$ and 
    $\frac{11}{2}k + 11$, respectively.
\end{theorem}

In the remainder of the paper, we shift our attention towards local and global distortion distances,
establishing multiplicative and additive inapproximability results (under the assumption $\P \neq \NP$) using differing techniques. 
We initially give a simple \enquote{centre vertex} construction that transfers the $\NP$-hardness of the 
corresponding adjacency-based metrics to both 
the local metrics $\Ddist{p}$ (\cref{theorem:local_dist_hardness_hamiltonian}) and the global metrics 
$\Ddist{(p)}$ with $p < \infty$ (see \cref{subsection:global_distortion}).
The proof of \cref{theorem:local_dist_hardness_hamiltonian} already yields constant-factor inapproximability for $\Ddist{p}$.
A different construction shows that $\Ddist{1}$ remains 
hard to approximate within a constant factor 
on trees of bounded diameter (\cref{theorem:local_dist_mult_hardness_trees}).
For the global distortion distances $\Ddist{(p)}$, we rule out the existence of a
\PTAS{} (see \cref{subsection:global_distortion}). 
Moreover, by suitably amplifying the centre vertex construction, we obtain an additive inapproximability result (\cref{theorem:additive_inapprox_global_distortion}). 

Finally, based on the connection between $\Ddist{(p)}$ and the $p$-Gromov--Wasserstein distance,
we establish the \NP-hardness of the latter. To the best of our knowledge, this is the first formal proof of this fact (see \cref{appendix:gw_literature} for a detailed discussion).

\begin{theorem}
    \label{theorem:gw_hard}
    For $1 \leq p \leq \infty$, computing the $p$-Gromov--Wasserstein distance is \NP-hard.
    This holds even if both input spaces are simple graphs of the same order with uniform distribution on the vertex sets.
\end{theorem}

\subsection*{Related Work}
Most metric embedding literature considers \emph{multiplicative distortion}, often under additional restrictions such as requiring embeddings to be non-contractive \cite{sidiropoulos_line_embedding_mult}.
Papadimitriou and Safra showed that the multiplicative distortion cannot be approximated below a factor of three, even for point sets in $\R^3$ \cite{papadimitriou_mult_distortion_mult_hardness_3d}.
For one dimension, subsequent work \cite{approximating_the_distortion} establishes a weaker inapproximability result, while Kenyon, Rabani, and Sinclair found an exact polynomial-time algorithm for distortion at most $3+2\sqrt{2}$ \cite{kenyon_line_metrics}.
On the positive side, Fellows \etal{} give exact algorithms for embeddings into the line under suitable parametrisations \cite{fellows2008parameterizedlowdistortionembeddings}.

The notion of additive distortion is mostly considered in the context of coarse graph theory
\cite{Georgakopoulos2025_coarse_graph_theory} and spanners \cite{ahmed_graph_spanners_tutorial}. 
Chartrand, Kubicki, and Schultz investigated structural properties of our graph metric $\Ddist{(1)}$ 
\cite{chartrand_global_bijective_distortion}; however, they do not settle its computational tractability.
In the metric embedding literature, most research considers mapping metric spaces into various line 
metrics. For instance, \hastad{} \etal{} consider embedding arbitrary distance matrices optimally
into the real
line with applications in computational biology \cite{matrix_to_line_HastadIvanssonLagergren1998}.
Resembling our setting more closely, Hall and Papadimitriou consider maximum additive distortion
into and between \emph{given} subsets of the real line and obtain simple constant-factor approximations 
\cite{approximating_the_distortion}.

In the more general setting of comparing two arbitrary metric spaces, two common similarity measures are the 
Gromov--Hausdorff and Gromov--Wasserstein distances. The former has been studied extensively for metric 
trees, yielding constant-factor inapproximability 
\cite{agarwal_et_al_GH_NP_hard_metric_trees,Schmiedl2017} as well as (parametrised) 
approximation algorithms based on merge trees
\cite{agarwal_et_al_GH_NP_hard_metric_trees,farahbakhshtouli_fpt_gh_14}.
Moreover, there are constant-factor approximations for subsets of the real line 
\cite{wenk_GH_approx_R1} and plane \cite{GH_approx_R2}.
The Gromov--Wasserstein distance is mainly studied empirically with many heuristic approaches and variants or extensions proposed in recent years \cite{GW-and-WL, peyre_gromov_wasserstein_discrepancy_coupling, solomon_entropic_gw,vayer_fused_gw,vayer_sliced_gw}.
\section{Preliminaries}
\label{section:preliminaries}
By $\R$ and $\Rnn$ we denote the real and nonnegative real numbers, respectively.
We denote by $\N$ the set of nonnegative integers.
For every positive $n \in \N$, we let $[n] \ceq \{1, \ldots, n\}$.
For sets $A, B$ with $|A| = |B|$, we denote by $\Bij(A, B)$ the set of all bijections from $A$ to $B$.

\subparagraph*{Matrices and Matrix Norms}
We consider matrices with real entries and with rows and columns indexed by arbitrary finite sets. 
For finite sets $V, W$, the set of all $V \times W$ matrices is denoted $\R^{V \times W}$, and we denote the entries of a matrix $A \in \R^{V \times W}$ by $A(v, w)$ for $v \in V, w \in W$.
For a square matrix $A \in \R^{V \times V}$ and a bijection $\pi: V \to W$, we let $A^\pi$ be the
$W \times W$-matrix with entries $A^\pi(w, w') \ceq A(\pi^{-1}(w), \pi^{-1}(w'))$.
For finite sets $V, W$ with $|V| = |W|$, 
a matrix $P \in \{0, 1\}^{V \times W}$ is a \emph{permutation matrix} if $P$ has exactly one $1$-entry in each row and column.
We denote the set of $V \times W$ permutation matrices by $\Perm(V, W)$.
A matrix $Q \in [0, 1]^{V \times W}$ is \emph{doubly stochastic} if its row and column sums equal one.
The set of doubly stochastic $V \times W$ matrices is denoted $\DS(V, W)$.

A function $\norm{\cdot} : \R^{V \times V} \to \Rnn$ is a \emph{(real) matrix norm} if for all $A, B \in \R^{V \times V}$ 
it holds that $\norm{A} = 0$ iff $A = 0$, $\norm{cA} = |c|\norm{A}$ for all $c \in \R$, and $\norm{A + B} \leq \norm{A} + \norm{B}$.
The following common matrix norms will be important. Let 
$A \in \R^{V \times V}$.
\begin{enumerate}
    \item Entrywise $p$-norm: $\norm{A}_{(p)} \ceq \left( \sum_{v, v' \in V} |A(v, v')|^p \right)^{1/p}$ for $1 \leq p < \infty$. For $p=\infty$, $\norm{A}_{(
    \infty)} \ceq \max_{v, v' \in V} |A(v, v')|$.
    \item Operator $p$-norm: $\norm{A}_p \ceq \max_{\x \in \R^V, \norm{\x}_p = 1} \norm{A\x}_p$, where the vector $p$-norm is defined by $\norm{\x}_p \ceq \left( \sum_{v \in V} |\x(v)|^p \right)^{1/p}$ for $1 \leq p < \infty$ and $\norm{\x}_{\infty} \ceq \max_{v \in V} |x(v)|$.
\end{enumerate}

\subparagraph*{Metric Spaces}
A \emph{metric space} is a pair $(X,d)$, where $X$ is a set and $d : X \times X \to \Rnn$ is a function satisfying $d(x, y) = 0$ if and only if $x=y$, $d(x,y) = d(y,x)$, and $d(x,z) \leq d(x,y) + d(y,z)$ for all $x,y,z \in X$.
The function $d$ is a \emph{metric} on $X$. We write $|X|$ for the \emph{cardinality} of $X$.
A \emph{weighted compact metric space} $\X$ is a triple $(X,d,\mu)$, where $(X,d)$ is a compact metric space and $\mu$ is a finite Borel measure on $X$. We refer to $\mu$ as the weight measure of the space.

\subparagraph*{Graphs} We use standard graph notation \cite{diestel2025graph}; graphs are simple, finite, and undirected.
The class of all graphs is denoted by $\mathcal G$.
Let $G$ and $H$ be graphs with $v, v' \in V(G)$.
We denote edges by $vv'$ instead of $\{v, v'\}$. 
The \emph{adjacency matrix} of $G$ is denoted  $A_G$.
By $G \cong H$, we denote that $G$ and $H$ are isomorphic.
We denote the
\emph{open} and \emph{closed neighbourhoods} of $v$ in $G$ by $N_G(v)$ and $N_G[v]$, respectively.
The \emph{degree} of $v$ is $\deg_G(v) \ceq |N_G(v)|$.
For $v, v'$, a path $P$ from $v$ to $v'$ is a sequence of distinct vertices $w_1, \ldots, w_\ell \in V(G)$ with $w_1 = v$, $w_\ell = v'$, and $w_iw_{i+1} \in E(G)$ for $i \in [\ell - 1]$. Note that paths may have length zero, \ie{} consist only of a single vertex.
The length of $P$ is the number of its edges.
For $n \geq 1$, we denote by $P_n, C_n$, and $K_n$ the \emph{path graph}, the \emph{cycle graph}, and the \emph{complete graph} on $n$ vertices, respectively.
For vertices $v, v' \in V(G)$ of a connected graph $G$, $\Gdist_G(v, v')$ denotes the length of a shortest path between $v$ and $v'$ in $G$.
We denote the \emph{diameter} of $G$ by $\diam(G)$. The \emph{eccentricity} of a vertex $v$ in $G$ is $\ecc_G(v) \ceq \max_{v' \in V(G)} \Gdist_G(v, v')$.
The \emph{distance matrix} of $G$ is the matrix $D_G \in \N^{V(G) \times V(G)}$ with entries $D_G(v, v') \ceq \Gdist_G(v, v')$ for all $v, v' \in V(G)$.
For a bijection $\pi \in \Bij(V(G), V(H))$, we define the distortion incurred by the pair $\{v, v'\}$ under $\pi$ as 
\[
    \dis(\pi, v, v') \ceq \big|\Gdist_G(v, v') - \Gdist_H(\pi(v), \pi(v'))\big|.
\]

\noindent We define metrics on $\G$ 
following \cite{gervens2022graphsimilaritybasedmatrix} but restricted to connected graphs of the same order.
\begin{definition}[Graph Metrics]
    \label{definition:graph_metric}
    A \emph{graph metric} is a partial function $\delta:\G \times \G \to \Rnn$
    such that for all connected graphs $F, G, H$ of the same order, it holds that
    \begin{description}
        \item[\textup{(GM0)}] $\delta(G,H)=\delta(G',H')$ for all $G', H'$ such that $G\cong G'$ and 
          $H\cong H'$; 
        \item[\textup{(GM1)}] $\delta(G,G)=0$;
        \item[\textup{(GM2)}] $\delta(G,H)=\delta(H,G)$;
        \item[\textup{(GM3)}] $\delta(F,H)\le \delta(F,G)+\delta(G,H)$.
    \end{description}
\end{definition}
\noindent Note that our graph metrics are, strictly speaking, \emph{pseudometrics}, as we do not require $\dist(G,H)>0$ for all $G\neq H$. 
However, all graph metrics considered below satisfy $\dist(G,H)>0$ whenever $G\not\cong H$.
Thus, on the isomorphism classes, these define proper metrics.
The proof of the following \cref{lemma:graph_metric} and all subsequent proofs are deferred to the appendix. 

\begin{lemma}
    \label{lemma:graph_metric}
    $\Ddist{p}$, $\Ddist{(p)}$, and $\Ddist{(p)}^{\GW}$ 
    (defined in \cref{subsection:gw_derivation})
    are graph metrics for $1\leq p\leq\infty$.
\end{lemma}

\subparagraph*{Approximation Algorithms}
An algorithm $\mathcal A$ for a minimisation problem is an
\emph{$(\alpha,\beta)$-approximation algorithm} if, for every instance
$\mathcal I$, it outputs a feasible solution $\mathcal A(\mathcal I)$ with
$\OPT(\mathcal I) \leq \mathrm{val}(\mathcal A(\mathcal I)) \leq \alpha \cdot \OPT(\mathcal I) + \beta$,
where $\OPT(\mathcal I)$ is the optimal value of $\mathcal I$. That is, $\mathcal A$ has a \emph{multiplicative error guarantee} of $\alpha$ and an \emph{additive error guarantee} of $\beta$.
A \emph{polynomial-time approximation scheme}, denoted \PTAS{}, is an algorithm that, given an
instance $\mathcal I$ of a minimisation problem and a parameter $\eps > 0$ as input,
returns a $(1 + \eps, 0)$-approximation
and runs in time polynomial in $\abs{\mathcal I}$ for every fixed $\eps$.

\section{Hardness of Maximum Distortion}
\label{section:hardness_max_distortion}

In this section, we investigate the maximum distortion distance and the closely related Gromov--Hausdorff distance.
Recall that the maximum distortion distance $\Ddist{(\infty)}$ is derived by applying the entrywise $\infty$-norm on the difference of aligned distance matrices.
Intuitively, this graph metric measures the maximum discrepancy between the shortest-path distance of a vertex pair in the two graphs guaranteed to be incurred under any bijection.

Formally, for connected graphs $G, H$ of the same order, it is defined as
\begin{equation}
    \label{equation:maximum_distortion_distance}
    \Ddist{(\infty)}(G, H) \ceq \min_{\pi \in \Bij(V(G), V(H))} \max_{v, v' \in V(G)} 
    \big|D_G(v, v') - D_H(\pi(v), \pi(v'))\big|.
\end{equation}
We note the similarity to the \emph{Gromov--Hausdorff} distance \cite{gromov_original_paper,memoli_properties_of_gromov_hausdorff}, which, for two arbitrary metric spaces $\X = (X, d_\X),\, \Y=(Y, d_\Y)$, is defined as
\begin{equation}
    \label{equation:gromov_hausdorff_distance}
    \Mdist{\GH}(\X, \Y) \ceq \mfrac{1}{2} \cdot \inf_{R \in \Corr(X,Y)} \sup_{(x, y), (x', y') \in R} 
    \big|d_\X(x, x') - d_\Y(y, y')\big|
\end{equation}
where $\Corr(X, Y)$ denotes the set of correspondences between $X$ and $Y$.
A \emph{correspondence} between two sets $X, Y$ is a relation $R \subseteq X \times Y$ such that for each $x \in X$, there is a $y\in Y$ with $(x, y) \in R$, and conversely, for every $y \in Y$, there is an $x\in X$ with $(x, y) \in R$. In particular, every bijection is a correspondence.
Thus, up to the factor $\frac{1}{2}$, \cref{equation:maximum_distortion_distance} is a variant of the Gromov--Hausdorff distance for unweighted finite graphs of the same order (we interpret each graph as a finite metric space where the metric is the shortest-path distance) in which the optimisation ranges only over bijections.
In contrast, \cref{equation:gromov_hausdorff_distance} is well-defined also for graphs of different orders (or more generally for metric spaces of arbitrary cardinality).

Extending the notation introduced in \cref{section:preliminaries}, for graphs $G, H$, and a bijection $\pi \in \Bij\big(V(G), V(H)\big)$,
the objective function of \cref{equation:maximum_distortion_distance} can be rewritten as
\begin{equation}
    \label{equation:max_distortion_objective}
    \dis_{(\infty)}(\pi) \ceq \max_{v, v' \in V(G)} \dis(\pi, v, v').
\end{equation}

In the following, we prove the computational hardness of the maximum distortion distance
and, as a byproduct, also obtain a similar inapproximability result for the Gromov--Hausdorff distance.
Our proofs are based on previous hardness results for Gromov--Hausdorff \cite{agarwal_et_al_GH_NP_hard_metric_trees,Schmiedl2017}.
For the definition of pathwidth, we refer to \cref{definition:tree_and_path_decomposition}.

\thmMaxDistortionHardness*

\begin{proof}[Proof Sketch]
We only sketch the case of pathwidth two here.
The proof is by reduction from \threepartition{} \cite{garey1979computers}, which asks whether $3m$ positive integers $a_1, \ldots, a_{3m}$ can be partitioned into $m$ triples, each summing to a given target value $A$.
Roughly speaking, given an instance $\mathcal I = (A, a_1, \ldots, a_{3m})$, we construct two trees $G$ and $H$
of the same order
encoding the integer values of the instance by suitable leaf counts (compare \cref{figure:trees_pathwidth_two}).
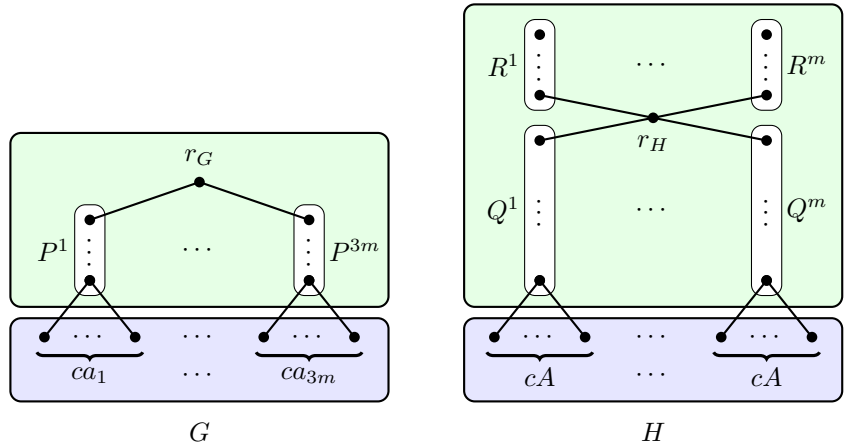
\begin{figure}[b!]
    \captionsetup{singlelinecheck = true, justification=justified}
    \centering

    \begin{tikzpicture}[
            box/.style={
                draw,
                fill=white,
                rounded corners,
                minimum width=1.6cm,
                minimum height=0.8cm,
                thick
            },
            dot/.style={
                draw,
                fill=black,
            }
        ]
    
        \newcommand{\branchleaves}[4]{%
            \coordinate (#1root) at (#2, #3);
            \coordinate (#1nstart) at (#2-0.6,#3-0.75);
            \coordinate (#1nend)   at (#2+0.6,#3-0.75);
    
            \fill (#1nstart) circle (2pt);
            \fill (#1nend)   circle (2pt);
        
            \coordinate (#1mid) at ($(#1nstart)!0.5!(#1nend)$);
            \node at (#1mid) {$\cdots$};
        
            \draw[thick] (#1root) -- (#1nstart);
            \draw[thick] (#1root) -- (#1nend);
        
            \draw[decorate,decoration={brace},very thick]
                ($(#1nend)+(0.1,-0.2)$)
                --
                ($(#1nstart)+(-0.1,-0.2)$)
                node[midway,below=2pt] {#4};
        }
    
        \newcommand{\kpath}[6][right]{%
    
            \begin{scope}[shift={(#3,#4)}]
        
                \coordinate (#2start) at (0,0);
                \coordinate (#2end) at (0,-#5);
        
                \def\pad{0.2}
        
                \draw[rounded corners, fill=white]
                    (-\pad,\pad) rectangle (\pad,-#5-\pad);
                    
                \fill (#2start) circle (2pt);
                \fill (#2end)   circle (2pt);
        
                \node[rotate=90] (mid) at ($(#2start)!0.5!(#2end)$) {$\cdots$};
                \node[#1=4pt] at (mid) {#6};
        
            \end{scope}
        }
    
        \coordinate (G) at (0,-3);
        \coordinate (H) at (6,-0.8); 

        \begin{scope}[shift=(G)]
            \def\gshift{0.05}                 
            \pgfmathsetmacro\gL{\gshift}      
            \pgfmathsetmacro\gR{3-\gshift}    

            \node[box, minimum width=5cm, minimum height=2.3cm, fill=green!10] at (1.5, -1.0) {};
            \node[box, minimum width=5cm, minimum height=1.1cm, fill=blue!10] at (1.5, -2.85) {};
    
            \node at (1.5, -3.8) {$G$};
        
            \coordinate (rG) at (1.5,-0.5);
            \fill (rG) circle (2pt);
            \node[above=3pt] at (rG) {$r_G$};
            \kpath[left]{a1}{\gL}{-1}{0.8}{$P^1$};
            \kpath[right]{a3m}{\gR}{-1}{0.8}{$P^{3m}$};
            \node at (1.5,-1.4) {$\cdots$};
            \draw[thick] (\gL,-1) -- (rG);
            \draw[thick] (\gR,-1) -- (rG);
            \branchleaves{a1}{\gL}{-1.8}{$ca_1$};
            \node at (1.5,-2.55) {$\cdots$};
            \node at (1.5,-3.05) {$\cdots$};
            \branchleaves{a3m}{\gR}{-1.8}{$ca_{3m}$};
            \fill (\gL, -1.8) circle (2pt);
            \fill (\gR, -1.8) circle (2pt);
        \end{scope}
    
        \begin{scope}[shift=(H)]
            \node[box, minimum width=5cm, minimum height=4cm, fill=green!10] at (1.5, -2.35) {};
            \node[box, minimum width=5cm, minimum height=1.1cm, fill=blue!10] at (1.5, -5.05) {};
    
            \node at (1.5, -6) {$H$};
    
            \coordinate (rH) at (1.5,-1.85);
            \fill (rH) circle (2pt);
            \node[below=3pt] at (rH) {$r_H$};
        
            \kpath[left]{a11}{0}{-0.75}{0.8}{$R^1$};
            \kpath[left]{a12}{0}{-2.15}{1.85}{$Q^1$};
            
            \kpath[right]{a11}{3}{-0.75}{0.8}{$R^m$};
            \kpath[right]{a12}{3}{-2.15}{1.85}{$Q^m$};
            
            \node at (1.5,-3.075) {$\cdots$};
            \node at (1.5,-1.15) {$\cdots$};
    
            \branchleaves{A1}{0}{-4}{$cA$};
            \branchleaves{Am}{3}{-4}{$cA$};
            \node at (1.5,-4.75) {$\cdots$};
            \node at (1.5,-5.25) {$\cdots$};
            \fill (0, -4) circle (2pt);
            \fill (3, -4) circle (2pt);
            
            \draw[thick] 
                (rH) -- (0,-1.55)
                (rH) -- (3,-1.55)
                (rH) -- (0,-2.15)
                (rH) -- (3,-2.15)
            ;
        \end{scope}
        
    \end{tikzpicture}

    \caption{The trees of pathwidth two constructed in the reduction of \cref{theorem:max_distortion_hardness}. White boxes correspond to paths of $\ell$ vertices ($P, R$) and paths of $2\ell$ vertices ($Q$). If $\mathcal I$ is a YES-instance, there is a bijection that preserves the coloured regions and maps each leaf set of $G$ into a single leaf set of $H$. If $\mathcal I$ is a NO-instance, for every bijection, a leaf set corresponding to some $a_i$ in $G$ is split up and, by choice of $c \ceq 3m\ell + 2$, its vertices end up in at least two distinct leaf sets of $H$.}
    \label{figure:trees_pathwidth_two}
\end{figure}
$G$ has $3m$ subtrees $P^1, \ldots, P^{3m}$, with their corresponding leaf sets encoding the values $a_1, \ldots, a_{3m}$. $H$ has $m$ subtrees $Q^1, \ldots, Q^m$, each with a leaf set for the value $A$, and $m$ pendant paths $R^1, \ldots, R^m$. Since \threepartition{} is strongly \NP-complete, the number of vertices in the graphs is polynomial.
If $\mathcal I$ is a YES-instance, a corresponding partition of the $a_i$ into triples that each sum to $A$ can be turned into a bijection mapping the leaf sets of three subtrees of $G$ into a single leaf set of $H$. 
For NO-instances, the size imbalance between the number of leaves and interior vertices (controlled by the choice of $c$) in $G$ and $H$ forces at least one leaf set of $G$ to be split across multiple leaf sets of $H$.
This incurs roughly twice the distortion of a YES-instance.
\end{proof}

\cref{theorem:max_distortion_hardness} can be seen as an improvement of a result by Schmiedl \cite[Theorem 3.7]{Schmiedl2017} which relies on more general metric spaces
(specifically they use \emph{subsets} of weighted tree metrics) and thus does not carry over to our setting of unweighted graphs. 
Schmiedl transfers this hardness to $\Mdist{\GH}$ by suitably scaling the constructed 
metrics \cite[Corollary 3.8]{Schmiedl2017}. Intuitively, for unweighted graphs scaling roughly corresponds
to subdividing edges. Using this idea to slightly modify the construction behind \cref{theorem:max_distortion_hardness},
we limit the advantage correspondences might have over bijections and are thereby able to transfer the hardness of $\Ddist{(\infty)}$ to $\Mdist{\GH}$ in the setting of unweighted graphs. Specifically,
for YES-instances there always exists a bijection with lower maximum distortion than any correspondence achieves in NO-instances.

\begin{theorem}
    \label{theorem:gh_inapprox_result_julius}
    For any $\eps > 0$, unless $\P = \NP$, there is no polynomial-time approximation algorithm for $\GSim{\Mdist{\GH}}$\footnote{Slightly abusing notation, we denote by $\GSim{\Mdist{\GH}}$ the problem of computing the Gromov--Hausdorff distance between two unweighted graphs given as metric spaces.}
    with additive error guarantee $\O(n^{1-\eps})$, even if both input
    graphs are trees of pathwidth three on 
    the same number of vertices.
\end{theorem}

\noindent Our result matches the additive approximation hardness implicit in
the multiplicative hardness proof of Agarwal \etal{} \cite[Theorem 3.3]{agarwal_et_al_GH_NP_hard_metric_trees}.
However, complementing their result, our 
construction yields \emph{unweighted} trees of \emph{equal order} for which approximation hardness is already witnessed by 
\emph{bijections}. 
Intuitively, this shows that the hardness does not inherently rely on the additional flexibility introduced by correspondences in the optimisation domain.

Finally, we note that $\Ddist{(\infty)}$ and $\Mdist{\GH}$ admit approximation with additive error at most $\eps n$ for every $\eps > 0$, based on coarsening the metric structure of the input graphs 
(see the proof of \cref{theorem:additive_approx_global_distortion}).
Hence, the above additive inapproximability results are essentially optimal.

\section{Maximum Distortion into the Path}
\label{section:max_distortion_path}
The hardness results of the previous section gave inapproximability even for very \enquote{path-like} trees, \ie{} trees of pathwidth two.
In this section, we show that when restricting one input graph even further, that is, fixing it to be a path graph, computing the maximum distortion becomes tractable. This holds even if the other graph is an arbitrary connected graph of the same order.

\thmMaxDistortionIntoPath*

This result is somewhat striking since maximum distortion embedding into the path is closely related to
\bandwidth{}.
In the latter,
the vertices of a graph $G$ are assigned distinct labels from ${1,\ldots,|V(G)|}$ so as to minimise the maximum difference between the labels of adjacent vertices. 
Equivalently, this can be viewed as a bijective embedding of 
$G$ into a path that minimises the maximum stretch of any edge.
Hence, the only difference between these two problems is over which pairs the maximum is taken.
Despite this close similarity, \bandwidth{} is computationally much harder: It is \NP-hard on trees of maximum degree three \cite{garey_graham_bandwidth_degree_3} and on caterpillars of bounded hairlength \cite{monien_bandwidth_hairlength_three}. Moreover, it does not admit a constant-factor approximation on caterpillars with unbounded hairlength \cite{dubey_uriel_unger_hardness_of_approximating_bandwidth,unger_caterpillar_bandwidth_inapprox_extended_abstract}.

The proof of \cref{theorem:max_distortion_graph_into_path} yields a simple structural characterisation of an optimal bijection whenever one input graph is a path. Given a diametral pair $\{u,v\}$ of $G$, we partition $V(G)$ according to the signed difference of the distances to $u$ and $v$. Then, any enumeration respecting the order of these sets, in either direction, yields a bijection into the path with optimal maximum distortion.
A formal description is provided in \cref{alg:graph_to_path}. 
This algorithm not only solves the decision problem but also computes an optimal alignment; note that this does not follow from a decision algorithm via the standard self-reduction, which does not obviously apply to $\Ddist{(\infty)}$ (see \cref{remark:returning_bijection}).

\subparagraph*{Overview of the Proof}
To prove the distortion guarantee of \cref{theorem:max_distortion_graph_into_path}, we need to show two directions: No bijection yields a smaller maximum distortion and, conversely, the bijection returned by \cref{alg:graph_to_path} attains equality.
The lower bound follows from the fact that the endpoints of $P_n$ are at distance $n-1$ whereas their preimages under any bijection can be at distance at most $\diam(G)$.
For the second statement, we verify that the bijection $\pi^\ast$ returned by \cref{alg:graph_to_path} satisfies $\dis_{(\infty)}(\pi^\ast) = n-1-\diam(G)$, which is shown in two separate lemmas. 
For notational convenience, for a fixed graph $G$ we let $d \ceq \diam(G)$ and let $C_{-d}, \ldots, C_d$ be the sets (in the following referred to as bags) computed by \cref{alg:graph_to_path}. 

\begin{lemma}
    \label{lemma:graph_to_path_no_large_increase}
    For all $x, x' \in V(G)$, 
    $\Gdist_{P_n}(\pi^\ast(x), \pi^\ast(x')) - \Gdist_G(x, x') \leq n-1-d.$
\end{lemma}
To show that the distance between any two vertices $x, x'$ does not increase too much under $\pi^\ast$, we find $d+1$ vertices that are either not mapped into the (unique) shortest $\pi^\ast(x)\pi^\ast(x')$-path or also lie on a shortest $xx'$-path in $G$.
Intuitively, these vertices cannot contribute to an increase in distance between $x$ and $x'$.
For the first kind, assuming $x \in C_b$ and $x' \in C_{b'}$ with $b \leq b'$ (exchanging $x$ and $x'$ if necessary), we count the non-empty bags $C_a$ with $a < b$ or $b' < a$.
The second kind is found by analysing $xx'$-separators that only consist of vertices from bags between $C_b$ and $C_{b'}$.

\begin{lemma}
    \label{lemma:graph_to_path_no_large_decrease}
    For all $x, x' \in V(G)$,  
    $\Gdist_G(x, x') - \Gdist_{P_n}(\pi^\ast(x), \pi^\ast(x')) \leq  n-1-d.$
\end{lemma}
For \cref{lemma:graph_to_path_no_large_decrease}, the argument is a bit more nuanced. If the distance between 
$x$ and $x'$ in $G$ is large, then we show that either $x$ and $x'$ lie in bags whose indices are sufficiently far apart, or there is a path $S$
that is disjoint from both a shortest $xx'$-path $R$ and a shortest $uv$-path $Q$.
We analyse the interaction of $S$ with the vertices of $R$ and $Q$ (an overview of the resulting case distinction is provided in \cref{fig:case_distinction_cases}) and, roughly speaking, prove that $S$ either allows \enquote{shortcuts} between $x$ and $x'$ (yielding a contradiction) or is sufficiently long.
In the latter case, we bound the length of $R$ with a counting argument based on the vertex counts of the aforementioned paths and hence show that $\Gdist_G(x, x')$ is sufficiently small.

\subparagraph*{Runtime Analysis} The distance matrix $D_G$ of an (undirected, unweighted) $n$-vertex
graph $G$ can be computed in time $\mathrm{O}(n^\omega \log n)$ \cite[Theorem 1]{seidel_apsp_undirected_unweighted_matrix_multiplication}, where
$\omega < 2.371177$ is the best known matrix multiplication exponent at the time of writing \cite{dupont2026improvingmatrixmultiplicationexponent}. 
Given $D_G$, a diametral pair $u,v$ is found in time $\mathrm{O}(n^2)$; the loop in \cref{alg:graph_to_path} and the bijection take linear time. 
Overall, the runtime is $\O(n^\omega \log n)$. 
As only the distances from $u$ and $v$ are needed, for sparse graphs with $m$ edges it can be faster to find $u,v$ via $n$ BFS runs and then run two BFS from $u$ and $v$, in total time $\O(nm)$.

\begin{algorithm}[!htbp]
    \vspace{2pt}
    \textbf{Input}: A connected graph $G$ on $n$ vertices\\
    \textbf{Output}: $\pi \in \Bij(V(G), V(P_n))$ with 
    $\dis_{(\infty)}(\pi) = n-1-\diam(G)$
    \begin{algorithmic}[1]
        \State $D_G \gets \text{distance matrix of $G$}$
        \State $\{u, v\} \gets \arg\max_{\{x, y\} \subseteq V(G)} D_G(x, y)$ \Comment{$\{u, v\}$ is a diametral pair in $G$}
        \State $d \gets D_G(u, v)$ \Comment{$d = \diam(G)$}
        \State Initialise $C_{-d}, \ldots, C_d$ as empty sets.
        \For{$w \in V(G)$}
            \State $i \gets D_G(u, w)$ \Comment{$w$ at distance $i$ from $u$}
            \State $j \gets D_G(v, w)$ \Comment{$w$ at distance $j$ from $v$}
            \State $C_{i-j} \gets C_{i-j} \cup \{w\}$ \Comment{add $w$ to $C_{i-j}$}
        \EndFor{}
        \State Enumerate $V(G)$ as $v_1,\ldots,v_n$ such that, for $k<\ell$, if $v_k\in C_b$ and $v_\ell\in C_{b'}$, then $b\leq b'$.
        \hspace*{\fill}\Comment{sort vertices according to their bag indices in non-decreasing order}
        \State \Return $\pi: V(G) \to [n],\ v_k \mapsto k$ \Comment{assume $V(P_n) = \{1, \ldots, n\}$}
    \end{algorithmic}
    
    \caption{(Embedding a Graph into a Path)\textbf{.}}
    \label{alg:graph_to_path}
\end{algorithm}
\section{Maximum Distortion on Graphs of Bounded Pathlength} 
\label{section:max_distortion_bounded_pathlength}
As shown in \cref{section:hardness_max_distortion}, computing $\Ddist{(\infty)}$, even approximately 
below a constant multiplicative error guarantee of two, is already hard if both input graphs have 
pathwidth two. In contrast, if one graph is a path, then exact computation becomes polynomial-time 
solvable (\cref{theorem:max_distortion_graph_into_path}).
In the gap between these results lies the class of connected graphs of pathwidth one --
the caterpillars. A \emph{caterpillar} is a tree whose vertices of degree larger than one form a path.
In the literature, these are sometimes referred to as caterpillars of \emph{hairlength one}.
Despite their seemingly simple structure, we did not find a polynomial-time algorithm for computing 
$\Ddist{(\infty)}$ exactly on caterpillars. 
Unfortunately, the restrictive structure of caterpillars does not 
readily admit a hardness reduction. We therefore leave the exact complexity open and instead present a
simple approximation algorithm with a combined approximation guarantee.

\begin{lemma}
    \label{lemma:max_distortion_caterpillar_approx}
    There is a $\left(\frac{7}{4}, \frac{11}{2}\right)$-approximation
    algorithm for $\GSim{\Ddist{(\infty)}}$ if both input graphs are caterpillars. 
    This algorithm runs in linear time in the order of the input graphs.
\end{lemma}

Intuitively, we associate with each caterpillar a multiset over $\R$ (extending the notion of maximum
distortion in the natural way), which we then translate into a set in order to apply a result of Hall and 
Papadimitriou \cite[Theorem 5]{approximating_the_distortion}.
Our refined analysis improves their guarantee from a factor 2 to factor $\frac{7}{4}$.

\begin{theorem}
    \label{proposition:2_approx_real_point_sets}
    Given finite sets $X = \{x_1, \ldots, x_n\},\, Y = \{y_1, \ldots, y_n\} \subseteq \R$ 
    with enumerations in increasing order, the better of the two bijections 
    $\pi(x_i) = y_i$ and $\pi(x_i) = y_{n+1-i}$ for $i \in [n]$ between $X$ and $Y$ has a maximum 
    additive distortion of at most $\frac{7}{4}$ times the optimum.
\end{theorem}

For a caterpillar $T$, the \emph{spine} of $T$ is the path $v_1, \ldots, v_k$ that remains after 
removing all its leaves. Contracting the edges between spine vertices and leaves in $T$ yields a 
collection of bags of vertices arranged along the spine which we view as a multiset with support $\{1, \ldots, k\}$.
This simplifies the metric structure and allows us to apply \cref{proposition:2_approx_real_point_sets}
after slightly shifting the points (to obtain sets) in a way that preserves the guarantees of \cref{proposition:2_approx_real_point_sets}, compare \cref{corr:multiset_to_set}.
The mapping found by the algorithm of \cref{lemma:max_distortion_caterpillar_approx} is either \emph{monotone} or \emph{antitone} (formal definitions are in 
\cref{appendix:max_distortion_bounded_pathlength}) in the sense that the order along the spine vertices 
is either preserved or reversed by the mapping.
Moreover, this mapping as well as the resulting maximum distortion can be computed in linear time using the monotonicity.

We note that the mapping between caterpillars found in this way is, in general, not optimal.
The following example, for simplicity described in terms of multisets, gives an infinite family with a multiplicative gap of $\frac{3}{2}$.
Translated back to caterpillars, this yields a lower bound of $\frac{3}{2} - \eps$ (for every $\eps > 0$) for our algorithmic approach.
Furthermore, \cref{example:caterpillar-74-example} shows that the factor $\frac{7}{4}$ in \cref{proposition:2_approx_real_point_sets} is optimal.

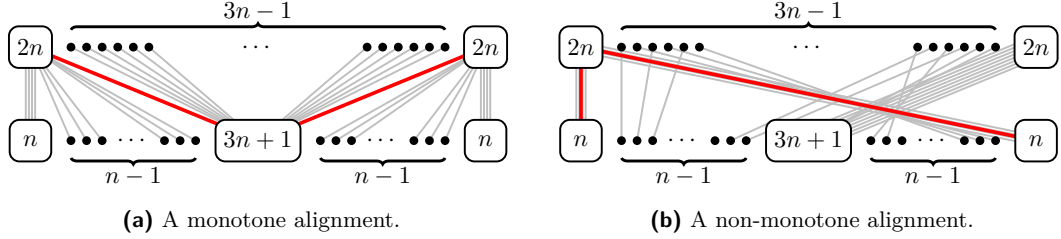
\begin{figure}[!htbp]
    \captionsetup{singlelinecheck = true, justification=justified}
    \centering
    \begin{subfigure}[t]{0.48\textwidth}
        \centering
        \begin{adjustbox}{width=\linewidth}
            \begin{tikzpicture}[
    scale=0.9,
    box/.style={
        draw,
        fill=white,
        rounded corners,
        minimum width=0.6cm,
        minimum height=0.6cm,
        thick
    }
]

    \coordinate (plotoffset) at (9, 0);

    \coordinate (c0) at (0.35, 0);
    \coordinate (cend) at (7.65, 0);
    \foreach \i in {1,...,6} {
        \pgfmathsetmacro{\px}{1+(\i-1)*0.25}
        \coordinate (c\i) at (\px, 0);
    }
    \foreach \i in {7,...,12} {
        \pgfmathsetmacro{\px}{7-(11-(\i-1))*0.25}
        \coordinate (c\i) at (\px, 0);
    }

    \pgfmathsetmacro{\oy}{-1.5}
    \coordinate (d0) at (0.35, \oy);
    \coordinate (dmid) at (4, \oy);
    \coordinate (dend) at (7.65, \oy);
    \foreach \i in {1,...,3} {
        \pgfmathsetmacro{\px}{1+(\i-1)*0.25}
        \coordinate (d\i) at (\px, \oy);
    }
    \foreach \i in {4,...,6} {
        \pgfmathsetmacro{\px}{3-(5-(\i-1))*0.25}
        \coordinate (d\i) at (\px, \oy);
    }
    \foreach \i in {7,...,9} {
        \pgfmathsetmacro{\px}{5+(\i-7)*0.25}
        \coordinate (d\i) at (\px, \oy);
    }
    \foreach \i in {10,...,12} {
        \pgfmathsetmacro{\px}{7-(11-(\i-1))*0.25}
        \coordinate (d\i) at (\px, \oy);
    }


    \foreach \i/\j/\k in {
        0/0/4,
        0/1/1,
        0/2/1,
        0/3/1,
        0/4/1,
        0/5/1,
        0/6/1,
        0/mid/1,
        1/mid/1,
        2/mid/1,
        3/mid/1,
        4/mid/1,
        5/mid/1,
        6/mid/1,
        7/mid/1,
        8/mid/1,
        9/mid/1,
        10/mid/1,
        11/mid/1,
        12/mid/1,
        end/mid/1,
        end/7/1,
        end/8/1,
        end/9/1,
        end/10/1,
        end/11/1,
        end/12/1,
        end/end/4
    } {
        \path (c\i) -- (d\j);

        \pgfgetlastxy{\ax}{\ay}
        \pgfmathsetmacro{\dx}{\ax}
        \pgfmathsetmacro{\dy}{\ay}

        \path (c\i);
        \pgfgetlastxy{\cx}{\cy}
        \path (d\j);
        \pgfgetlastxy{\dx}{\dy}

        \pgfmathsetmacro{\vx}{\dx-\cx}
        \pgfmathsetmacro{\vy}{\dy-\cy}
        \pgfmathsetmacro{\len}{veclen(\vx,\vy)}

        \pgfmathsetmacro{\ux}{-\vy/\len}
        \pgfmathsetmacro{\uy}{\vx/\len}

        \foreach \t in {1,...,\k} {
            \pgfmathsetmacro{\off}{(\t-(\k+1)/2)*0.05}

            \draw[thick, gray!50]
                ($ (c\i) + (\off*\ux,\off*\uy) $)
                --
                ($ (d\j) + (\off*\ux,\off*\uy) $);
        }
    }

    \draw[line width=1.5pt, red] (c0) -- (4, -1.5);
    \draw[line width=1.5pt, red] (cend) -- (4, -1.5);

    \begin{scope}
        \node[box] (a0) at (c0) {$2n$};
        \foreach \i in {1,...,12} {
            \fill (c\i) circle (2pt);
        }
        \node at (4, 0) {$\cdots$};
        \node[box] (ar) at (cend) {$2n$};

        \draw [
            very thick,
            decoration={
                brace,
                raise=0.25cm
            },
            decorate
        ] (1, 0) -- (7, 0)
        node [pos=0.5,anchor=south,yshift=0.25cm] {$3n-1$};
    \end{scope}

    \begin{scope}
        \node[box] (b0) at (d0) {$n$};
        \foreach \i in {1,...,12} {
            \fill (d\i) circle (2pt);
        }
        \node at (2, \oy) {$\cdots$};

        \draw [
            very thick,
            decoration={
                brace,
                mirror,
                raise=0.25cm
            },
            decorate
        ] (1, \oy) -- (3, \oy)
        node [pos=0.5,anchor=north,yshift=-0.25cm] {$n-1$};

        \node[box] (bmiddle) at (4, \oy) {$3n+1$};
        \node at (6, \oy) {$\cdots$};

        \draw [
            very thick,
            decoration={
                brace,
                mirror,
                raise=0.25cm
            },
            decorate
        ] (5, \oy) -- (7, \oy)
        node [pos=0.5,anchor=north,yshift=-0.25cm] {$n-1$};

        \node[box] (bs) at (dend) {$n$};
    \end{scope}

\end{tikzpicture}
        \end{adjustbox}
        \caption{A monotone alignment.}
        \label{subfig:monotone_32}
    \end{subfigure}
    \hfill
    \begin{subfigure}[t]{0.48\textwidth}
        \centering
        \begin{adjustbox}{width=\linewidth}
            \begin{tikzpicture}[
    scale=0.9,
    box/.style={
        draw,
        fill=white,
        rounded corners,
        minimum width=0.6cm,
        minimum height=0.6cm,
        thick
    }
]

    \coordinate (c0) at (0.35, 0);
    \coordinate (cend) at (7.65, 0);
    \foreach \i in {1,...,6} {
        \pgfmathsetmacro{\px}{1+(\i-1)*0.25}
        \coordinate (c\i) at (\px, 0);
    }
    \foreach \i in {7,...,12} {
        \pgfmathsetmacro{\px}{7-(11-(\i-1))*0.25}
        \coordinate (c\i) at (\px, 0);
    }

    \pgfmathsetmacro{\oy}{-1.5}
    \coordinate (d0) at (0.35, \oy);
    \coordinate (dmid) at (4, \oy);
    \coordinate (dend) at (7.65, \oy);
    \foreach \i in {1,...,3} {
        \pgfmathsetmacro{\px}{1+(\i-1)*0.25}
        \coordinate (d\i) at (\px, \oy);
    }
    \foreach \i in {4,...,6} {
        \pgfmathsetmacro{\px}{3-(5-(\i-1))*0.25}
        \coordinate (d\i) at (\px, \oy);
    }
    \foreach \i in {7,...,9} {
        \pgfmathsetmacro{\px}{5+(\i-7)*0.25}
        \coordinate (d\i) at (\px, \oy);
    }
    \foreach \i in {10,...,12} {
        \pgfmathsetmacro{\px}{7-(11-(\i-1))*0.25}
        \coordinate (d\i) at (\px, \oy);
    }

    \newcommand{\drawbundle}[4]{%
        \path #1;
        \pgfgetlastxy{\cx}{\cy}

        \path #2;
        \pgfgetlastxy{\dx}{\dy}

        \pgfmathsetmacro{\vx}{\dx-\cx}
        \pgfmathsetmacro{\vy}{\dy-\cy}
        \pgfmathsetmacro{\len}{veclen(\vx,\vy)}

        \pgfmathsetmacro{\ux}{-\vy/\len}
        \pgfmathsetmacro{\uy}{\vx/\len}

        \foreach \t in {1,...,#3} {
            \pgfmathsetmacro{\off}{(\t-(#3+1)/2)*0.05}

            \draw[thick, #4]
                ($ #1 + (\off*\ux,\off*\uy) $)
                --
                ($ #2 + (\off*\ux,\off*\uy) $);
        }
    }

    \drawbundle{(c0)}{(d0)}{4}{gray!50}
    \drawbundle{(c0)}{(dend)}{4}{gray!50}
    \drawbundle{(cend)}{(dmid)}{8}{gray!50}

    \drawbundle{(c1)}{(d1)}{1}{gray!50}
    \drawbundle{(c2)}{(d12)}{1}{gray!50}
    \drawbundle{(c3)}{(d2)}{1}{gray!50}
    \drawbundle{(c4)}{(d11)}{1}{gray!50}
    \drawbundle{(c5)}{(d3)}{1}{gray!50}
    \drawbundle{(c6)}{(d10)}{1}{gray!50}
    \drawbundle{(c12)}{(dmid)}{1}{gray!50}
    \drawbundle{(c11)}{(dmid)}{1}{gray!50}
    \drawbundle{(c10)}{(d6)}{1}{gray!50}
    \drawbundle{(c9)}{(d7)}{1}{gray!50}
    \drawbundle{(c8)}{(d5)}{1}{gray!50}
    \drawbundle{(c7)}{(d8)}{1}{gray!50}

    \draw[line width=1.5pt, red] (c0) -- (d0);
    \draw[line width=1.5pt, red] (c0) -- (dend);

    \begin{scope}
        \node[box] (a0) at (c0) {$2n$};
        \foreach \i in {1,...,12} {
            \fill (c\i) circle (2pt);
        }
        \node at (4, 0) {$\cdots$};
        \node[box] (ar) at (cend) {$2n$};

        \draw [
            very thick,
            decoration={
                brace,
                raise=0.25cm
            },
            decorate
        ] (1, 0) -- (7, 0)
        node [pos=0.5,anchor=south,yshift=0.25cm] {$3n-1$};
    \end{scope}

    \begin{scope}
        \node[box] (b0) at (d0) {$n$};
        \foreach \i in {1,...,12} {
            \fill (d\i) circle (2pt);
        }
        \node at (2, \oy) {$\cdots$};

        \draw [
            very thick,
            decoration={
                brace,
                mirror,
                raise=0.25cm
            },
            decorate
        ] (1, \oy) -- (3, \oy)
        node [pos=0.5,anchor=north,yshift=-0.25cm] {$n-1$};

        \node[box] (bmiddle) at (4, \oy) {$3n+1$};
        \node at (6, \oy) {$\cdots$};

        \draw [
            very thick,
            decoration={
                brace,
                mirror,
                raise=0.25cm
            },
            decorate
        ] (5, \oy) -- (7, \oy)
        node [pos=0.5,anchor=north,yshift=-0.25cm] {$n-1$};

        \node[box] (bs) at (dend) {$n$};
    \end{scope}

\end{tikzpicture}
        \end{adjustbox}
        \caption{A non-monotone alignment.}
        \label{subfig:non_monotone_32}
    \end{subfigure}
    
    \caption{Two alignments between the same two multisets. 
    Each multiset is arranged horizontally    
    with its bags (black dots and boxes) placed at consecutive integer positions, \ie{} the top multiset has support $[3n+1]$, the bottom one $[2n+1]$.
    Black dots denote a single point (bags of size one) while a box labelled $k$ denotes a bag of $k$ points, all at the same position. Both multisets have $7n-1$ elements.
    Alignments are visualised by grey edges; an edge connects two bags $B, B'$ if there is an $x \in B$ with $\pi(x) \in B'$. 
    Red edges indicate pairs of maximum distortion under the respective alignment.
    }
    \label{fig:non_monotone_caterpillar_space}
\end{figure}
\begin{example}
    \label{example:caterpillar_spaces_32}
    For $n \geq 2$, consider the multisets depicted in \cref{fig:non_monotone_caterpillar_space}.
    The monotone alignment (\cref{subfig:monotone_32}) incurs a distortion of $3n$ since the leftmost and rightmost bags of the top multiset are at distance $3n$ and both contain a point that is mapped into the bag of size $3n+1$ in the bottom multiset. 
    By symmetry of the two multisets, the antitone alignment yields the same result.
    In contrast, the non-monotone alignment (\cref{subfig:non_monotone_32}) splits up the leftmost bag into the left- and rightmost bags yielding a distortion of only $2n$.
    The remaining bags are processed left-to-right, assigning each point of the top space to a point not yet assigned in the left- or rightmost possible bags, alternating between the two choices.
    Any two points in the top space at distance $2n + k$ have images in the bottom space at distance at least $k$. 
    Thus, the maximum distortion is at most (and in fact exactly) $2n$.
\end{example}

\cref{lemma:max_distortion_caterpillar_approx} can easily be extended to 
connected graphs of bounded pathlength.
The \emph{length} of a path decomposition of a graph $G$ is the largest diameter of its bags measured in the shortest-path metric of $G$.
The \emph{pathlength} of $G$ is the minimum length taken over all path decompositions of $G$. 
Path decompositions are formally defined in \cref{definition:tree_and_path_decomposition}.

\newcounter{savedtheorem}
\setcounter{savedtheorem}{\value{theorem}}
\setcounter{theorem}{\getrefnumber{theorem:pathlength_approx}-1}
\begin{theorem}[rephrased] 
    \label{theorem:max_distortion_approx_2_4}
    (1) There is a $\big (\frac{7}{4}, \frac{11}{2}k + 11 \big )$-approximation algorithm for $\GSim{\Ddist{(\infty)}}$
    if both input graphs have pathlength at most $k$, which runs in time $\O(n^3)$. \\
    (2) If only one input graph has pathlength at most $k$, there is a $(5, 12k+12)$-approximation 
    algorithm for $\GSim{\Ddist{(\infty)}}$ running in time $\O(n^3)$.
\end{theorem}
\setcounter{theorem}{\value{savedtheorem}}

\noindent Here we use that graphs of bounded pathlength have embeddings of small maximum distortion into 
suitable caterpillars \cite[Lemma 2]{dragan2025graphparameterscoarselyequivalent_pathlength}. Such 
caterpillars can be found in cubic time \cite{Dragan2017_pathlength_approx, dragan2025graphparameterscoarselyequivalent_pathlength}. 
For (2), we rely on the algorithm \textsc{Intervals} of 
Hall and Papadimitriou \cite[Section 3.2]{approximating_the_distortion}.

\section{Beyond Maximum Distortion}
\label{section:beyond_max_distortion}
So far, we have only considered the maximum distortion distance, corresponding to \cref{equation:distortion_distance_template} with 
entrywise $\infty$-norm.
In this section we investigate local and global distortion variants obtained from \cref{equation:distortion_distance_template} with
$\norm{\cdot}_p$ and $\norm{\cdot}_{(p)}$.
This corresponds to replacing the objective function $\Delta_{(\infty)}$ by $\Delta_p$ and $\Delta_{(p)}$, respectively.
For the operator $p$-norms with $1 \leq p \leq \infty$, the objective function is defined as
$\dis_{p}(\pi) \ceq \norm{D_G^\pi - D_H}_p$.
Since $\norm{\cdot}_1$ is the maximum absolute column sum, for $p=1$
the objective simplifies to
\[
    \dis_{1}(\pi) = \max_{v \in V(G)} \sum_{v' \in V(G)} \dis(\pi, v, v'),
\]
that is, the maximum cumulative distortion incurred \enquote{locally} by a single vertex (\ie{} the sum of its distortions to all other vertices).
Instead, we may also consider cumulative distortion of all vertices \enquote{globally} (or more precisely all ordered vertex pairs), 
which corresponds to the entrywise $1$-norm of $D_G^\pi - D_H$. 
More generally, for $1 \leq p < \infty$, we obtain
\begin{equation}
    \label{equation:global_distortion_def}
    \dis_{(p)}(\pi) \ceq \norm{D_G^\pi - D_H}_{(p)} =  \left(\sum_{v, v' \in V(G)} \dis(\pi, v, v')^p\right)^{1/p}.
\end{equation}
In the limit, as $p$ tends to $\infty$, this recovers the maximum distortion objective (compare \cref{equation:max_distortion_objective}).
Writing $\Ddist{p} \ceq \Ddist{\norm{\cdot}_p}$ and $\Ddist{(p)} \ceq \Ddist{\norm{\cdot}_{(p)}}$ as shorthands,
we obtain
\[
    \Ddist{p}(G, H) = \min_{\pi \in \Bij(V(G), V(H))} \dis_p(\pi), 
    \qquad
    \Ddist{(p)}(G, H) = \min_{\pi \in \Bij(V(G), V(H))} \dis_{(p)}(\pi),
\]
for $1 \leq p \leq \infty$. In analogy to \cite{gervens2022graphsimilaritybasedmatrix}, we refer to these as \emph{local} and \emph{global distortion distances}, respectively.
We note that the graph metrics $\Ddist{(p)}$ coincide with \emph{Gromov--Monge distances} \cite{Dumont2025_existence_of_gromov_monge_maps,memoli2022distancedistributionsinverseproblems} on equal-order graphs with uniform probability measure on the vertex sets.

In the following, we investigate the computational hardness of these distances in detail.
As before, our hardness results hold for fairly restricted graph classes.
On the positive side, we can give an additive approximation algorithm 
based on a metric coarsening argument. 

\subsection{Local Distortion}
\label{subsection:local_distortion}
Many results for local distortion can be derived from previous results on the graph metrics $\Adist{p}$ of \cite{gervens2022graphsimilaritybasedmatrix} by making use of the following simple observation, 
establishing a condition under which $\Adist{\norm{\cdot}}$ and $\Ddist{\norm{\cdot}}$ coincide.

\begin{observation}
    \label{observation:diameter_two}
    Let $G, H$ be connected $n$-vertex graphs of diameter at most two. Then, for any norm $\norm{\cdot}$ and for 
    any bijection $\pi \in \Bij(V(G), V(H))$, it holds that
    \begin{equation}
        \label{equation:distance_vs_adjacency} 
        \norm{D_{G}^{\pi} - D_{H}} = \norm{A_G^\pi - A_H}.
    \end{equation}
\end{observation}
This is because in graphs of diameter at most two, 
any non-adjacent pair of distinct vertices $\{v, v'\}$ is at distance two.
As a result, each (edge-)mismatch under $\pi$ corresponds precisely to a pairwise distortion of one, \ie{} $\dis(\pi, v, v') = 1$.

To use \cref{observation:diameter_two}, we augment each graph with a so-called \enquote{centre vertex}, thereby reducing its diameter to at most two. Formally, for a graph $G$, let $\widehat{G}$ denote the graph obtained by adding a new centre vertex $c_G$ and connecting it to each $v \in V(G)$ by an edge.
Applying this technique to the proof of \cite[Theorem 16]{gervens2022graphsimilaritybasedmatrix} yields \NP-hardness for the graph metrics $\Ddist{p}$ where $1 \leq p \leq \infty$.
Here, we let $b(p) \ceq \max\big\{ \frac{1}{p}, 1-\frac{1}{p}\big\}$ with $b(\infty) \ceq 1$.

\begin{theorem}
    \label{theorem:local_dist_hardness_hamiltonian}
    For $1 \leq p \leq \infty$ and any $\eps > 0$, unless $\P = \NP$, there is no polynomial-time approximation algorithm for \GSim{\Ddist{p}} with multiplicative error guarantee $3^{b(p)} - \eps$.
\end{theorem}

\noindent The proof is by reduction from \hamcycle{} on $3$-regular graphs, which is \NP-complete \cite{hamcycle_threeregular_np_complete}. Given a $3$-regular graph $G$, 
we compute an instance of $\GSim{\Ddist{p}}$ consisting of $\widehat{C_n}$ and $\widehat{G}$.
If $G$ has a Hamiltonian cycle, we extend the natural bijection between the cycle $C_n$ and $G$ so that each vertex has an edge mismatch of exactly one.
When $G$ does not have a Hamiltonian cycle, 
under every bijection there is a vertex with at least three mismatches.
This observation is preserved for $\widehat{C_n}$ and $\widehat{G}$, and thus directly transfers into the decision gap for $\Ddist{p}(\widehat{C_n},\widehat{G})$ between YES-instances and NO-instances.

For $p=1$, we can get a similar result that holds even when the input graphs are trees, albeit with a 
smaller inapproximability factor.
\begin{theorem}
    \label{theorem:local_dist_mult_hardness_trees}
    For any $\eps > 0$, unless $\P = \NP$, there is no polynomial-time approximation algorithm for $\GSim{\Ddist{1}}$ with multiplicative error guarantee $\frac{7}{6} - \eps$ or additive error guarantee $\O(n^{1-\eps})$, even when both input graphs are trees of diameter four.
\end{theorem}
This result is proven by another reduction from \threepartition{} which employs a tree construction 
similar to the one of 
\cref{theorem:max_distortion_hardness}.
Again, we use the fact that for a NO-instance every bijection necessarily splits the leaf set corresponding to some $a_i$ in the constructed tree encoding the values $a_1, \ldots, a_{3m}$.
However, since $\dis_1$ considers the sum of distortions relative to some fixed vertex, a split in a leaf set implies that there is a vertex that is separated from at least half of the vertices in that leaf set. This observation, together with the size bounds of $\frac{A}{4} < a_i < \frac{A}{2}$ on the leaf sets for all $i \in [3m]$, yields the multiplicative gap. 

\subsection{Global Distortion}
\label{subsection:global_distortion}
Moving to global distortion, we can use the same construction as for \cref{theorem:local_dist_hardness_hamiltonian} to obtain a (weaker) \NP-hardness result for $\Ddist{(p)}$ for all $1 \leq p < \infty$.
In the YES-case, a bijection aligning the cycle optimally with the Hamiltonian cycle of $G$ yields a total of $\frac{n}{2}$ mismatched edges.
For NO-instances, any bijection not mapping $c_{C_n}$ to $c_G$ necessarily has even more edge mismatches due to degree differences.
Any other bijection also produces strictly more than $\frac{n}{2}$ edge mismatches due to the imperfect cycle alignment.
Since $\widehat{C_n}$ and $\widehat{G}$ have diameter two, \cref{observation:diameter_two} implies $\Ddist{(p)}(\widehat{C_n}, \widehat{G}) \leq n^{1/p}$ if and only if $G$ is Hamiltonian, because each mismatched edge appears twice in the difference of adjacency matrices.
By a modification of this construction which expands vertices into cliques, we 
\enquote{stack}
multiple copies of an instance for \hamcycle{} to obtain a larger additive gap.

\begin{theorem}
    \label{theorem:additive_inapprox_global_distortion}
    For $1 \leq p < \infty$ and any $\eps > 0$,  unless $\P = \NP$, there is no polynomial-time approximation algorithm for $\GSim{\Ddist{(p)}}$ with additive error guarantee $\O\big(n^{\frac{1}{p}-\eps}\big)$.
\end{theorem}

When considering multiplicative approximation, a result of Arvind \etal{} \cite[Theorem 5]{arvind_approx_graph_isomorphism} combined with our centre vertex augmentation yields that there is no \PTAS{} for \GSim{\Ddist{(p)}} for any $1 \leq p < \infty$, unless $\P = \NP$.
Their problem \textsf{Min-PGI} corresponds to $\GSim{\Adist{(1)}}$ which, after augmenting the graphs with a centre vertex (this forces centres to be mapped onto each other), is equivalent to $\GSim{\Ddist{(p)}}$ by \cref{observation:diameter_two} since all individual distortions are either zero or one.
Another way to see that no \PTAS{} exists is by a connection to $\textsc{TSP}(1, 2)$, the travelling salesperson problem restricted to distances one and two, which is hard to approximate below a fixed constant \cite{engebretsen_tsp12}. Intuitively, this problem can be encoded in $\GSim{\Ddist{(p)}}$ by fixing one input to be the wheel graph ($\widehat{C_n}$) of the corresponding order.

Having shown additive and multiplicative approximation hardness, one may hope to find tractability in restricted cases.
Analogously to \cref{section:max_distortion_path}, we may consider global distortion into the path graph.
However, in stark contrast to \cref{theorem:max_distortion_graph_into_path}, global distortion remains \NP-hard in this case, at least for $p=1$, based on a reduction from \hampath{}. The proof relies on a suitable rearrangement of the distortion sum which is not possible if the individual distortions are raised to powers $p > 1$.
\begin{theorem}
    \label{theorem:global_distortion_into_path_hardness}
    \GSim{\Ddist{(1)}} is \NP-complete even if one of the input graphs is a path. 
\end{theorem}

\noindent Analogously to the relation between \GSim{\Ddist{(\infty)}} with one input graph fixed to be a path and \bandwidth{}, \GSim{\Ddist{(1)}} with the same restriction is closely related to the \textsc{Minimum Linear Arrangement} problem (in short: MLA), in which 
the vertices of a graph $G$ are assigned distinct labels from $1,\ldots,|V(G)|$ so as to minimise the sum of absolute differences between the labels of adjacent vertices. Equivalently, this can be viewed as a bijective embedding of the graph $G$ into a path that minimises the total stretch of the edges.
MLA is known to be \NP-complete \cite{GareyJohnsonStockmeyer_MLA_np_complete}, but admits polynomial-time algorithms if the input graph is restricted to trees \cite{chung_linear_arrangement_trees} and can be solved in linear time under further restrictions \cite{alemany_puig_projective_linearizations_MLA}.
In contrast, for $\GSim{\Ddist{(1)}}$, even the complexity of mapping a tree into a path remains open.

On the positive side, we show that there is an additive $\eps n^3$-approximation for $\GSim{\Ddist{(1)}}$ based on 
coarsening the metric structure of the input graphs. Closing the gap between this approximation
guarantee and the inapproximability of \cref{theorem:additive_inapprox_global_distortion} remains open.

\begin{theorem}
    \label{theorem:additive_approx_global_distortion}
    There is an algorithm for $\GSim{\Ddist{(1)}}$ that, for every $\eps > 0$, has additive error guarantee $\eps n^3$ and runs in time $n^{\O(1/\eps^2)}$.
\end{theorem}

Finally, we note that whenever $p$ is not an integer, 
we do not even know whether $\GSim{\Ddist{(p)}}$ is in \NP.
This is because for $p = \frac{k}{2}$ for some odd $k$ and a given bijection $\pi \in \Bij(V(G), V(H))$ and a bound $d \in \N$, one has to determine if 
\[
    \sum_{v, v' \in V(G)} \sqrt{\big|D_G(v, v') - D_H(\pi(v), \pi(v'))\big|^k} \leq \sqrt{d^k}.
\]
This is somewhat
reminiscent of the
\textsc{Sum of Square Roots}
problem \cite{sum_of_square_roots,goemans_sum_of_square_roots_mentioned} in which one is given nonnegative integers 
$b_1, \ldots, b_n, \ell$
and has to decide whether $\sum_{i = 1}^n \sqrt{b_i} \leq \ell$ holds.
This problem is known to be in \PSPACE, with Allender \etal{} showing containment in the fourth level of the counting hierarchy \cite{allender_et_al_counting_hierarchy}.
However, even membership in \NP{} remains open.

\section{Hardness of Gromov--Wasserstein}
\label{section:gw}
Closely related to our global distortion distances $\Ddist{(p)}$ is the well-known family of Gromov--Wasserstein distances, which was introduced by \memoli{} in 2007 \cite[Section 5]{memoli_gw_original}.
Formally, for $1 \leq p < \infty$,
the \emph{$p$-Gromov--Wasserstein distance} between two weighted compact metric spaces $\X = (X, d_{\X}, \mu_\X)$ and $\Y = (Y, d_{\Y}, \mu_\Y)$ with $\mu_\X(X) = \mu_\Y(Y)$ is defined as
\begin{equation}
    \label{equation:gromov_wasserstein}
    \Mdist{(p)}^{\GW}(\X, \Y) \ceq \frac{1}{2} \inf_{\lambda \in \mathcal C(\mu_\X, \mu_\Y)} \left(
    \int_{X \times Y} \int_{X \times Y} 
    \left| d_\X(x, x') - d_\Y(y, y') \right|^p d\lambda(x, y)d\lambda(x', y') 
    \right)^{1/p}
\end{equation}
where $\mathcal{C}(\mu_{\X},\mu_{\Y})$ denotes the set of couplings
between $\mu_{\X}$ and $\mu_{\Y}$.
A \emph{coupling} between $\mu_\X$ and $\mu_\Y$ is a measure $\lambda$ on $X \times Y$ such that $\lambda(X_0 \times Y) = \mu_{\mathcal X}(X_0)$ and 
$\lambda(X \times Y_0) = \mu_{\mathcal Y}(Y_0)$
for all Borel sets $X_0 \subseteq X$ and $Y_0 \subseteq Y$.
\memoli{} also defined a version of this distance for $p = \infty$.
Due to space constraints, we defer the treatment of $\Mdist{(\infty)}^{\GW}$ to \cref{appendix:gw_hardness_infty}.

To the best of our knowledge, there has not yet been a formal proof of the \NP-hardness of computing the Gromov--Wasserstein distance, although it has often been claimed, cited, or otherwise implied in the literature. For an extended discussion, see \cref{appendix:gw_literature}.
In the remainder of this section, we 
develop our \NP-hardness proof based on the close connection between $\Mdist{(p)}^{\GW}$ and $\Ddist{(p)}$, which we outline below.

\paragraph*{From Global Distortion to Gromov--Wasserstein}
\label{subsection:gw_derivation}
We now show that a natural fractional relaxation of $\Ddist{(p)}$, where instead of mapping vertices one-to-one we allow fractional assignments between vertices, can be seen as a special case of the $p$-Gromov--Wasserstein distance.
Specifically, for graphs $G, H$ with vertex sets $V, W$, we use that each bijection $\pi \in \Bij(V, W)$ can be expressed as a permutation matrix $P \in \{0, 1\}^{V \times W}$ with $P(v, w) = 1$ if and only if $\pi(v) = w$.
Subsequently, the set of permutation matrices $\Perm(V, W)$ can be relaxed to the set of doubly stochastic matrices $\DS(V, W)$.
Formally, 
for $1 \leq p < \infty$ we have

\begin{align*}
    \Ddist{(p)}(G, H) &= \mathmakebox[55pt][c]{\min_{\pi \in \Bij(V, W)}} \dis_{(p)}(\pi) \\ 
    &= \mathmakebox[55pt][c]{\min_{\pi \in \Bij(V, W)}} \left( \sum_{v, v' \in V} |D_G(v, v') - D_H(\pi(v), \pi(v'))|^p \right)^{1/p} \\
    &= \mathmakebox[55pt][c]{\min_{P \in \Perm(V, W)}} \left( \sum_{\substack{v, v' \in V\\w, w' \in W}} |D_G(v, v') - D_H(w, w')|^pP(v, w)P(v', w') \right)^{1/p} \\
    &\geq \mathmakebox[55pt][c]{\min_{Q \in \DS(V, W)}} 
    \underbrace{\left( \sum_{\substack{v, v' \in V\\w, w' \in W}} |D_G(v, v') - D_H(w, w')|^pQ(v, w)Q(v', w') \right)^{1/p}}_{\hspace{0pt}\scalebox{1}{$\eqqcolon \dis_{(p)}(Q)$}}\label{equation:minimiser_Q} \\
    &\eqqcolon \Ddist{(p)}^{\GW}(G, H)
\end{align*}
where we denote this fractional relaxation as $\Ddist{(p)}^{\GW}$ and $\dis_{(p)}(Q)$ is the natural generalisation of $\dis_{(p)}(\pi)$ to doubly stochastic matrices.
Note that by Birkhoff's Theorem \cite{birkhoff_theorem}, the doubly stochastic matrices are precisely the convex combinations of permutation matrices.

The superscript $\GW$ in $\Ddist{(p)}^{\GW}$ is chosen due to the close relation with the $p$-Gromov--Wasserstein distance.
This is because, modulo a scaling factor of $2n^{2/p}$, $\Ddist{(p)}^{\GW}(G, H)$ can be seen as an instance of $\Mdist{(p)}^{\GW}(\X, \Y)$ where $\X \ceq (V, \Gdist_G, \mu_G)$ and $\Y \ceq (W, \Gdist_H, \mu_H)$ are the weighted metric spaces associated with the graphs $G$ and $H$ where $\mu_G$ and $\mu_H$ are the uniform distributions on the respective vertex sets.
The discrepancy of $2n^{2/p}$ comes from the factor $\frac{1}{2}$ in \cref{equation:gromov_wasserstein}, the $p$-th root, and the normalisation of each coupling $\lambda \in \mathcal C(\mu_G,\mu_H)$ as a $V \times W$ matrix with row and column sums of $\frac{1}{n}$, rather than one for doubly stochastic matrices.
We also define $\Ddist{(\infty)}^\GW$ and discuss its relation to $\Mdist{(\infty)}^{\GW}$ in \cref{appendix:gw_hardness_infty}.
While $\Ddist{(\infty)}^{\GW}$ actually coincides with $\Ddist{(\infty)}$ (compare \cref{proposition:gw_max_distortion_relation}), this is not the case if $p \neq \infty$. 
Indeed, for $1 \leq p < \infty$, there exist graphs $G, H$ with $\Ddist{(p)}^{\GW}(G, H) < \Ddist{(p)}(G, H)$.
For instance, augmenting \cite[Example 3.6]{grohe2025thoughtsgraphsimilarity} with a centre vertex,
and attaching additional leaves to force bijections to preserve centre vertices, gives such an example by combining \cite[Theorem 3.10]{grohe2025thoughtsgraphsimilarity} and \cref{observation:diameter_two}.

\paragraph*{Proof of the Hardness of Gromov--Wasserstein}
The following theorem, together with the above relation between $\Ddist{(p)}^{\GW}$ and $\Mdist{(p)}^{\GW}$, implies the \NP-hardness of the $p$-Gromov--Wasserstein distance (\cref{theorem:gw_hard}).

\begin{theorem}
    \label{theorem:glob_gw_dist}
    For $1 \leq p \leq \infty$, $\GSim{\Ddist{(p)}^{\GW}}$ is \NP-hard. 
    Moreover, the problem is \NP-complete if $p$ is an integer or $p = \infty$. 
\end{theorem}

\begin{proof}[Proof Sketch]
    \NP-hardness is shown by reduction from \threepartition{} and inspired by \cite[Theorem 8]{grohe_rattan_woeginger_edit_distance}.
    An instance $\mathcal I = (A, a_1, \ldots, a_{3m})$ is encoded as a pair of graphs as follows: The graph $G$ contains for each $i \in [3m]$ a path on $a_i$ vertices, whereas $H$ contains $m$ disjoint cycles on $A$ vertices each.
    In a second step, we add a centre vertex to $G$ and $H$ each, to obtain connected graphs.
    Finally, we add to both graphs a large clique with an edge between one of its vertices and the centre vertex.
    We refer to each of the three vertex sets added in the above three steps as a separate \emph{layer} going forward.
    
    For a YES-instance, the witnessing partition mapping yields a natural bijection mapping sets of three paths contiguously into a single cycle, incurring distortion only at the $6m$ endvertices of the paths.

    For a NO-instance, even though the construction is simple, its analysis is much more involved, since we need to carefully track \enquote{fractional distortions}.
    We proceed in three steps.
    We first bound the amount of \enquote{mass} that is mapped between distinct layers of the constructed graphs (\cref{subclaim:bounding_inter_layer_mapping}). Here, the large clique serves as an \enquote{anchor} and forces any mass that is mapped between different layers to incur a large distortion due to a change of distance to almost all the vertices of the clique layer.
    We then observe that the distortion incurred by the $6m$ endvertices of the paths in $G$ is essentially unavoidable due to the degree differences between $G$ and $H$ together with a mass-preservation argument (\cref{subclaim:degree_diff_distortion}).
    In the final step, we argue that at least one path of $G$ must be split up fractionally across multiple cycles of $H$ and determine a lower bound on this additional distortion (\cref{subclaim:d2Q}).
    Combining the above steps, we can find a total distortion that is strictly larger than for the YES-instances.
\end{proof}

\section{Concluding Remarks}
In this paper, we studied various distortion-based graph metrics with a focus on their computational 
complexity and their relations to each other and to the Gromov--Hausdorff and Gromov--Wasserstein distances.
Thereby, our paper connects graph similarity research with the theory and application of metric spaces.
Our findings feature a broad range of algorithmic and hardness results, including, to the best of our knowledge,
the first formal \NP-hardness proof for Gromov--Wasserstein distances.

Many interesting questions remain open, most notably whether inapproximability can be shown for Gromov--Wasserstein or be improved for the other metrics we considered.
Moreover, the precise complexity of maximum distortion on caterpillars remains to be settled.
Similarly, an exact algorithm for maximum distortion between subsets of the real line is still unknown.


Beyond exact algorithms, approximation algorithms for more general (albeit possibly still restricted) settings would be 
desirable. For instance, we do not know of any approximation algorithm for trees in any of the graph metrics studied 
in this paper, even when the error guarantee may grow sublinearly in the input size. 
The only positive result in this direction is the multiplicative $\O\big(\!\sqrt{n}\big)$-approximation algorithm for 
$\Mdist{\GH}$ on metric trees (with unit edge lengths) by \cite{agarwal_et_al_GH_NP_hard_metric_trees}.
However, their approach is based on merge trees (which are also used in \cite{farahbakhshtouli_fpt_gh_14}) and does not 
carry over to the bijective setting, since it exploits that correspondences, 
like the interleavings between merge trees, may match sets of different sizes.

Given the wide range of hardness results obtained in this paper, 
the search for \FPT{} algorithms for graph distortion metrics seems worthwhile. 
However, we note that our results rule out many standard graph parameters for this purpose,
for example pathwidth and maximum degree for trees (\cref{theorem:max_distortion_hardness}).

    \paragraph*{AI Disclosure}
     We used AI only to prove a minor result (\cref{theorem:global_distortion_into_path_hardness}),
    to close a small gap in the analysis leading to \cref{proposition:2_approx_real_point_sets} 
    (\cref{claim:aux_mon_and_anti}), to find the matching lower-bound example 
    (\cref{example:caterpillar-74-example}), and to find mistakes in our draft.
    For \cref{theorem:global_distortion_into_path_hardness}, a first draft of the proof was generated by ChatGPT 5.5.
    Regarding \cref{proposition:2_approx_real_point_sets}, we originally set out to prove a factor of $\frac{3}{2}$ 
    as is suggested by both the example due to Hall and Papadimitriou 
    \cite[Figure 2]{approximating_the_distortion} and our \cref{example:caterpillar_spaces_32}.
    Guided by structural insights from our proof attempts, ChatGPT 6 Astra discovered 
    \cref{example:caterpillar-74-example} and contributed \cref{claim:aux_mon_and_anti}, which, combined with 
    our insights, yields the optimal factor $\frac{7}{4}$. 
    
    We carefully verified all AI-assisted findings, and the final versions of them were written by us.
    With the exception of the above, all proofs are entirely our own and fully written by us.
    We take full responsibility for the correctness of all results. 

\bibliography{references}

\appendix
\section{Hardness of Gromov--Wasserstein in the Literature}
\label{appendix:gw_literature}
In the literature, computing the Gromov--Wasserstein distances $\Mdist{(p)}^{\GW}$ is often said to be \NP-hard 
\cite{beier_wrongful_gw_hardness,
caratelli2026_wrongful_gw_hardness_by_non_convexity_kravtsova, 
heile2026_wrongful_gw_hardness_by_QAP_and_nonconvex,
houry_26_wrongful_gw_hardness_by_QAP,
le2021_wrongful_gw_hardness_by_non_convexity,
martin2026_wrongful_gw_hardness_by_QAP,
nakashima_2025_wrongful_gw_hardness_by_nonconvex,
peyre2020_gw_hardness_by_non_convex_qap_inprecise,
scetbon_wrongful_gw_hardness_by_QAP,
seyedi2025_wrongful_gw_hardness_by_QAP,
seyedi2026_wrongful_gw_claims,
shahbazi2026_wrongful_GW_hardness_by_nonconvex,
tran_sum_of_squares_wrongful_np_hardness_of_GW_by_kravtsova,
zhang_wrongful_hardness_of_gw_by_QAP}
or otherwise computationally intractable 
\cite{chen_2024_wrongful_gw_intractable,
nakashima_2025_wrongful_gw_hardness_by_nonconvex,
seyedi2026_wrongful_gw_intractability,
seyedi2025_wrongful_gw_hardness_by_QAP,
seyedi2026_wrongful_gw_claims}.
Many authors 
\cite{caratelli2026_wrongful_gw_hardness_by_non_convexity_kravtsova,
chen_2024_wrongful_gw_intractable,
heile2026_wrongful_gw_hardness_by_QAP_and_nonconvex,
martin2026_wrongful_gw_hardness_by_QAP,
seyedi2026_wrongful_gw_claims,
tran_sum_of_squares_wrongful_np_hardness_of_GW_by_kravtsova}
attribute this to 
\cite[Theorem 4]{kravtsova_2025_wrongful_np_hardness_of_gw} which states the \NP-hardness with an (incorrect) proof; however, this part of Theorem 4 has been removed in an updated version 
\cite{kravtsova_2026_note_computational_complexity_gw}.

The two main lines of reasoning for the \NP-hardness of Gromov--Wasserstein we found are (1) it being an instance of the notoriously hard Quadratic Assignment Problem (QAP) 
\cite{heile2026_wrongful_gw_hardness_by_QAP_and_nonconvex,
houry_26_wrongful_gw_hardness_by_QAP,
martin2026_wrongful_gw_hardness_by_QAP,
seyedi2025_wrongful_gw_hardness_by_QAP,
seyedi2026_wrongful_gw_claims,
zhang_wrongful_hardness_of_gw_by_QAP,
zhang_2023_wrongful_gw_hardness_by_QAP}
and (2) its inherently non-convex quadratic nature 
\cite{caratelli2026_wrongful_gw_hardness_by_non_convexity_kravtsova,
heile2026_wrongful_gw_hardness_by_QAP_and_nonconvex,
le2021_wrongful_gw_hardness_by_non_convexity,
scetbon_wrongful_gw_hardness_by_QAP,
seyedi2026_wrongful_gw_intractability, 
shahbazi2026_wrongful_GW_hardness_by_nonconvex}.
However, we stress that neither of these lines of reasoning is sufficient to prove hardness of $\Mdist{(p)}^{\GW}$.
The main flaw in both arguments is that hardness of $\Mdist{(p)}^{\GW}$ is claimed by showing it to be an instance of an \NP-hard problem. 
While both QAP and non-convex quadratic programming are \NP-hard in general 
\cite{quadratic_assignment_np_completeness, pardalos_quadratic_programming_one_negative_eigenvalue}, 
the implicit reductions the arguments (1) and (2) rely on go in the wrong direction.
Instead of reducing Gromov--Wasserstein to an \NP-hard problem, one would have to reduce an \NP-hard problem to Gromov--Wasserstein to prove hardness.

In the case of argument (1), there is an additional inaccuracy:
While the bijective version of Gromov--Wasserstein (which can be seen as a special case of Gromov--Monge \cite{memoli2022distancedistributionsinverseproblems,Dumont2025_existence_of_gromov_monge_maps}) has a natural formulation as a QAP, for the standard Gromov--Wasserstein distances such a formulation is not known.
This discrepancy even appears in work where (1) is not used to argue \NP-hardness \cite{peyre2020_gw_hardness_by_non_convex_qap_inprecise,solomon_entropic_gw}.
We suspect this misconception to be (at least in part) due to the conflation (at least in notation or name) of the bijective version with the standard formulation of Gromov--Wasserstein in the literature.
For example, Ryner \etal{} \cite{ryner_gw_conflation} refer to the bijective variant as Gromov--Wasserstein discrepancy, whereas Peyré \etal{} \cite{peyre_gromov_wasserstein_discrepancy_coupling} refer to the continuous setting by the same name.
Meanwhile, the term \enquote{discrete Gromov--Wasserstein} is used for the continuous setting
where the weight measures are discrete
in \cite{tran_sum_of_squares_wrongful_np_hardness_of_GW_by_kravtsova} while \enquote{discrete Gromov--Wasserstein discrepancy} refers to the bijective variant in \cite{ryner_gw_conflation}.

In our opinion, calling the bijective version \enquote{Gromov--Wasserstein} might be confusing since the term \enquote{Wasserstein} comes from the Wasserstein distance which is based on continuous transport plans. 
Similarly, the term \enquote{discrete} does not necessarily imply that transport plans are bijective. 
This version is better understood as a Gromov--Monge-type problem as mentioned in \cref{section:beyond_max_distortion}.

Finally, we note that there are works in which the authors do not explicitly claim \NP-hardness; however, the wording might leave the reader with this impression 
\cite{takeda_imprecise_language_gw_hard,
zhang_2023_wrongful_gw_hardness_by_QAP};
see also \cite[Section 4]{ryner_gw_conflation}.

We note that the citations in this section are not exhaustive.

\section{\texorpdfstring{Proofs Omitted from \cref{section:preliminaries}}{Proofs Omitted from Section~\ref{section:preliminaries}}}
\label{appendix:preliminaries}
\begin{proof}[Proof of \cref{lemma:graph_metric}]
    Fix $1 \leq p \leq \infty$ and let $\delta \in \{ \Ddist{p}, \Ddist{(p)}, \Ddist{(p)}^{\GW} \}$. 
    We verify the conditions of \cref{definition:graph_metric} for $\delta$.
    (GM1) holds by taking the identity permutation or the identity matrix in the case of $\Ddist{(p)}^{\GW}$; (GM2) is witnessed by taking the inverse (respectively the transpose for $\Ddist{(p)}^{\GW}$) of a minimiser of \cref{equation:distortion_distance_template} (respectively $\argmin_{Q \in \DS(V, W)}\dis_{(p)}(Q)$).
    Note that $\delta(G, G') = 0$ whenever $G \cong G'$.
    Together with (GM2) and (GM3), this implies (GM0).
    
    It remains to show (GM3). For $\Ddist{(p)}^{\GW}$ we refer to \cite[Theorem 5.1a]{Memoli2011}
    which shows that the generalisation $\Mdist{(p)}^{\GW}$ (see \cref{equation:gromov_wasserstein}) is a metric on the collection of all isomorphism classes of weighted compact metric spaces
    with fully supported measures, i.e. measures assigning positive mass to all non-empty open subsets. 
    For $\delta\in\{\Ddist{p}, \Ddist{(p)}\}$ let $F, G, H$ be connected graphs of the same order.
    Choose $\pi \in \Bij(V(F), V(G))$ and $\rho \in \Bij(V(G), V(H))$ so that $\dist(F, G)=\norm{D_F^{\pi} - D_G}$ and $\dist(G, H)=\norm{D_G^{\rho} - D_H}$
    where $\norm{\cdot}$ is the norm corresponding to $\delta$.
    Let $\rho \circ \pi\in\Bij(V(F), V(H))$ denote the composition of $\rho$ and $\pi$.
    The triangle inequality and permutation invariance of $\norm{\cdot}$ now give
    \begin{align*}
        \dist(F, H)
        &\le \norm{D_F^{\rho \circ \pi} - D_H}\\
        &= \norm{D_F^{\rho\circ\pi} - D_G^{\rho} + D_G^{\rho} - D_H}\\
        &\leq \norm{D_F^{\rho\circ\pi} - D_G^{\rho}} + \norm{D_G^{\rho} - D_H}\\
        &= \norm{D_F^{\pi} - D_G} + \norm{D_G^{\rho} - D_H}\\
        &= \dist(F, G) + \dist(G, H)\qedhere.
    \end{align*}
\end{proof}

\section{\texorpdfstring{Proofs Omitted from \cref{section:hardness_max_distortion}}{Proofs Omitted from Section~\ref{section:hardness_max_distortion}}}
\label{appendix:hardness_max_distortion}
We first state the definition of path decompositions and pathwidth.
\begin{definition}[Path Decompositions and Pathwidth]
    \label{definition:tree_and_path_decomposition}
    Let $G$ be a graph. A \emph{tree decomposition} of $G$ is a pair $(T, \beta)$ consisting of a tree $T$ and a map $\beta: V(T) \to 2^{V(G)}$ satisfying the following properties:
    \begin{enumerate}
        \item $\bigcup_{t \in V(T)} \beta(t) = V(G)$.
        \item For each $e \in E(G)$, there exists $t \in V(T)$ with $e \subseteq \beta(t)$.
        \item For each $v \in V(G)$, $T[\{t \in V(T) \mid v \in \beta(t) \}]$ is connected.
    \end{enumerate}
    A \emph{path decomposition} is a tree decomposition $(T, \beta)$ where $T$ is a path.
    
    The \emph{width} of a path decomposition is the size of its largest bag minus one. The \emph{pathwidth} of a graph is the minimum width of its path decompositions.
\end{definition}

\begin{proof}[Proof of \cref{theorem:max_distortion_hardness}]
    We give a reduction from the 
    \NP-complete \threepartition{} problem \cite[Problem SP15]{garey1979computers}.
    The input consists of integers $A$ and $a_1, \ldots, a_{3m}$ in unary representation satisfying $\sum_{i=1}^{3m} a_i = mA$ and $\frac{A}{4} < a_i < \frac{A}{2}$ for $1 \leq i \leq 3m$.
    The question then is to decide whether $a_1, \ldots, a_{3m}$ can be partitioned into $m$ triples such that the sum of each triple is exactly $A$.
    The constraints on the values $a_i$ ensure that exactly three values need to be grouped to form a total of $A$.
    In particular, in the case of a NO-instance, for any partition of the $a_i$ into triples, formally a mapping $\sigma: [3m] \to [m]$ with $|\sigma^{-1}(j)| = 3$ for every $j \in [m]$, there exists some $j \in [m]$ such that $a_{i_1} + a_{i_2} + a_{i_3} \geq A + 1$ where $\{i_1, i_2, i_3\} = \sigma^{-1}(j)$.
    Observe that an instance $\mathcal I \ceq (A, a_1, \ldots, a_{3m})$ is a YES-instance if and only if, for every $c \geq 1$, $c\mathcal I$ is also a YES-instance. Here, the instance $c\mathcal I$ is defined as $(cA, ca_1, \ldots, ca_{3m})$.
    Similarly, if $\mathcal I$ is a NO-instance,
    then in $c\mathcal I$ there will always exist a triple whose sum is at least $cA + c$.
    
    Note that \threepartition{} is strongly \NP-complete, that is, it remains hard for instances $(A, a_1, \ldots, a_{3m})$ with $A \leq p(m)$ for some fixed polynomial $p$. This allows us to encode the integer values of a given instance using vertex multiplicities.

    In the following, we first prove the case where the input graphs are trees of pathwidth two. Then, we give the necessary modifications for the case where the input graphs are trees of maximum degree three.

    \paragraph*{Trees of Pathwidth Two}
    Given an instance $\mathcal I \ceq (A, a_1, \ldots, a_{3m})$ of \threepartition{}, we construct a pair of graphs $(G, H)$ parametrised by integers $c, \ell$ for $\GSim{\Ddist{(\infty)}}$.
    The specific values of $c$ and $\ell$ will be chosen later in the proof.
    $G$ is constructed from the disjoint union of $3m$ paths on $\ell$ vertices each, where one of the respective endvertices is connected by an edge to a fresh vertex $r_G$.
    Then, for each $i \in [3m]$, we add $ca_i$ isolated vertices (denote this set of vertices by $C^i$) and connect all of them to the $i$-th leaf of the graph obtained so far by an edge.
    $H$ is obtained from the disjoint union of $m$ paths on $2\ell$ vertices ($Q^1, \ldots, Q^m$) and $m$ paths on 
    $\ell$ vertices ($R^1, \ldots, R^{m}$). Again we connect these paths to a common vertex $r_H$ by an edge to some endvertex of each path. For each $j \in [m]$, we add $cA$ isolated vertices (denote this vertex set by $D^j$) and introduce an edge between each such vertex and the remaining endvertex of the $j$-th path on $2\ell$ vertices.
    The resulting graphs are depicted in \cref{figure:trees_pathwidth_two}.
    For notational convenience, we write $P^i_s$ for the endvertices of the paths $P^i$ adjacent to $r_G$ and $P^i_t$ for the other endvertex, respectively. The same notation is used for the paths $Q^j$ and $R^j$. Moreover, let $D \ceq D^1 \cup \cdots \cup D^m$ be the leaf vertices corresponding to the $m$ copies of the value $A$.
     
    Note that $G$ and $H$ have $1 + 3m\ell + cmA$ vertices each, which is polynomial in the size of $\mathcal I$ if $A, c, \ell = \poly(m)$. Moreover, it is easy to see that both graphs are trees of pathwidth two, since a path decomposition of width two is obtained for either graph by placing the root in every bag, using two additional vertices to collect the edges along each $P^i$ ($Q^j, R^j$ respectively), and collecting the leaves once the respective endvertices have been reached.

    We claim that $\Ddist{(\infty)}(G, H) \leq 2\ell+1$ if $\mathcal I$ is a YES-instance and $\Ddist{(\infty)}(G, H) \geq 4\ell$ otherwise.

    \yescase Suppose that $\mathcal I$ is a YES-instance and let $\sigma: [3m] \to [m]$ be a witnessing assignment.
    We define a bijection $\pi \in \Bij(V(G), V(H))$ between $G$ and $H$ such that $\pi(r_G) = r_H$ and, for every $j \in [m]$, it holds that
    \begin{itemize}
        \item $\pi\left( \bigcup_{i \in \sigma^{-1}(j)} C^i \right) = D^j$, and
        \item if $\sigma^{-1}(j) = \{i_1, i_2, i_3\}$, then $\pi(P^{i_1} \cup P^{i_2}) = Q^j$ and $\pi(P^{i_3}) = R^j$ such that $\pi(P^{i_1}_s) = Q^j_s$, $\pi(P^{i_2}_t) = Q^j_t$, $\pi(P^{i_3}_s) = R^j_t$ and each edge of each $P^i$ is mapped onto an edge of $Q^j$ or $R^j$.
    \end{itemize}
    Informally, $\pi$ maps the root of $G$ to the root of $H$ and the leaf sets of the two trees are mapped onto each other bijectively as determined by the witnessing assignment $\sigma$.
    Finally, whenever the leaf sets of the paths $P^{i_1}, P^{i_2}, P^{i_3}$ are grouped together into the leaf set of $Q^j$, two of these paths are mapped onto $Q^j$ preserving the order of vertices to the root, while the remaining third path is mapped onto $R^j$ such that the order to the root vertex is reversed.
    
    \begin{claim}
        \label{subclaim:pi_distortion_at_most_2ell_plus_one}
        $\pi$ incurs a distortion of at most $2\ell + 1$, \ie{}
        $\dis_{(\infty)}(\pi) \leq 2\ell + 1$.
    \end{claim}
    \begin{claimproof}
        We prove the claim by showing $\dis(\pi, v, v') \leq 2\ell + 1$ for all (unordered) pairs $\{v, v'\}$ of distinct vertices $v, v' \in V(G)$.
        Any such pair containing $r_G$ incurs a distortion of at most $2\ell + 1$ 
        since $\ecc_H(r_H) = 2\ell + 1$ and $r_H = \pi(r_G)$.
        
        If $v, v' \in C^i$ for some $i \in [3m]$, then clearly $\dis(\pi, v, v') = 0$.
        If $v \in C^i$ and $v' \in C^{i'}$ where $i \neq i'$, then $\dis(\pi, v, v') = (2\ell+2) - 2 = 2\ell$ if $\pi(v), \pi(v') \in D^j$ for some $j \in [m]$ and $\dis(\pi, v, v') = |(2\ell+2) - (4\ell+2)| = 2\ell$ otherwise. 
        If $v \in C^i$ and $v' \in P^i$, then $\dis(\pi, v, v') \leq 2\ell + 1$. To see this, let $j \in [m]$ such that $\pi(v) \in D^j$. Then, if $\pi(P^i) \subseteq Q^j$, this holds trivially since the vertex in $Q^j$ that is furthest from $\pi(v)$ is at distance $2\ell$. It remains to deal with the case $\pi(P^i) = R^j$.
        Observe that $\Gdist_H(\pi(v), \pi(P^i_t)) = 2\ell+2$ since $\Gdist_H(\pi(v), r_H) = 2\ell+1$ and $\pi(P^i_t)$ is adjacent to $r_H$.
        Thus, $\Gdist_H(\pi(v), \pi(v')) = 2\ell+2+ \Gdist_G(P^i_t, v')$, whereas $\Gdist_G(v, v') = 1 + \Gdist_G(P^i_t, v')$, which implies the claim.
        
        If $v' \in P^{i'}$, then if $\sigma(i) = \sigma(i')$, by the same reasoning as above, we get $\dis(\pi, v, v') \leq 2\ell+1$. Otherwise, let $v'$ be the unique vertex in $P^{i'}$ at distance $k$ from the root, that is, $\Gdist_G(r_G, v') = k$. Then $\Gdist_G(v, v') = \ell + 1 + k$ and $\Gdist_H(\pi(v), \pi(v')) = 2\ell + 1 + (b\ell + k)$ for some $b \in \{0, 1\}$ if $\pi(P^{i'}) \subseteq Q^{j'}$ for some $j' \in [m]$, where the value of $b$ depends on whether $P^{i'}$ is mapped as the first or the second subpath into the path $Q^{j'}$.
        In any case, $\dis(\pi, v, v') = |(\ell + 1) + k - (2\ell + 1 + (b\ell + k))| \leq \ell + \ell = 2\ell$.
        If $\pi(P^{i'}) = R^{j'}$, then $\dis(\pi, v, v') = |(\ell + 1 + k) - (3\ell + 2 - k)| = 2\ell + 1 - 2k \leq 2\ell$.
        
        This already covers all pairs where one of the vertices belongs to $\{r_G\} \cup C^1 \cup \ldots \cup C^{3m}$.
        If $v, v' \in P^i$ for some $i \in [3m]$, then $\dis(\pi, v, v') = 0$ since paths are mapped isometrically under $\pi$. 
        So suppose $v \in P^{i}$ and $v' \in P^{i'}$ for distinct $i, i' \in [3m]$, let $k \ceq \Gdist_G(v, r_G)$ and $k' \ceq \Gdist_G(v', r_G)$ and note that $1 \leq k, k' \leq \ell$.
        If $\pi(v), \pi(v') \in Q^j$ for some $j \in [m]$, then clearly $\dis(\pi, v, v') \leq 2\ell$ since $|Q^{j}| = 2\ell$.
        If $\pi(v) \in Q^j$ and $\pi(v') \in Q^{j'}$ for distinct $j, j' \in [m]$, then $\Gdist_H(\pi(v), \pi(v')) = b\ell +  k + b'\ell + k'$ where $b, b' \in \{0, 1\}$ depending on whether $P^i$ (resp.\@ $P^{i'}$) is mapped into the first or second subpath of $\ell$ vertices in $Q^j$ (resp.\@ $Q^{j'}$). Hence $\dis(\pi, v, v') = |(k + k') - (b\ell + k + b' \ell + k')| \leq 2\ell$.
        
        If $\pi(v) \in R^j$ and $\pi(v') \in R^{j'}$ for some $j, j' \in [m]$, then $\dis(\pi, v, v') \leq 2\ell$ since vertices in distinct $R^j$ are at distance at most $2\ell$ in $H$ and $\Gdist_G(v, v') = k + k' \geq 2$. Note that $j \neq j'$ since $i \neq i'$ and hence $v$ and $v'$ must be mapped into distinct paths.

        The last remaining case is where one of $\{v, v'\}$ is mapped into $Q^j$ and the other into $R^{j'}$ for $j, j' \in [m]$. Assume \WLOG{} that $\pi(v) \in Q^j$.
        Then there is $b \in \{0, 1\}$ such that $\Gdist_H(\pi(v), \pi(v')) = \Gdist_H(\pi(v), r_H) + \Gdist_H(r_H, \pi(v')) = (b\ell + k) + (\ell - k' + 1)$.
        Hence $\dis(\pi, v, v') = |(k + k') - (b\ell + k + \ell - k' + 1)| = |2k' - (b+1)\ell - 1| \leq 2\ell$.

        This covers all possible pairs of vertices in $G$; hence, the claim follows.
    \end{claimproof}

    It follows that $\Ddist{(\infty)}(G, H) \leq 2\ell + 1$ in case of a YES-instance. Importantly, \cref{subclaim:pi_distortion_at_most_2ell_plus_one} and thus also the upper bound on the maximum distortion remain true independent of the choice of $c$ and $\ell$.

    \nocase
    Suppose that $\mathcal I$ is a NO-instance and assume towards a contradiction that there is a $\pi \in \Bij(V(G), V(H))$ with $\dis_{(\infty)}(\pi) < 4\ell$.
    Choose $c \ceq 3m\ell + 2$. Then it holds that $|V(H) \setminus D| = 1 + 3m\ell < c$.
    For every $i \in [3m]$ and all distinct $v, v' \in C^i$ it holds that $\Gdist_G(v, v') = 2$. Moreover, for distinct $j, j' \in [m]$, we have $\Gdist_H(w, w') = 4\ell+2$ for all $w \in D^j$ and $w' \in D^{j'}$.
    Hence, it follows that $\pi(C^i) \subseteq (V(H) \setminus D) \cup D^j$ for some $j \in [m]$.
    Since $ca_i > 3m\ell + 1$, it also holds that $\pi(C^i) \cap D \neq \emptyset$.
    In total, each $i \in [3m]$ is assigned precisely one $j \in [m]$. Let $\sigma: [3m] \to [m]$ be the corresponding assignment function.
    For each $j \in [m]$, it holds that
    \begin{align*}
        c \cdot \sum_{i \in \sigma^{-1}(j)} a_i &= \sum_{i \in \sigma^{-1}(j)} ca_i\\
        &= \textstyle \left|\pi\left( \bigcup_{i \in \sigma^{-1}(j)} C^i\right)\right| \\
        &= \textstyle \left| \bigcup_{i \in \sigma^{-1}(j)} \pi(C^i)\right| \\
        &\leq |D^j| + |V(H) \setminus D| \\
        &< cA + c \tag*{$|V(H) \setminus D| < c$}\\
        &= c(A+1).
    \end{align*}
    Dividing both sides of the above inequality by $c$ then yields $\sum_{i \in \sigma^{-1}(j)} a_i \leq A$ for all $j \in [m]$.
    Since $\sum_{i \in [3m]} a_i = mA$ this implies that
    $\sum_{i \in \sigma^{-1}(j)} a_i = A$ for all $j \in [m]$.
    Moreover, $|\sigma^{-1}(j)| = 3$ for every $j \in [m]$ as $\frac{A}{4} < a_i < \frac{A}{2}$ for all $i \in [3m]$.
    Hence $\mathcal I$ is a YES-instance, a contradiction.
    Therefore, $\dis_{(\infty)}(\pi) \geq 4\ell$ for all $\pi \in \Bij(V(G), V(H))$.
    \bigskip
    
    In total, we have shown that YES-instances yield maximum distortion at most $2\ell+1$ whereas NO-instances have maximum distortion at least $4\ell$. 
    Since $\lim_{\ell \to \infty} \frac{4\ell}{2\ell+1} = 2$, for every $\eps > 0$ there is some $\ell \geq 2$ such that $\frac{4\ell}{2\ell+1} > 2-\eps$, for example we can pick $\ell(\eps) \ceq  \max\{2, \left\lfloor \frac{1}{\eps} \right\rfloor + 1\}$. 
    Hence if there were a polynomial-time approximation algorithm for $\GSim{\Ddist{(\infty)}}$ with multiplicative error guarantee $2-\eps$ this would yield a polynomial-time decision algorithm for \threepartition{} by the above construction where we choose $\ell \ceq \ell(\eps)$. This contradicts our assumption that $\P \neq \NP$.
    Note that for fixed $\eps$, $\ell(\eps)$ is a fixed constant. 
    Choosing $\ell = \poly(m)$ suitably, we can also show additive inapproximability with error $\O(n^{1-\eps})$ for any $\eps > 0$. 
    To avoid repetition, we only show this at the very end of the proof for the case of trees of maximum degree three (see \cref{subclaim:additive_approx_hardness_max_distortion_trees}).

    \paragraph*{Trees of Maximum Degree Three}
    We now modify the construction of the above case as follows.
    Instead of choosing $\ell$ to be constant, we will choose $\ell = \Omega(m)$.
    In particular, this means that the distance between leaf sets attached to distinct paths scales at least linearly with $m$.
    As long as $c = \poly(m)$, since the $a_i$ are polynomially bounded by $m$, this means that we can replace each leaf set $C^i$, $D^j$ by roughly balanced rooted binary trees of height $\O(\log m)$ on the corresponding number of vertices. 
    For simplicity, we still refer to the resulting trees as $C^i$ and $D^j$ and denote their roots by $r(C^i)$ and $r(D^j)$ respectively. 
    Then each tree's root is connected by an edge to the corresponding vertex where the leaf set was originally attached. 
    Similarly, the root vertices in $G$ and $H$ are each replaced by binary trees of height $\lceil\log_2 3m\rceil$ (denoted $R_G$ and $R_H$) with exactly $3m$ leaves. 
    In $G$, each of these leaf vertices is attached to an endvertex of one of the $3m$ paths, whereas in $H$, only $2m$ leaves will be attached to the paths $Q^j$ and $R^j$; the remaining $m$ leaves are kept dangling.
    The trees obtained in this way are depicted in \cref{figure:trees_maximum_degree_three}.
    Note that $r \ceq |R_G| = |R_H| \leq 6m$.
    We choose $c \ceq 6m + 3m\ell+1$, which is polynomial in $m$ given that $\ell = \poly(m)$.
    For the remainder of the proof, fix some $\eps > 0$.
    
    \begin{figure}[hbtp]
    \captionsetup{singlelinecheck = true, justification=justified}
    \centering

    \begin{tikzpicture}[
            box/.style={
                draw,
                fill=white,
                rounded corners,
                minimum width=1.6cm,
                minimum height=0.8cm,
                thick
            }
        ]
    
        \newcommand{\kpath}[8]{%

            \begin{scope}[shift={(#2,#3)}]
                \def\ksidearg{#8}\def\ksideleft{left}%
                \ifx\ksidearg\ksideleft
                    \tikzset{kvlabel/.style={left=4pt},
                             kplabel/.style={right=4pt}}%
                \else
                    \tikzset{kvlabel/.style={right=4pt},
                             kplabel/.style={left=4pt}}%
                \fi

                \coordinate (#1start) at (0,0);
                \coordinate (#1end) at (0,-#4);

                \def\pad{0.2}

                \draw[rounded corners, fill=white]
                    (-\pad,\pad) rectangle (\pad,-#4-\pad);

                \fill (#1start) circle (2pt);
                \node[kvlabel] at (#1start) {#6};
                \fill (#1end)   circle (2pt);
                \node[kvlabel] at (#1end) {#7};

                \node[rotate=90, draw=none, fill=none] (mid) at ($(#1start)!0.5!(#1end)$) {$\cdots$};
                \node[kplabel, draw=none, fill=none] at (mid) {#5};
            \end{scope}
        }
    
        \newcommand{\bintree}[6]{%
            \pgfmathsetmacro{\btr}{min(2, 0.3*#6*28.4527)}
            \foreach \d in {0,...,#4} {
                \pgfmathtruncatemacro{\btlast}{2^\d - 1}
                \pgfmathtruncatemacro{\btup}{\d - 1}
                \foreach \i in {0,...,\btlast} {
                    \pgfmathsetmacro{\btx}{#2 + #6*((\i+0.5)*2^(#4-\d) - 2^(#4-1))}
                    \coordinate (#1-\d-\i) at (\btx, #3-\d*#5);
                    \ifnum\d>0
                        \pgfmathtruncatemacro{\btp}{div(\i,2)}
                        \draw[thick] (#1-\btup-\btp) -- (#1-\d-\i);
                    \fi
                    \fill (#1-\d-\i) circle (2pt);
                }
            }
        }
    
        \coordinate (G) at (0,-3);
    
        \begin{scope}[shift=(G)]
            \node[box, minimum width=5cm, minimum height=2.6cm, fill=green!10] at (1.5, -1.1) {};
            \node[box, minimum width=5cm,   minimum height=1.5cm, fill=blue!10]  at (1.5, -3.3) {};
    
            \kpath{G0}{0}{-1}{1}{$P^1$}{}{}{};
            \draw[thick] (0,-2) -- (0,-2.75);
            \bintree{GB0}{0}{-2.75}{3}{0.2}{0.2};
            
            \kpath{G3}{3}{-1}{1}{$P^{3m}$}{}{}{left};
            \draw[thick] (3,-2) -- (3,-2.75);
            \bintree{GB3}{3}{-2.75}{3}{0.2}{0.2};
    
            \bintree{GT}{1.5}{0}{3}{0.2}{0.4}
            \draw[thick] (GT-3-0) -- (G0start);
            \draw[thick] (GT-3-7) -- (G3start);
    
            \node at (1.5,-1.5)  {$\cdots$};
            \node at (1.5,-3.15) {$\cdots$};
            \node at (1.5,-3.75) {$\cdots$};
            \node at (0,-3.75)   {$ca_1$};
            \node at (3,-3.75)   {$ca_{3m}$};
            \node at (1.5, -4.5) {$G$};
        \end{scope}

        \coordinate (H) at (6,-2);

    \begin{scope}[shift=(H)]
        \node[box, minimum width=5.25cm, minimum height=3.6cm, fill=green!10] at (1.5, -1.6) {};
        \node[box, minimum width=6cm,   minimum height=1.5cm, fill=blue!10]  at (1.5, -4.3) {};

        \kpath{HQ1a}{-0.2}{-1}{2}{$Q^1$}{}{}{right};
        \node at (0.35, -2) {$\cdots$};
        \kpath{HQma}{0.9}{-1}{2}{$Q^m$}{}{}{left};

        \kpath{HR1}{2.1}{-1}{1}{$R^1$}{}{}{right};
        \node at (2.65, -1.5) {$\cdots$};
        \kpath{HRm}{3.2}{-1}{1}{$R^m$}{}{}{left};

        \foreach \x in {0,3} {
            \bintree{HB\x}{\x}{-3.75}{4}{0.15}{0.15}
            \node at (\x,-4.7) {$cA$};
        }
        \draw[thick]
            (-0.2,-3) -- (0,-3.75)
            (0.9,-3) -- (3,-3.75);

        \bintree{HT}{1.5}{0}{3}{0.2}{0.4}
        \draw[thick] (HT-3-0) -- (HQ1astart);
        \draw[thick] (HT-3-2) -- (HQmastart);
        \draw[thick] (HT-3-5) -- (HR1start);
        \draw[thick] (HT-3-7) -- (HRmstart);

        \draw[red, thick] (HT-3-3) circle (0.14);
        \draw[red, thick] (HT-3-4) circle (0.14);

        \node at (1.5,-4.15) {$\cdots$};
        \node at (1.5,-4.75) {$\cdots$};
        \node at (1.5, -5.5) {$H$};
    \end{scope}
    
    \end{tikzpicture}

    \caption{The trees of maximum degree three constructed in the reduction. White boxes correspond to paths of $\ell$ vertices ($P, R$) and paths of $2\ell$ vertices ($Q$). 
    The first and last $m$ leaves of the upper binary tree are connected to the paths $Q^1, \ldots, Q^m$ and $R^1,\ldots, R^m$, respectively.
    The leaves in between (circled in red) are \enquote{kept dangling}.
    Note that the vertex counts do not match those of a valid construction for an instance $\mathcal I$ but are for illustration purposes only.}
    \label{figure:trees_maximum_degree_three}
\end{figure}
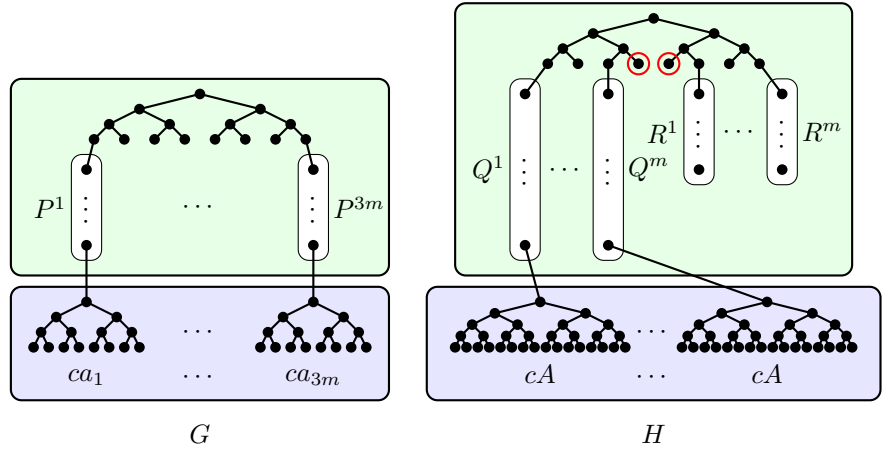
    
    \yescase Let $\mathcal I$ be a YES-instance.
    Then there is a witnessing assignment $\sigma: [3m] \to [m]$ with $|\sigma^{-1}(j)| = 3$ for each $j \in [m]$.
    We modify the bijection for the YES-case of the previous case by mapping $V(R_G)$ into $V(R_H)$ (previously mapping $r_G$ to $r_H$) and, for each $j \in [m]$ and $\{i_1, i_2, i_3\}$, mapping $V(C^{i_1}) \cup V(C^{i_2}) \cup V(C^{i_3})$ into $V(D^j)$ (previously mapping the leaf sets corresponding to $a_{i_1}, a_{i_2}, a_{i_3}$ in $G$ into a leaf set of size $A$ in $H$). 
    
    Let $\pi$ be a bijection satisfying these conditions and note that $\pi$ keeps the overall structure of the bijection we used in the case of pathwidth two. However, the precise distances between vertices of the modified subtrees are no longer fully determined since we have not specified these mappings; they are arbitrary bijections. The next claim shows that this may increase the maximum distortion between the two trees only by a negligible amount.
    
    \begin{claim}
        Let $d \geq 1$ such that $A, a_1, \ldots, a_{3m} \leq m^d$. 
        Then $\pi$ incurs a maximum distortion of at most $2\ell + 1 +  6\lceil \log cm^d \rceil$.
    \end{claim}
    \begin{claimproof}
        For sufficiently large $m$, we can identify the vertices of $C^i$ for each $i \in [3m]$ in $G$ and $D^j$ for each $j \in [m]$ in $H$, as well as the vertex sets of $R_G$ and $R_H$, respectively.
        This is because the diameter of each such set in the respective graphs is bounded by $2\lceil \log cm^d \rceil$ and thus mapping arbitrarily within these sets incurs a distortion of at most $2 \lceil \log cm^d \rceil$.
        Since each pair of distinct vertices $(v, v')$ has a unique shortest path in $G$ and the same holds for $(\pi(v), \pi(v'))$ in $H$, it is clear that at most three such binary trees are (partially) traversed in any shortest path.
        Hence, this yields an \emph{additional} distortion (compared to the distortion discussed in \cref{subclaim:pi_distortion_at_most_2ell_plus_one}) of at most $6 \lceil \log cm^d \rceil$.
    \end{claimproof}
    Since $\lceil \log cm^d \rceil = \O(\log m)$ and $\ell = \Omega(m)$, it holds that $\dis_{(\infty)}(\pi) \leq \left(2 + \frac{\eps}{3}\right)\ell$ given that $m$ is sufficiently large.

    \nocase Suppose $\mathcal I$ is a NO-instance.
    Similarly to before, by choice of $c$, any bijection between $V(G)$ and $V(H)$ needs to map two vertices belonging to the same tree $C^i$ into distinct subtrees $D^j, D^{j'}$ of $H$, which incurs a distortion of at least $4\ell - 2\lceil \log cm^d \rceil$.
    Hence, for sufficiently large $m$ it follows that $\dis_{(\infty)}(\pi) \geq \left(4 - \frac{\eps}{3}\right)\ell$ for all $\pi \in \Bij(V(G), V(H))$.

    Note that in both cases we require $m$ to be sufficiently large, say at least $m_0$.
    However, it is easy to see that instances with $m < m_0$ can be checked immediately by brute force since $A \leq p(m)$ and thus these instances only consist of integers bounded by a constant.
    For larger values of $m$, we have shown that the gap between YES- and NO-instances is at least
    \[
        \frac{\left(4 - \frac{\eps}{3}\right)\ell}{\left(2 + \frac{\eps}{3}\right)\ell} > 2-\eps.
    \]
    Hence, a polynomial-time approximation algorithm with multiplicative error guarantee $2-\eps$ would yield a polynomial-time algorithm for \threepartition{}.
    This proves the multiplicative approximation hardness.
    It remains to analyse the additive approximation hardness.

    \begin{claim}
        \label{subclaim:additive_approx_hardness_max_distortion_trees}
        For each $\eps > 0$, there exists a $b(\eps) \in \N$ such that $\ell = m^{b(\eps)}$
        yields $\ell = \omega(n^{1-\eps})$.
    \end{claim}
    \begin{claimproof} 
        Recall that by strong \NP-hardness of \threepartition{}, there is some $d \in \N$
        and a polynomial $p(x)$ of degree $d$ such that the problem remains \NP-hard for instances with
        $A \leq p(m)$. For such instances, our reduction yields graphs of order
        \begin{align*}
            n &= r + 3m\ell + cmA \\
            &\leq 6m + 3m\ell + (6m + 3m\ell + 1)mA = \O(\ell \cdot m^{d+2}),
        \end{align*}
        where we use $r \leq 6m$.
        Choosing $\ell \ceq m^b$ for some constant $b > 0$ still results in a polynomial-time reduction and $m = \Omega(n^{1/(b+d+2)})$, so $\ell = m^b = \Omega(n^{b/(b+d+2)})$. 
        As $\lim_{b \to \infty} \frac{b}{b+d+2} = 1$, for every $\eps > 0$, we may choose
        \[
            \mfrac{b}{b+d+2} > 1-\eps \quad\iff\quad 
             b > \mfrac{(d+2)(1-\eps)}{\eps}
        \]
        to obtain $\ell = \Omega(n^{b/(b+d+2)}) \subseteq \omega(n^{1-\eps})$.
    \end{claimproof}
    \cref{subclaim:additive_approx_hardness_max_distortion_trees} shows that approximating $\GSim{\Ddist{(\infty)}}$ with additive error guarantee $\O(n^{1-\eps})$
    remains \NP-hard. 
\end{proof}

In preparation for the proof of \cref{theorem:gh_inapprox_result_julius}, we generalise the definition
of $\dis_{(\infty)}$ to correspondences.
For graphs $G, H$, and a correspondence $R \in \Corr(V(G), V(H))$, define
\begin{align*}
    \dis_{(\infty)}(R) \ceq \max_{(v,w), (v',w') \in R} \big| \Gdist_G(v,v')-\Gdist_H(w,w') \big|
\end{align*}
and note that
$\Mdist{\GH}(G, H) = \frac{1}{2} \cdot \min_{R} \dis_{(\infty)}(R)$ where $R$ ranges over the set
$\Corr(V(G), V(H))$.
Further, for $U \subseteq V(H)$, define
\begin{align*}
    R^{-1}(U) \ceq \{v \in V(G) \mid (v, u) \in R \text{ for some } u \in U\}.
\end{align*}

\begin{proof}[Proof of \cref{theorem:gh_inapprox_result_julius}]
    For an instance $\mathcal I = (A, a_1, \ldots, a_{3m})$ of
    \threepartition{} and an even integer $\ell \geq 4$, we construct two graphs $G, H$ of order
    polynomial in $A,m,\ell$ such that
    $\Mdist{\GH}(G, H) \leq 2\ell$ if $\mathcal I$ is a YES-instance
    and $\Mdist{\GH}(G, H) \geq \frac{5}{2}\ell$ if $\mathcal I$ is a NO-instance.
    After providing the reduction and proving its correctness, we use this to derive the approximation hardness claim by choosing $\ell$ appropriately.

    The graph $G$ contains for each $i \in [3m]$ a path on $\ell$ vertices
    $p^1_i, \hdots, p_i^\ell$.
    These $3m$ paths are connected using a vertex $r_G$ with edges to each $p^\ell_i$.
    Moreover, for each $i \in [3m]$ and $h \in [a_i]$ there is a path on
    $6\ell$ vertices $s^1_{i,h}, \hdots, s^{6\ell}_{i,h}$. For each $i \in [3m]$ the corresponding
    $a_i$ paths are connected to $p^1_i$ using edges from $s^{6\ell}_{i,h}$.

    Similarly, the graph $H$ contains a vertex $r_H$ and for each $j \in [m]$ a path on $3\ell$ vertices
    $q^1_j, \hdots, q^{3\ell}_j$ such that $q^{3\ell}_j$ and $r_H$ are connected by an edge.
    Further, for each $j \in [m]$ and $k \in [A]$ there is a path on $6\ell$ vertices
    $t^1_{j,k}, \hdots, t^{6\ell}_{j,k}$ and an edge from $t^{6\ell}_{j,k}$ to $q^1_j$.

    The graphs $G$ and $H$ are depicted in \cref{fig:max_distortion_gh_trees}.
    In the following, we denote the vertex sets of $G$ and $H$ by $V$ and $W$, respectively.
    Note that both graphs are trees containing exactly $1+3m\ell+6m\ell A$ vertices each, which
    is polynomial in $A, m$, and $\ell$.
    Clearly, $G$ and $H$ can be constructed from $\mathcal I$ in polynomial time.
    Moreover, $G$ and $H$ have pathwidth at most three.

    To see this, when removing the root of either graph, each remaining component is a subdivision of a star and has a path decomposition of width two by keeping the centre in each bag and collecting the edges along pendant paths using two additional vertices.
    Thus, we may combine these path decompositions to a path decomposition of the entire tree by adding the root to every bag and attaching paths consecutively at their endpoints. Each bag contains at most four vertices yielding a width of three. It is also not hard to see that $G$ and $H$ require a width of three, for example
    using \cite[Theorem 4]{pathwidth_of_trees}.

\begin{figure}[t]
    \captionsetup{singlelinecheck = true, justification=justified}
    \centering
    \begin{tikzpicture}
        \newcommand{\kpath}[8]{%

            \begin{scope}[shift={(#2,#3)}]
                \def\ksidearg{#8}\def\ksideleft{left}%
                \ifx\ksidearg\ksideleft
                    \tikzset{kvlabel/.style={left=4pt},
                             kplabel/.style={right=4pt}}%
                \else
                    \tikzset{kvlabel/.style={right=4pt},
                             kplabel/.style={left=4pt}}%
                \fi

                \coordinate (#1start) at (0,0);
                \coordinate (#1end) at (0,-#4);

                \def\pad{0.2}

                \draw[rounded corners, fill=white]
                    (-\pad,\pad) rectangle (\pad,-#4-\pad);

                \fill (#1start) circle (2pt);
                \node[kvlabel] at (#1start) {#6};
                \fill (#1end)   circle (2pt);
                \node[kvlabel] at (#1end) {#7};

                \node[rotate=90, draw=none, fill=none] (mid) at ($(#1start)!0.5!(#1end)$) {$\cdots$};
                \node[kplabel, draw=none, fill=none] at (mid) {#5};
            \end{scope}
        }

        \coordinate (G) at (0, 0);
        \coordinate (H) at (7.5, 0);

        \begin{scope}[shift=(G)]
            \coordinate (rG) at (1.5,-0.5);
            \fill (rG) circle (2pt);
            \node[above=3pt] at (rG) {$r_G$};
            \kpath{P1}{-0.2}{-1}{1}{}{$p_1^\ell$}{$p_1^1$}{left};
            \kpath{P3m}{3.2}{-1}{1}{}{$p_{3m}^\ell$}{$p_{3m}^1$}{right};
            \node at (1.5,-1.5) {$\cdots$};
            \draw[thick] (-0.2,-1) -- (rG);
            \draw[thick] (3.2,-1) -- (rG);

            \kpath{S11}{-0.75}{-2.75}{1}{}{$s_{1,1}^{6\ell}$}{$s_{1,1}^1$}{left};
            \kpath{S1a1}{0.35}{-2.75}{1}{}{$s_{1, a_1}^{6\ell}$}{$s_{1, a_1}^{1}$}{right};
            \kpath{S3m1}{2.65}{-2.75}{1}{}{$s_{3m, 1}^{6\ell}$}{$s_{3m, 1}^{1}$}{left};
            \kpath{S3ma3m}{3.75}{-2.75}{1}{}{$s_{3m, a_{3m}}^{6\ell}$}{$s_{3m, a_{3m}}^{1}$}{right};

            \draw[thick] (-0.75, -2.75) -- (-0.2, -2);
            \draw[thick] (0.35, -2.75) -- (-0.2, -2);
            \draw[thick] (2.65, -2.75) -- (3.2, -2);
            \draw[thick] (3.75, -2.75) -- (3.2, -2);

            \node at (-0.2, -3.25) {$\cdots$};
            \node at (3.2, -3.25) {$\cdots$};

            \node at (1.5, -4.4) {$G$};
        \end{scope}
        
        \begin{scope}[shift=(H)]
            \coordinate (rH) at (1.5,-0.5);
            \fill (rH) circle (2pt);
            \node[above=3pt] at (rH) {$r_H$};
            \kpath{Q1}{-0.2}{-1}{1}{}{$q_1^{3\ell}$}{$q_1^1$}{left};
            \kpath{Qm}{3.2}{-1}{1}{}{$q_{m}^{3\ell}$}{$q_{m}^1$}{right};
            \node at (1.5,-1.5) {$\cdots$};
            \draw[thick] (-0.2,-1) -- (rH);
            \draw[thick] (3.2,-1) -- (rH);

            \kpath{T11}{-0.75}{-2.75}{1}{}{$t_{1,1}^{6\ell}$}{$t_{1,1}^1$}{left};
            \kpath{T1A}{0.35}{-2.75}{1}{}{$t_{1, A}^{6\ell}$}{$t_{1, A}^{1}$}{right};
            \kpath{Tm1}{2.65}{-2.75}{1}{}{$t_{m, 1}^{6\ell}$}{$t_{m, 1}^{1}$}{left};
            \kpath{TmA}{3.75}{-2.75}{1}{}{$t_{m, A}^{6\ell}$}{$t_{m, A}^{1}$}{right};

            \draw[thick] (-0.75, -2.75) -- (-0.2, -2);
            \draw[thick] (0.35, -2.75) -- (-0.2, -2);
            \draw[thick] (2.65, -2.75) -- (3.2, -2);
            \draw[thick] (3.75, -2.75) -- (3.2, -2);

            \node at (-0.2, -3.25) {$\cdots$};
            \node at (3.2, -3.25) {$\cdots$};
            \node at (1.5, -4.4) {$H$};
        \end{scope}    
    \end{tikzpicture}
    
    \caption{The trees $G$ and $H$ for the proof of
    \cref{theorem:gh_inapprox_result_julius}. White boxes correspond to paths.}
    \label{fig:max_distortion_gh_trees}
\end{figure}
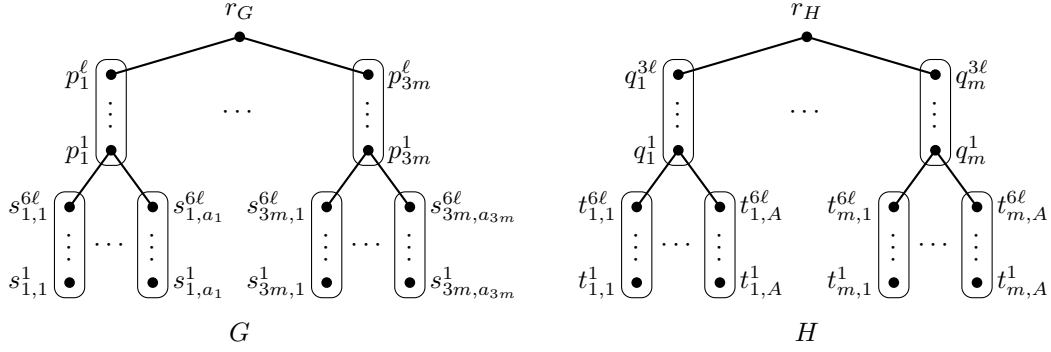 
    
    \yescase Suppose that $\mathcal I$ is a YES-instance.
    Then there is a function
    $\sigma \colon [3m] \to [m]$
    with $\abs*{\sigma^{-1}(j)}=3$ and $\sum_{i \in \sigma^{-1}(j)} a_i = A$ for all $j \in [m]$.
    Equivalently, there exists a bijection
    $\kappa \colon [m] \times [3] \to [3m]$
    such that $a_{\kappa(j,1)} + a_{\kappa(j,2)} + a_{\kappa(j,3)}=A$ for all $j \in [m]$.
    Note that $\kappa$ relates to $\sigma$ by the condition $\sigma^{-1}(j) = \{\kappa(j,1), \kappa(j,2), \kappa(j,3)\}$
    for all $j \in [m]$. We will use $\kappa$ to define a correspondence between $G$ and $H$.
    For $j \in [m]$ and $k \in [3]$
    let $\eta(j,k,\cdot) \colon [a_{\kappa(j,k)}] \to [A]$ be the
    injection defined by $\eta(j,k,h) = h +  \sum_{c=1}^{k-1} a_{\kappa(j,c)}$ and note that
    the images of $\eta(j,1,\cdot), \eta(j,2,\cdot)$ and $\eta(j,3,\cdot)$ partition $[A]$
    according to the sizes $a_{\kappa(j,1)}, a_{\kappa(j,2)}, a_{\kappa(j,3)}$ into sets of
    consecutive integers. Consider the correspondence $R \in \Corr(V, W)$ defined by
    \begin{bracketenumerate}
        \item $(r_G, r_H) \in R$,
        \item $\big( p^r_{\kappa(j,k)}, q_j^{(k-1)\ell + r} \big ) \in R$
            for all $j \in [m],\, k \in [3],\, r \in [\ell]$, and
        \item $\big ( s^r_{\kappa(j,k), h}, t^r_{j, \eta(j,k,h)} \big ) \in R$
            for all $j \in [m],\, k \in [3],\, h \in [a_{\kappa(j,k)}],\, r \in [6\ell]$.
    \end{bracketenumerate}

    \noindent Note that $R$ is the graph of a bijection between $V$ and $W$, \ie{} $R = \{(v, \pi(v)) \mid v \in V\}$ for some $\pi \in \Bij(V, W)$.

    \begin{claim}
        \label{subclaim:gh_4mu}
        $\dis_{(\infty)}(R) = 4\ell$.
    \end{claim}
    \begin{claimproof}
        We consider the distortions between all non-trivial combinations (up to symmetry) of pairs in
        $R$.

        \proofsubparagraph{Case (1) and (2):} Bounding each term individually yields
        \begin{align*}
            \Big|\Gdist_{G}\left(r_G, p^r_{\kappa(j,k)}\right)
            -\Gdist_{H}\left(r_H, q_j^{(k-1)\ell + r}\right)\Big|
            \leq \max\{ \ell, 3\ell \} = 3\ell.
        \end{align*}

        \proofsubparagraph{Case (1) and (3):} It holds that
        \begin{align*}
            \abs*{\Gdist_{G}\left(r_G, s^r_{\kappa(j,k), h} \right)
                -\Gdist_{H}\left(r_H, t^r_{j, \eta(j,k,h)}\right)}
            = \abs{\ell + 6\ell + 1 - r - (3\ell + 6\ell+1 - r)}
            = 2\ell.
        \end{align*}

        \proofsubparagraph{Case (2) and (2):} Let $j, j' \in [m],\, k, k' \in [3]$, and $r, r' \in [\ell]$.
        If $j = j'$, then we can again bound in terms of the maximal possible distances
        \begin{align*}
            \abs*{\Gdist_{G}\left (p^r_{\kappa(j,k)}, p^{r'}_{\kappa(j',k')} \right )
                -\Gdist_{H}\left (q_j^{(k-1)\ell + r}, q_{j'}^{(k'-1)\ell + r'} \right )}
            \leq \max \{2\ell, 3\ell-1\} = 3\ell-1.
        \end{align*}
        If, on the other hand $j \neq j'$, observe that the distance in $H$ decomposes along
        the root $r_H$. Further, it holds that $\kappa(j,k) \neq \kappa(j',k')$,
        meaning the distance in $G$ decomposes along $r_G$. Consequently,
        \begin{align*}
            &\abs*{\Gdist_{G}\left (p^r_{\kappa(j,k)}, p^{r'}_{\kappa(j',k')} \right )
                -\Gdist_{H}\left (q_j^{(k-1)\ell + r}, q_{j'}^{(k'-1)\ell + r'} \right )} \\
            =\ &\abs*{2\ell - (r-1+r'-1)
                - \big ( 6\ell - ( (k-1)\ell+r-1+(k'-1)\ell+r'-1 ) \big )} \\
            =\ &\abs*{4\ell - (k-1)\ell - (k'-1)\ell} \leq 4\ell
        \end{align*}
        as $k,k' \in [3]$.

        \proofsubparagraph{Case (2) and (3):} Let
        $j, j' \in [m],\, k,k' \in [3],\, r \in [\ell],\, r' \in [6\ell]$, and $h' \in [a_{\kappa(j',k')}]$.
        We first consider the subcase $j=j'$. From $j = j'$ it follows that the distance in $H$
        decomposes along $q^1_j$ as
        \begin{align*}
            \Gdist_{H} \big ( q_j^{(k-1)\ell + r},  t^{r'}_{j', \eta(j',k',h')} \big )
            = (k-1)\ell + r + 6\ell - r'.
        \end{align*}
        If additionally $k = k'$, then $\kappa(j,k) = \kappa(j', k')$ and the distance in $G$ decomposes
        along $p_{\kappa(j,k)}^1$ as
        \begin{align*}
            \Gdist_{G} \big (p^r_{\kappa(j,k)}, s^{r'}_{\kappa(j',k'), h'} \big )
            = 6\ell + r - r'.
        \end{align*}
        Thus, the distortion incurred in this case is $\abs*{(k-1)\ell } \leq 2 \ell$ because $k \in [3]$.
        If, on the other hand $k \neq k'$, then $\kappa(j,k) \neq \kappa(j', k')$ and the distance in $G$
        decomposes along $r_G$ as
        \begin{align}\label{eq:tsdist}
            \Gdist_{G} \left ( p^r_{\kappa(j,k)}, s^{r'}_{\kappa(j',k'), h'} \right )
            = \ell - (r-1) + 7\ell - (r'-1)
            = 8\ell + 2 - r - r'.
        \end{align}
        Hence the distortion incurred in this case is
        \begin{align*}
            \abs*{(k-1)\ell + r + 6\ell - r' - (8\ell + 2 - r - r')}
            = \abs*{(k-1)\ell - 2\ell + 2r - 2}.
        \end{align*}
        As $r \in [\ell]$,
        this results for $k=1$ in $\abs{2r - 2\ell - 2} \leq 2\ell$,
        for $k=2$ in $\abs{2r-\ell-2} \leq \ell$,
        and for $k = 3$ in $\abs{2r - 2} \leq 2\ell - 2$.

        Next, we consider the subcase $j \neq j'$. Now the distance in $H$ decomposes along $r_H$ as
        \begin{align*}
            \Gdist_{H} \left ( q_j^{(k-1)\ell + r},  t^{r'}_{j', \eta(j',k',h')} \right )
            &= 3\ell + 1 -(k-1)\ell - r + 9\ell + 1 - r' \\
            &= 12\ell - (k-1)\ell + 2 - r - r'.
        \end{align*}
        Since $j \neq j'$ implies $\kappa(j,k) \neq \kappa(j',k')$, the distance in
        $G$ is simply the one given by Equation (\ref{eq:tsdist}) and the resulting distortion is
        \begin{align*}
            \abs*{12\ell - (k-1)\ell + 2 - r - r' - (8\ell + 2 - r - r')}
            = \abs*{4\ell - (k-1)\ell} \leq 4\ell.
        \end{align*}

        \proofsubparagraph{Case (3) and (3):}
        Let $j,j' \in [m],\, k,k' \in [3],\, h \in [a_{\kappa(j,k)}],\, h' \in [a_{\kappa(j',k')}]$, and $r, r' \in [6\ell]$. \WLOGcap{} we may assume $r \leq r'$.
        We again first consider the case $j = j'$. If additionally
        $\eta(j,k,h)=\eta(j',k',h')$, then $\kappa(j,k) = \kappa(j',k')$ and $h=h'$ by definition of
        $\eta$. Hence the distortion incurred in this case is
        \begin{align*}
            \abs*{\Gdist_{G} \left (s^r_{\kappa(j,k), h}, s^{r'}_{\kappa(j',k'), h'} \right)
                - \Gdist_{H}\left(t^r_{j, \eta(j,k,h)}, t^{r'}_{j', \eta(j',k',h')} \right )}
                = \abs*{r' - r - (r'-r)} = 0.
        \end{align*}
        If, on the other hand $\eta(j,k,h) \neq \eta(j',k',h')$, then
        \begin{align*}
            \Gdist_{H}\left (t^r_{j, \eta(j,k,h)}, t^{r'}_{j', \eta(j',k',h')} \right )
            = 12\ell + 2 - r - r'
        \end{align*}
        while the distance in $G$ depends on the relation of $\kappa(j,k)$ and $ \kappa(j',k')$.
        Clearly, if $\kappa(j,k)=\kappa(j',k')$, then
        \begin{align*}
            \Gdist_{G}\left (s^r_{\kappa(j,k), h}, s^{r'}_{\kappa(j',k'), h'} \right )
            = 12\ell + 2 - r - r'
        \end{align*}
        and the distortion is zero. Otherwise $\kappa(j,k)\neq \kappa(j',k')$ and
        \begin{align}\label{eq:last_helper_dist}
            \Gdist_{G} \left (s^r_{\kappa(j,k), h}, s^{r'}_{\kappa(j',k'), h'} \right )
            = 14\ell + 2 - r - r'
        \end{align}
        and the distortion is $2\ell$.

        Finally, consider the subcase $j \neq j'$. The distance in $H$ decomposes along $r_H$ as
        \begin{align*}
            \Gdist_{H}\left (t^r_{j, \eta(j,k,h)}, t^{r'}_{j', \eta(j',k',h')} \right )
            = 18\ell + 2 - r - r'.
        \end{align*}
        Since $j \neq j'$ implies $\kappa(j,k) \neq \kappa(j',k')$, the distance in $G$ is given by Equation
        (\ref{eq:last_helper_dist}) and thus the distortion in this case is $4\ell$.\bigskip

        \noindent As all distortions are bounded by $4\ell$ and there are pairs attaining $4\ell$, the claim follows.
    \end{claimproof}

    By the above claim, $\dis_{(\infty)}(R) = 4\ell$ and hence $\Mdist{\GH}(G, H) \leq  \frac{1}{2} \cdot 4\ell = 2\ell$. Importantly, \cref{subclaim:gh_4mu} holds regardless of the precise choice of $\ell$.

    \nocase Suppose next that $\mathcal I$ is a NO-instance and assume towards a contradiction that
    there exists a correspondence $R \in \Corr(V,W)$ with $\dis_{(\infty)}(R) < 5\ell$.
    
    For $i \in [3m],\, h \in [a_i]$, and $j \in [m],\, k \in [A]$ we refer to
    \[
        S(G, i, h) \ceq \left\{s^1_{i,h}, \ldots, s^{3\ell/2}_{i,h}\right\} \quad\text{and}\quad
        S(H, j,k) \ceq \left\{t^1_{j,k}, \ldots, t^{\ell/2}_{j,k}\right\}
    \]
    as a \emph{leaf segment} of $G$ and $H$, respectively.
    Note that while the size of leaf segments differs between the two graphs,
    the number of leaf segments is exactly $mA$ in both trees.
    We shall prove that $R$ induces a bijection between
    the set of leaf segments in $G$ and the set of leaf segments in $H$.

    Let $j \in [m]$ and $k \in [A]$.
    Every vertex $w \in S(H, j,k)$ has eccentricity at least $17\ell + \frac{\ell}{2}$.
    Since $\dis_{(\infty)}(R) < 5\ell$, any $v \in V$ with $(v, w) \in R$ must satisfy $\ecc_G(v) > 12\ell + \frac{\ell}{2}$
    which implies that $v$ is contained in a leaf segment of $G$.
    Specifically,
    there exist $i \in [3m]$ and $h \in [a_i]$ such that $R^{-1}(S(H, j, k)) \cap S(G, i, h)
    \neq \emptyset$.

    We claim that in fact $R^{-1}(S(H, j, k)) \subseteq S(G, i, h)$.
    Suppose towards a contradiction that there exists $(v', w') \in R$ with
    $w' \in S(H, j, k)$ but $v' \notin S(G, i, h)$. By the previous paragraph, there exists $(v,w) \in R$
    such that $v \in S(G, i, h)$ and $w \in S(H, j, k)$.
    Since $\ecc_G(v') > 12\ell + \frac{\ell}{2}$, $v'$ lies in a different leaf segment of $G$ than $v$.
    Distinct leaf segments are far apart in $G$ and in particular we have $\Gdist_G(v, v') \geq 9\ell + 2$.
    But $\Gdist_{H}(w, w') < \frac{\ell}{2}$, implying
    $\dis_{(\infty)}(R) > 9\ell + 2 - \frac{\ell}{2} \geq 5\ell$, a contradiction.

    Hence $R^{-1}$ maps each leaf segment of $H$ completely into some leaf segment of $G$,
    meaning there is a well-defined map $\phi \colon [m] \times [A] \to [3m] \times \N$
    given by $\phi(j,k) = (i, h)$ whenever
    $R^{-1}(S(H, j, k)) \subseteq S(G, i, h)$.
    We shall prove that $\phi$ is injective, meaning different $H$ leaf segments cannot be images of
    the same $G$ leaf segment, or put differently, $R$ cannot split a leaf segment of $G$ into
    multiple leaf segments of $H$. Suppose towards a contradiction that there are $(v,w),(v',w') \in R$
    with $v,v' \in S(G, i, h)$ and $w \in S(H, j, k), w' \in S(H, j', k')$, where $(j,k) \neq (j',k')$.
    Note that $\Gdist_{G}(v,v') \leq \frac{3}{2}\ell - 1$ while $\Gdist_{H}(w,w') \geq 11\ell$ and
    thus $\dis_{(\infty)}(R) \geq 11\ell - (\frac{3}{2}\ell - 1) \geq 5\ell$, a contradiction.

    Since there are precisely $mA$ many leaf segments in both $G$ and $H$,
    it follows that the image of $\phi$ is 
    $\{(i, h) \mid i \in [3m], h \in [a_i]\}$, meaning $\phi$ provides a bijection between the set
    of leaf segments of $H$ and the set of leaf segments of $G$.
    It remains to show that for all $i \in [3m]$ and $h \neq h' \in
    [a_i]$ the first components of $\phi^{-1}(i, h)$ and $\phi^{-1}(i, h')$ coincide, or put differently,
    each $\{i\} \times [a_i]$ is mapped into a single $\{j\} \times [A]$.
    Intuitively, this means $R$ does not map leaf segments of $G$ corresponding to a common
    number into leaf segments of $H$ corresponding to different parts of the desired partition.
    To this end, suppose towards a contradiction there are pairs $(v, w), (v', w') \in R$ with
    $v \in S(G, i, h)$ and $\ v' \in S(G, i, h')$, and
    $w \in S(H, j, k)$ and $ w' \in S(H, j', k')$, where $j \neq j' \in [m]$.
    Since $j \neq j'$, it holds that $\Gdist_{H}(w,w') \geq 17\ell$,
    whereas $\Gdist_{G}(v,v') \leq 12\ell$.
    Thus $\dis_{(\infty)}(R) \geq 5\ell$, contradiction.
    We conclude that for each $i \in [3m]$ there exists $j \in [m]$ satisfying
    $\phi^{-1}( \{i\} \times [a_i]) \subseteq \{j\} \times [A]$.

    To show that this in fact implies $\mathcal I$ is a YES-instance, we define for
    each $j \in [m]$ the set
    \begin{align*}
        U_j \ceq \big \{ i \in [3m] \mid \phi^{-1}(\{i\} \times [a_i]) \subseteq \{j\} \times [A] \big \}.
    \end{align*}
    Since $\phi$ is a bijection and each $\{i\} \times [a_i]$ is mapped (under $\phi^{-1}$)
    into a single $\{j\} \times [A]$, we have $\sum_{i \in U_j} a_i = A$ for all $j \in [m]$.
    Further, $\frac{A}{4} < a_i < \frac{A}{2}$ implies $\abs{U_j}=3$.
    Consequently, the set $\{U_j\}_{j \in [m]}$
    partitions $[3m]$ into three-element sets summing to $A$. Thus, by mapping each $i \in U_j$ to $j$ one obtains
    the desired function $\sigma \colon [3m] \to [m]$ witnessing that $\mathcal I$ is a YES-instance
    and contradicting our initial assumption. Hence, for NO-instances every correspondence has distortion
    at least $5\ell$, implying $\Mdist{\GH}(G,H) \geq \frac{5}{2}\ell$.\bigskip

    In total, we have shown that, unless $\P = \NP$, there is no polynomial-time
    approximation algorithm for $\GSim{\Mdist{\GH}}$ with additive error guarantee below
    $\frac{1}{2}(5\ell - 4\ell) = \frac{\ell}{2}$.
    Similarly to \cref{subclaim:additive_approx_hardness_max_distortion_trees}, we may choose
    for any $\eps > 0$ some $b > \frac{1}{\eps} (1-\eps)(d+1)$ and $\ell \ceq 2\big\lceil m^b \big\rceil$
    to obtain that approximation with additive error guarantee $\O(n^{1-\eps})$ remains \NP-hard.
    Here, $d$ corresponds to the degree of the polynomial bounding $A$ as guaranteed by strong 
    \NP-hardness of \threepartition{}.
\end{proof}

\section{\texorpdfstring{Proofs Omitted from \cref{section:max_distortion_path}}{Proofs Omitted from Section~\ref{section:max_distortion_path}}}
\label{appendix:max_distortion_path}
\begin{remark}[Self-Reduction for the Maximum Distortion Distance]
    \label{remark:returning_bijection}
    In our setting, the usual self-reduction approach of fixing parts of the solution and querying the decision algorithm on partial instances does not yield an optimisation algorithm that computes the alignment.
    Specifically, this is because fixing a partial alignment cannot readily be encoded as an instance of \GSim{\Ddist{(\infty)}}, \ie{} into a pair of unweighted graphs and a target value $d$.
    This is in contrast to, for example, \textsc{Maximum Matching}, where the decision algorithm can be used to obtain a matching of maximum size by iteratively removing edges from the input graph and querying the decision algorithm.
\end{remark}

\subsection{Additional Notation for Paths}

Recall that for a graph $G$ and vertices $u, v \in V(G)$, a path $P$ from $u$ to $v$ is a sequence of distinct vertices $w_1, \ldots, w_\ell \in V(G)$ with $w_1 = u$, $w_\ell = v$, and $w_iw_{i+1} \in E(G)$ for $i \in [\ell - 1]$.
We say that $P$ is a $uv$-path in $G$ and denote this by $u \arr{P} v$.
We write $|P| \ceq \ell$ for the number of vertices in $P$ and denote the \emph{length} of $P$ by $\len(P) \ceq \ell - 1$.
Note that we allow paths of length zero, consisting of only a single vertex. The \emph{interior} of $P$ is the set of its vertices without the two endpoints.

Slightly abusing notation, we may view a path $P$ as the set of vertices or as the induced subgraph $G[P]$ of $G$ (if $P$ is a shortest path) when convenient. Similarly, we may reverse the orientation of $P$ without explicitly stating this; this will always be clear from context.
For vertices $x, x'$ lying on path $P$, we write $x \arr{P} x'$ for the subpath of $P$ from $x$ to $x'$.
Given paths $u \arr{P} v$ and $v \arr{Q} w$, we denote the concatenation of the two paths by $u \arr{P} v \arr{Q} w$, assuming that this yields a valid path, \ie{} $P \cap Q = \{v\}$.

\subsection{\texorpdfstring{Proof of \cref{theorem:max_distortion_graph_into_path}}{Proof of Theorem~\ref{theorem:max_distortion_graph_into_path}}}
Throughout this subsection, we fix an arbitrary connected $n$-vertex graph $G$, let $d \ceq \diam(G)$ and $\{u, v\}$ be a diametral pair in $G$. Moreover, we denote by $\pi^\ast$ the bijection computed by \cref{alg:graph_to_path} for the fixed diametral pair $\{u, v\}$.
Let $Q$ denote a shortest $uv$-path in $G$.

For a vertex $w \in V(G)$ and $r \in \N$, let $B(w, r) \ceq \{ w' \in V(G) \mid \Gdist_G(w, w') = r \}$.
For $0 \leq i, j \leq d$ with $i + j \geq d$, define the sets $B_{i, j} \ceq B(u, i) \cap B(v, j)$.
Ignoring the fact that some of these sets may be empty, they form a partition of $V(G)$.

For notational convenience, for a vertex $w \in V(G)$, define the $i$-, $j$-, and $a$-\emph{value} of $w$ as $i(w) \ceq i,\ j(w) \ceq j$, and $a(w) \ceq i-j$ for the unique $0 \leq i, j \leq d$ such that $w \in B_{i, j}$.

Let $M \ceq \{-d, \ldots, d\}$ and for each $a \in M$ define $C_a \ceq \bigcup_{i-j=a} B_{i,j}$ as the set of all vertices with the same $a$-value.
Moreover, let $I \ceq \{-d, -d+2, -d+4, \ldots, d-4, d-2, d\}$ and note that $C_a$ is non-empty for each $a \in I$, since it contains the unique vertex of the path $Q$ at distance $i = \frac{1}{2}(a+d)$ from $u$ and distance $d-i$ from $v$. Thus, at least $d+1$ sets are non-empty.

\begin{proof}[Proof of \cref{lemma:graph_to_path_no_large_increase}]
    Let $\{x, x'\}$ be an arbitrary pair of distinct vertices in $G$ and
    $b \ceq a(x),\ b' \ceq a(x')$, in particular $x \in C_b$ and $x' \in C_{b'}$.
    We may \WLOG{} assume that $b \leq b'$.
    We show the existence of at least $d+1$ vertices that do not contribute to an increase in distance between $x$ and $x'$ in $G$ compared to the distance of $\pi^\ast(x)$ and $\pi^\ast(x')$ in $P_n$.

    Let $m \ceq |I \cap [b, b']|$ and
    consider a shortest $xx'$-path $R$ in $G$.
    Vertices in $C_a$ for $a \in I \setminus [b, b']$ may never increase the value $\Gdist_{P_n}(\pi^\ast(x), \pi^\ast(x'))$ compared to $\Gdist_G(x, x')$ by the choice of enumeration of vertices of $\pi^\ast$, since the images of such vertices in $P_n$ either occur before or after the subpath spanned by the vertices $\pi^\ast(x)$ and $\pi^\ast(x')$ in $P_n$.
    This yields at least $d+1-m$ vertices not contributing to an increase in distance under $\pi^\ast$.
    
    Hence, it remains to find an additional $m$ vertices within $C_b \cup \cdots \cup C_{b'}$ that do not increase the distance between $x$ and $x'$. 
    Note that two of these $m$ vertices can always be chosen to be $x$ and $x'$ because their images under $\pi^\ast$ are clearly not between $\pi^\ast(x)$ and $\pi^\ast(x')$ in $P_n$.

    If $b' - b < 3$, then we have already found sufficiently many vertices, since in this case, it holds that $m = |I \cap [b, b']| \leq 2$.
    If not, we use the following claim.

    \begin{claim}
        \label{subclaim:slice_separators}
        Suppose $b' - b \geq 3$.
        Then, for each $b < a < b'-1$, the set $C_a \cup C_{a+1}$ is an $(x, x')$-separator in $G$.
    \end{claim}
    \begin{claimproof}
        Consider an arbitrary $xx'$-path $y_1, \ldots, y_\ell$ in $G$ and consider $w \ceq y_k$, $w' \ceq y_{k+1}$ for some $k \in [\ell-1]$. 
        Let $i, i', j, j'$ such that $w \in B_{i,j}$ and $w' \in B_{i', j'}$.
        Since $ww' \in E(G)$ it is easy to see that
        $|\Gdist_G(w, u) - \Gdist_G(w', u)| \leq 1$ and $|\Gdist_G(w, v) - \Gdist_G(w', v)| \leq 1$, or put differently, $|i - i'| \leq 1$ and $|j-j'| \leq 1$, which implies that 
        $|a(w) - a(w')| = |(i - j) - (i' - j')| \leq 2$.
        Thus, no single edge can \enquote{jump over} two or more consecutive sets $C_a$ at once.
        It follows that each $xx'$-path must have an edge with at least one endpoint in $C_a \cup C_{a+1}$ and hence removing $C_a \cup C_{a+1}$ must disconnect $x$ from $x'$ in $G$.
    \end{claimproof}
 
    Now, using \cref{subclaim:slice_separators}, we find a sequence of $(x, x')$-separators in $G$, 
    where separators $C_a \cup C_{a+1}$ and $C_{a'} \cup C_{a'+1}$ overlap whenever $|a - a'| = 1$.
    As there are $b'-2 - (b+1) + 1 = b' - b - 2$ (counting the elements in $\{b+1, \ldots, b'-2\}$) separators of this form in total, this implies that any $xx'$-path in $G$ contains at least $\lceil \frac{1}{2} (b'-b-2) \rceil$ vertices in $C_{b+1}\cup\cdots \cup C_{b'-1}$.
    Since each such vertex is contained in $C_a$ for $b < a < b'$, it must be mapped between $\pi^\ast(x)$ and $\pi^\ast(x')$ by the definition of $\pi^\ast$.
    Hence, all these vertices were on a shortest path between $x$ and $x'$ in $G$ and remain between $\pi^\ast(x)$ and $\pi^\ast(x')$ after being mapped into $P_n$.
    In particular, this yields at least $2 + \lceil \frac{1}{2} (b' - b - 2) \rceil = \lceil \frac{1}{2}(b' - b)\rceil + 1 \geq |I \cap [b, b']| = m$ vertices.
    
    Combined with the vertices found previously, this yields at least $d+1$ vertices which do not contribute to an increase in the distance between $x$ and $x'$ in $P_n$.
    It follows that $\Gdist_{P_n}(\pi^\ast(x), \pi^\ast(x')) \leq \Gdist_G(x, x') + n-1-d$.
    Since $\{x, x'\}$ was arbitrary, this concludes the proof.
\end{proof}

\begin{proof}[Proof of \cref{lemma:graph_to_path_no_large_decrease}]
    Let $\{x, x'\}$ be an arbitrary pair in $G$ and assume that $\Gdist_G(x, x') \geq n-d+1$, as otherwise there is nothing we need to show since $\Gdist_{P_n}(y, y') \geq 1$ for all distinct $y, y' \in V(P_n)$.
    Let $R$ denote a shortest $xx'$-path.
    We employ the following high-level strategy:
    Since $R$ is long, it contains many vertices and thus must share some vertices with $Q$.
    We then show that we can either find roughly the same number of vertices that must lie strictly between $\pi^\ast(x)$ and $\pi^\ast(x')$ in $P_n$, or that there are sufficiently many vertices that are neither part of $R$ nor $Q$. By a counting argument, this bounds the length of $R$ and yields a contradiction.
 
    We start with two easy claims.
    \begin{claim}
        \label{subclaim:crow_foot}
        Let $w \in B_{i, j}$ where $0 < i, j \leq d$. 
        Then there are (not necessarily distinct) neighbours $w',w'' \in N_G(w)$ such that $i(w') = i(w) - 1$ and $j(w'') = j(w) - 1$.
    \end{claim}
    \begin{claimproof}
    By definition of $B_{i, j}$, there is a shortest
    $wu$-path
    of length $i$, so $w$ has a neighbour at distance $i-1$ from $u$. Analogously, $w$ has a neighbour at distance $j-1$ from $v$.
    \end{claimproof}

    \begin{claim}
        \label{subclaim:consecutive_in_Q}
        There is a shortest $xx'$-path $R$ in $G$ for which $R \cap Q$ is a subpath of $Q$.
    \end{claim}
    \begin{claimproof}
        Consider any shortest $xx'$-path $R$ in $G$.
        By assumption, $\len(R) > n-d$ and hence $|R| > n-d+1$. 
        Together with $|V(G) \setminus Q| = n - (d+1)$, this implies that $R \cap Q \neq \emptyset$.
        Now suppose towards a contradiction that for each such $R$, $R \cap Q$ is not a subpath of $Q$, \ie{}
        $R$ \enquote{leaves} and later \enquote{re-enters} $Q$.
        Let $R$ be one of these paths with the minimal number of entries into $Q$. By assumption $R$ enters $Q$ at least once.
        Then there are vertices $y, y', y''$ occurring in this order on $R$ such that $y, y'' \in Q \cap R$ and $y' \in R \setminus Q$. 
        Let $Q'$ be the $yy''$ subpath of $Q$, $R_1$ the $xy$ subpath of $R$, and $R_2$ the $y''x'$ subpath of $R$. Then, the concatenation $R'$ of $R_1$, $Q'$, and $R_2$ yields a shortest $xx'$-path, since $Q$ is a shortest $uv$-path, and hence $Q'$ is a shortest path between $y$ and $y''$.
        But then $R'$ enters $Q$ fewer times than $R$, contradicting the minimal number of entries of $R$ into $Q$.
    \end{claimproof}

    For the remainder of the proof, fix a shortest $xx'$-path $R$ satisfying \cref{subclaim:consecutive_in_Q}
    and let $Q_1, Q_2, Q_3$ be the subpaths of $Q$ satisfying 
    $Q = u \arr{Q_1} q \arr{Q_2} q' \arr{Q_3} v$
    where $q, q'$ are the vertices of $R \cap Q$ that are closest to $x, x'$, respectively.
    Further, let $X, X'$ be the $xq$- and $x'q'$-subpaths of $R$ and note that the concatenation of $X$, $Q_2$, and $X'$ yields the path $R$ again, in notation 
    $x \arr{X} q \arr{Q_2} q' \arr{X'} x'$.

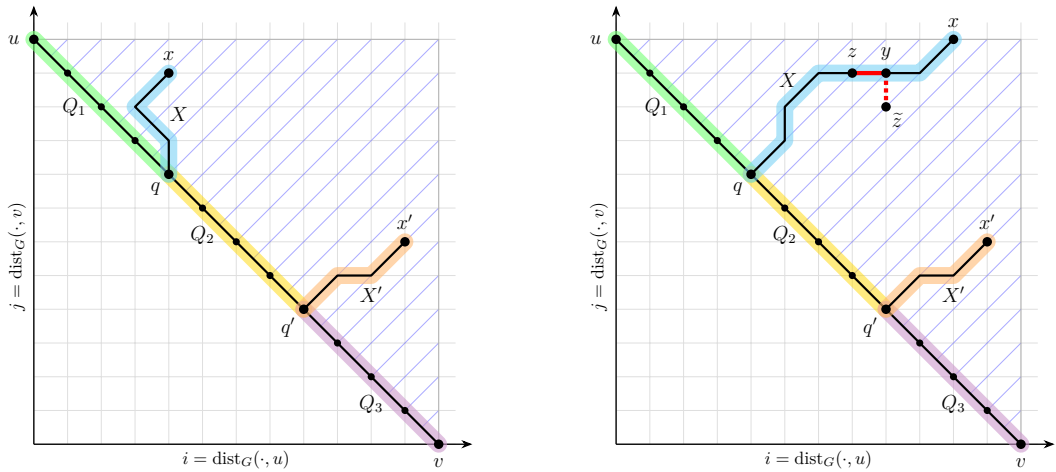
\begin{figure}[b!]
     \captionsetup{singlelinecheck = true, justification=justified}
    \centering
    \begin{subfigure}[t]{0.45\textwidth}
        \centering
        \begin{adjustbox}{width=\linewidth}
            \begin{tikzpicture}[>=Stealth, line join=round, font=\small, scale=0.7]

    \def\d{12}
    \def\lw{1.3pt}
    \def\thinpoint{3pt}
    \def\thickpoint{4pt}

    \coordinate (u) at (0, 12);
    \coordinate (v) at (12, 0);
    \coordinate (x) at (4, 11);
    \coordinate (q) at (4, 8);
    \coordinate (q') at (8, 4);
    \coordinate (x') at (11, 6);

    \coordinate (s) at (6, 6);
    
    \draw[step=1cm, gray!25, thin] (0,0) grid (\d+.5,\d+.5);
    
    \newcommand{\hatchtriangle}[2][blue!55]{
        \foreach \c in {-#2,...,#2}{%
            \ifnum\c<0
                \draw[#1] ({(#2+\c)/2},{(#2-\c)/2}) -- ({#2+\c},{#2});
            \else
                \draw[#1] ({(#2+\c)/2},{(#2-\c)/2}) -- ({#2},{#2-\c});
            \fi
        }%
    }
    \hatchtriangle[line width=0.6pt, blue!55, opacity=0.6]{\d};
    
    \draw[gray!60] (0,\d) -- (\d,\d) -- (\d,0);
    
    \draw[thick,->] (0,0) -- (0,\d+1) ;
    \draw[thick,->] (0,0) -- (\d+1,0) ;
    
    \node[left=10pt, rotate=90] at (0,7.35) {\large $j=\mathrm{dist}_G(\cdot,v)$};
    \node[below] at (6,0) {\large $i=\mathrm{dist}_G(\cdot,u)$};

    \draw[line width=10pt, green!55, opacity=0.6, line cap=round]   (0,\d) -- (4,8);   
    \draw[line width=10pt, yellow!75!orange, opacity=0.5, line cap=round] (4,8) -- (8,4);   
    \draw[line width=10pt, violet!40, opacity=0.6, line cap=round]  (8,4) -- (\d,0);  
    \draw[line width=\lw] (0,\d) -- (\d,0);
    
    \draw[line width=10pt, cyan!45, opacity=0.6, line cap=round, line join=round]
        (q) -- (4, 9) -- (3, 10) -- (4, 11);
    \draw[line width=\lw]
        (q) -- (4, 9) -- (3, 10) -- (4, 11);

    \draw[line width=10pt, orange!55, opacity=0.6, line cap=round, line join=round]
       (q') -- (9, 5) -- (10, 5) -- (11,6);
    \draw[line width=\lw]
       (q') -- (9, 5) -- (10, 5) -- (11,6);

    \foreach \k in {0,...,\d}{
        \fill (\k,{\d-\k}) circle (\thinpoint);
    }

    \node[left=5pt] at (0,\d) {\Large $u$};
    \node[below=5pt] at (\d,0) {\Large $v$};
    \node[above=4pt] at (x) {\Large $x$};
    \node[below left=2pt] at (q) {\Large $q$};
    \node[below left=2pt] at (q') {\Large $q'$};
    \node[above=4pt] at (x') {\Large $x'$};
    
    \foreach \p in {u, v, x, x', q, q'}{
        \fill (\p) circle (\thickpoint);
    }

    \node at (1.2, 10) {\Large $Q_1$};
    \node at (5, 6.2) {\Large $Q_2$};
    \node at (10, 1.2) {\Large $Q_3$};
    \node at (4.3, 9.8) {\Large $X$};
    \node at (10, 4.45) {\Large $X'$};
\end{tikzpicture}
        \end{adjustbox}
        \caption{Visualisation of the case $a(x) \leq a(q) \leq a(q') \leq a(x')$. 
        Specifically, $x$ may only appear in the cone above $q$, $x'$ may only occur in the cone to the right of $q'$.
        }
        \label{subfig:good_case}
    \end{subfigure}
    \hfill
    \begin{subfigure}[t]{0.45\textwidth}
        \centering
        \begin{adjustbox}{width=\linewidth}
            \begin{tikzpicture}[>=Stealth, line join=round, font=\small, scale=0.7]

    \def\d{12}
    \def\lw{1.3pt}
    \def\thinpoint{3pt}
    \def\thickpoint{4pt}

    \coordinate (u) at (0, 12);
    \coordinate (v) at (12, 0);
    \coordinate (x) at (10, 12);
    \coordinate (q) at (4, 8);
    \coordinate (q') at (8, 4);
    \coordinate (x') at (11, 6);
    \coordinate (z) at (7, 11);
    \coordinate (y) at (8, 11);
    \coordinate (ztilde) at (8, 10);

    \coordinate (s) at (6, 6);
    
    \draw[step=1cm, gray!25, thin] (0,0) grid (\d+.5,\d+.5);
    
    \newcommand{\hatchtriangle}[2][blue!55]{
        \foreach \c in {-#2,...,#2}{%
            \ifnum\c<0
                \draw[#1] ({(#2+\c)/2},{(#2-\c)/2}) -- ({#2+\c},{#2});
            \else
                \draw[#1] ({(#2+\c)/2},{(#2-\c)/2}) -- ({#2},{#2-\c});
            \fi
        }%
    }
    \hatchtriangle[line width=0.6pt, blue!55, opacity=0.6]{\d};
    
    \draw[gray!60] (0,\d) -- (\d,\d) -- (\d,0);
    
    \draw[thick,->] (0,0) -- (0,\d+1) ;
    \draw[thick,->] (0,0) -- (\d+1,0) ;
    
    \node[left=10pt, rotate=90] at (0,7.35) {\large $j=\mathrm{dist}_G(\cdot,v)$};
    \node[below] at (6,0) {\large $i=\mathrm{dist}_G(\cdot,u)$};

    \draw[line width=10pt, green!55, opacity=0.6, line cap=round]   (0,\d) -- (4,8);   
    \draw[line width=10pt, yellow!75!orange, opacity=0.5, line cap=round] (4,8) -- (8,4);   
    \draw[line width=10pt, violet!40, opacity=0.6, line cap=round]  (8,4) -- (\d,0);  
    \draw[line width=\lw] (0,\d) -- (\d,0);
    
    \draw[line width=10pt, cyan!45, opacity=0.6, line cap=round, line join=round]
       (4,8) -- (5,9) -- (5,10) -- (6,11) -- (9,11) -- (x);
    \draw[line width=\lw]
       (4,8) -- (5,9) -- (5,10) -- (6,11) -- (9,11) -- (x);

    \draw[line width=2.5pt, red] (z)--(y);
    \draw[line width=2.5pt, red, dotted] (y)--(ztilde);

    \draw[line width=10pt, orange!55, opacity=0.6, line cap=round, line join=round]
       (q') -- (9, 5) -- (10, 5) -- (11,6);
    \draw[line width=\lw]
       (q') -- (9, 5) -- (10, 5) -- (11,6);

    \foreach \k in {0,...,\d}{
        \fill (\k,{\d-\k}) circle (\thinpoint);
    }

    \node[left=5pt] at (0,\d) {\Large $u$};
    \node[below=5pt] at (\d,0) {\Large $v$};
    \node[above=4pt] at (x) {\Large $x$};
    \node[below left=2pt] at (q) {\Large $q$};
    \node[below left=2pt] at (q') {\Large $q'$};
    \node[above=4pt] at (x') {\Large $x'$};
    
    \node[above=6pt] at (z) {\Large $z$};
    \node[above=4pt] at (y) {\Large $y$};
    \node[below right] at (ztilde) {\Large $\widetilde{z}$};
    
    \foreach \p in {u, v, x, x', q, q', y, z, ztilde}{
        \fill (\p) circle (\thickpoint);
    }

    \node at (1.2, 10) {\Large $Q_1$};
    \node at (5, 6.2) {\Large $Q_2$};
    \node at (10, 1.2) {\Large $Q_3$};
    \node at (5.1, 10.8) {\Large $X$};
    \node at (10, 4.45) {\Large $X'$};
\end{tikzpicture}
        \end{adjustbox}
        \caption{The edge $yz$ (marked in bold red) is bad since $a(z) \leq a(q) < a(y)$. The edge $y\wz$ (dotted red, with $j(y) > j(\wz)$) exists by \cref{subclaim:crow_foot}.
        Note that only one of many such scenarios is depicted.
        }
        \label{subfig:bad_case}
    \end{subfigure}
    
    \caption{Visualisation of the relevant relations between $a(x), a(q), a(q'), a(x')$. 
    Each integer point $(i, j)$ corresponds to the set $B_{i, j}$, hence $u$ is at coordinate $(0, d)$ and $v$ is at coordinate $(d, 0)$ with the diametral path (visualised in black) being the diagonal of points $(i, j)$ satisfying $i+j = d$.
    Each vertex of the graph $G$ lies in the upper triangle of the plot.
    The blue diagonal lines correspond to the sets $C_a$.}
    \label{fig:basic_case}
\end{figure}
    We claim that the following ordering relation of the $a$-values of the vertices $x, q, q'$ and $x'$ always holds:
    \begin{equation}
        \label{equation:vertex_constellation}
        a(x) \leq a(q) \leq a(q') \leq a(x').
    \end{equation}
    To this end, we may always assume that $a(q) \leq a(q')$ holds; otherwise, we simply exchange $x$ with $x'$ and $q$ with $q'$.
    For now, suppose that  \cref{equation:vertex_constellation} indeed holds; a visualisation of this case is given in \cref{subfig:good_case}.
    Then, because
    there are $|Q_2|-2$ vertices between $q$ and $q'$ along $Q$ and 
    the $a$-values along $Q$ are strictly increasing (from $u$ to $v$), it follows that $\pi^\ast$ maps at least $|Q_2|-2$ vertices between the bags $C_{a(x)}$ and $C_{a(x')}$, and thus
    $\Gdist_{P_n}(\pi^\ast(x), \pi^\ast(x')) \geq |Q_2| -1$.
    Moreover, since $|R| \leq n - (d+1) + |Q_2|$ it holds that $\Gdist_G(x, x') \leq n-d+|Q_2|-2$.
    Putting these two inequalities together, we obtain
    $\Gdist_G(x, x') \leq \Gdist_{P_n}(\pi^\ast(x), \pi^\ast(x')) + n-1-d$.
    In particular, showing that \cref{equation:vertex_constellation} always holds concludes the proof.

    For the remainder of this proof, we assume that \cref{equation:vertex_constellation} does not hold.
    We say that an edge $yz \in E(R)$ within $X$ or $X'$ is \emph{bad} if one of the following two cases holds:
    \begin{enumerate}
        \item $yz \in E(X)$, $y$ is closer to $x$ than $z$ (in $X$), $a(y) > a(q)$, and $a(z) \leq a(q)$, 
        \item $yz \in E(X')$, $y$ is closer to $x'$ than $z$ (in $X'$), $a(y) < a(q')$, and $a(z) \geq a(q')$.
    \end{enumerate}
    We refer the reader to \cref{subfig:bad_case} (specifically the edge $yz$ depicted therein) for an example of a bad edge.

    Since \cref{equation:vertex_constellation} does not hold, at least one of the inequalities $a(q) < a(x)$ and $a(x') < a(q')$ must hold. Either case implies the existence of a bad edge within $X$ or $X'$ respectively. 
    In the following, we may \WLOG{} assume the bad edge $e = yz$ to be contained within $X$. The other case is symmetric.
    Since $yz$ is bad, it holds that $a(z) < a(y)$, so in particular $j(z) \geq j(y)$ as one of the following three cases must hold:
    
    \begin{itemize}
        \item $\mathmakebox[2.23cm][l]{i(y) = i(z) + 1}$ and $j(y) = j(z)$,
        \item $\mathmakebox[2.23cm][l]{i(y) = i(z) + 1}$ and $j(y) = j(z) - 1$,
        \item $\mathmakebox[2.23cm][l]{i(y) = i(z)}$ and $j(y) = j(z) - 1$.
    \end{itemize}
    
    \noindent A visualisation of the first of these cases is given in
    \cref{subfig:bad_case}.
    
    By \cref{subclaim:crow_foot}, in any of these cases $y$ needs to be adjacent to a vertex 
    $\wz$ with $j(\wz) < j(y)$.
    By definition $\Gdist_G(\wz, v) = j(\wz)$, so there is a shortest path from $\wz$ to $v$ of length $j(\wz)$.
    Let $s$ denote the vertex on this path that is contained in $R \cup Q$ and closest to $\wz$, call the resulting $\wz s$-subpath $S$.
    In the following, we perform a case distinction based on where $s$ may lie in $R \cup Q$ precisely. Each case assumption will ultimately lead to a contradiction, which in turn implies that $e$ cannot have been a bad edge.
    The cases correspond to the five subpaths that $R$ and $Q$ have been decomposed into; only the last of these cases (corresponding to $X$) is split into two subcases based on where $s$ lies relative to $y$. We exclude $q$ or $q'$ in some of the cases to obtain a partition of $R \cup Q$, so each vertex falls into exactly one of the six cases. 
    \cref{fig:case_distinction_cases} provides a schematic visualisation of these cases.
    
    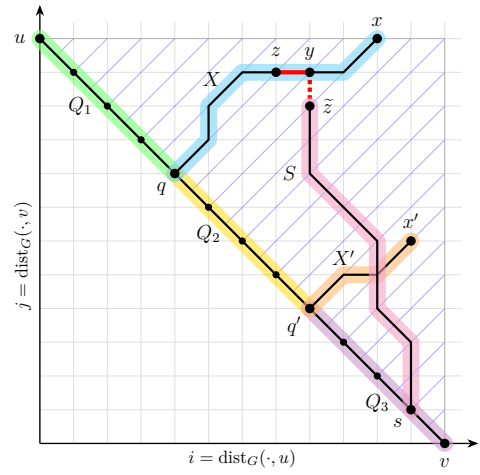
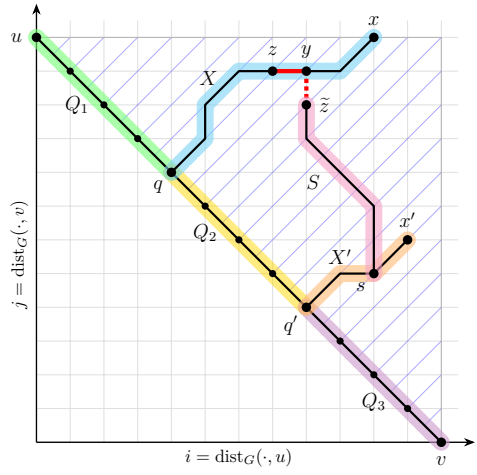
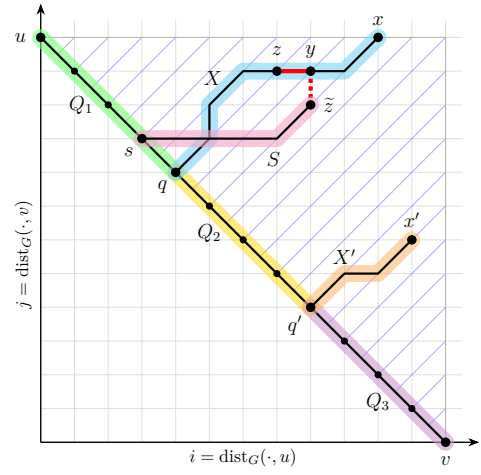
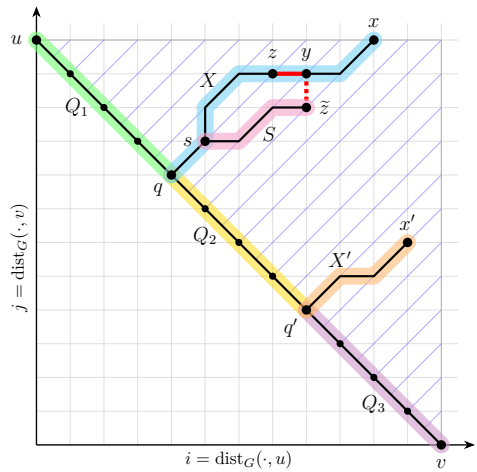
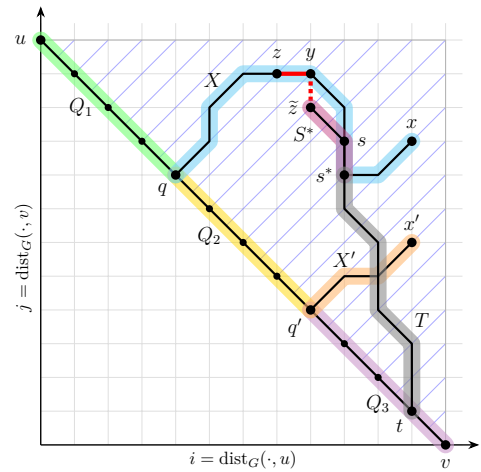
\begin{figure}[p]
    \captionsetup{singlelinecheck = true, justification=justified}
    \centering
    
    \begin{subfigure}[t]{0.45\textwidth}
        \centering
        \begin{adjustbox}{width=\linewidth}
            \begin{tikzpicture}[>=Stealth, line join=round, font=\small, scale=0.7]

    \def\d{12}
    \def\lw{1.3pt}
    \def\thinpoint{3pt}
    \def\thickpoint{4pt}

    \coordinate (u) at (0, 12);
    \coordinate (v) at (12, 0);
    \coordinate (x) at (10, 12);
    \coordinate (q) at (4, 8);
    \coordinate (q') at (8, 4);
    \coordinate (x') at (11, 6);
    \coordinate (z) at (7, 11);
    \coordinate (y) at (8, 11);
    \coordinate (ztilde) at (8, 10);

    \coordinate (s) at (6, 6);
    
    \draw[step=1cm, gray!25, thin] (0,0) grid (\d+.5,\d+.5);
    
    \newcommand{\hatchtriangle}[2][blue!55]{
        \foreach \c in {-#2,...,#2}{%
            \ifnum\c<0
                \draw[#1] ({(#2+\c)/2},{(#2-\c)/2}) -- ({#2+\c},{#2});
            \else
                \draw[#1] ({(#2+\c)/2},{(#2-\c)/2}) -- ({#2},{#2-\c});
            \fi
        }%
    }
    \hatchtriangle[line width=0.6pt, blue!55, opacity=0.6]{\d};
    
    \draw[gray!60] (0,\d) -- (\d,\d) -- (\d,0);
    
    \draw[thick,->] (0,0) -- (0,\d+1) ;
    \draw[thick,->] (0,0) -- (\d+1,0) ;
    
    \node[left=10pt, rotate=90] at (0,7.35) {\large $j=\mathrm{dist}_G(\cdot,v)$};
    \node[below] at (6,0) {\large $i=\mathrm{dist}_G(\cdot,u)$};

    \draw[line width=10pt, green!55, opacity=0.6, line cap=round]   (0,\d) -- (4,8);   
    \draw[line width=10pt, yellow!75!orange, opacity=0.5, line cap=round] (4,8) -- (8,4);   
    \draw[line width=10pt, violet!40, opacity=0.6, line cap=round]  (8,4) -- (\d,0);  
    \draw[line width=\lw] (0,\d) -- (\d,0);
    
    \draw[line width=10pt, cyan!45, opacity=0.6, line cap=round, line join=round]
       (4,8) -- (5,9) -- (5,10) -- (6,11) -- (9,11) -- (x);
    \draw[line width=\lw]
       (4,8) -- (5,9) -- (5,10) -- (6,11) -- (9,11) -- (x);

    \draw[line width=2.5pt, red] (z)--(y);
    \draw[line width=2.5pt, red, dotted] (y)--(ztilde);

    \draw[line width=10pt, orange!55, opacity=0.6, line cap=round, line join=round]
       (q') -- (9, 5) -- (10, 5) -- (11,6);
    \draw[line width=\lw]
       (q') -- (9, 5) -- (10, 5) -- (11,6);
    
    \draw[line width=10pt, magenta!45, opacity=0.6, line cap=round, line join=round]
       (ztilde) -- (7,9) -- (7,7) -- (6,6);
    \draw[line width=\lw]
       (ztilde) -- (7,9) -- (7,7) -- (6,6);
    
    \foreach \k in {0,...,\d}{
        \fill (\k,{\d-\k}) circle (\thinpoint);
    }

    \node[left=5pt] at (0,\d) {\Large $u$};
    \node[below=5pt] at (\d,0) {\Large $v$};
    \node[above=4pt] at (x) {\Large $x$};
    \node[below left=2pt] at (q) {\Large $q$};
    \node[below left=2pt] at (q') {\Large $q'$};
    \node[above=4pt] at (x') {\Large $x'$};
    
    \node[above=6pt] at (z) {\Large $z$};
    \node[above=4pt] at (y) {\Large $y$};
    \node[below right] at (ztilde) {\Large $\widetilde{z}$};
    
    \node[below left=2pt] at (s) {\Large $s$};

    \foreach \p in {u, v, x, x', q, q', y, z, ztilde, s}{
        \fill (\p) circle (\thickpoint);
    }

    \node at (1.2, 10) {\Large $Q_1$};
    \node at (5, 6.2) {\Large $Q_2$};
    \node at (10, 1.2) {\Large $Q_3$};
    \node at (7.6, 8.1) {\Large $S$};
    \node at (5.1, 10.8) {\Large $X$};
    \node at (10, 4.45) {\Large $X'$};
\end{tikzpicture}
        \end{adjustbox}
        \caption{\textbf{\textsf{\textcolor{lipicsGray}{Case 1}}} ($s \in Q_2$)}
        \label{subfig:case_1}
    \end{subfigure}
    \hfill
    \begin{subfigure}[t]{0.45\textwidth}
        \centering
        \begin{adjustbox}{width=\linewidth}
            \begin{tikzpicture}[>=Stealth, line join=round, font=\small, scale=0.7]

    \def\d{12}
    \def\lw{1.3pt}
    \def\thinpoint{3pt}
    \def\thickpoint{4pt}

    \coordinate (u) at (0, 12);
    \coordinate (v) at (12, 0);
    \coordinate (x) at (10, 12);
    \coordinate (q) at (4, 8);
    \coordinate (q') at (8, 4);
    \coordinate (x') at (11, 6);
    \coordinate (z) at (7, 11);
    \coordinate (y) at (8, 11);
    \coordinate (ztilde) at (8, 10);

    \coordinate (s) at (11, 1);
    
    \draw[step=1cm, gray!25, thin] (0,0) grid (\d+.5,\d+.5);
    
    \newcommand{\hatchtriangle}[2][blue!55]{
        \foreach \c in {-#2,...,#2}{%
            \ifnum\c<0
                \draw[#1] ({(#2+\c)/2},{(#2-\c)/2}) -- ({#2+\c},{#2});
            \else
                \draw[#1] ({(#2+\c)/2},{(#2-\c)/2}) -- ({#2},{#2-\c});
            \fi
        }%
    }
    \hatchtriangle[line width=0.6pt, blue!55, opacity=0.6]{\d};
    
    \draw[gray!60] (0,\d) -- (\d,\d) -- (\d,0);
    
    \draw[thick,->] (0,0) -- (0,\d+1) ;
    \draw[thick,->] (0,0) -- (\d+1,0) ;
    
    \node[left=10pt, rotate=90] at (0,7.35) {\large $j=\mathrm{dist}_G(\cdot,v)$};
    \node[below] at (6,0) {\large $i=\mathrm{dist}_G(\cdot,u)$};

    \draw[line width=10pt, green!55, opacity=0.6, line cap=round]   (0,\d) -- (4,8);   
    \draw[line width=10pt, yellow!75!orange, opacity=0.5, line cap=round] (4,8) -- (8,4);   
    \draw[line width=10pt, violet!40, opacity=0.6, line cap=round]  (8,4) -- (\d,0);  
    \draw[line width=\lw] (0,\d) -- (\d,0);
    
    \draw[line width=10pt, cyan!45, opacity=0.6, line cap=round, line join=round]
       (4,8) -- (5,9) -- (5,10) -- (6,11) -- (9,11) -- (x);
    \draw[line width=\lw]
       (4,8) -- (5,9) -- (5,10) -- (6,11) -- (9,11) -- (x);

    \draw[line width=2.5pt, red] (z)--(y);
    \draw[line width=2.5pt, red, dotted] (y)--(ztilde);

    \draw[line width=10pt, orange!55, opacity=0.6, line cap=round, line join=round]
       (q') -- (9, 5) -- (10, 5) -- (11,6);
    \draw[line width=\lw]
       (q') -- (9, 5) -- (10, 5) -- (11,6);
    
    \draw[line width=10pt, magenta!45, opacity=0.6, line cap=round, line join=round]
       (ztilde) -- (8, 8) -- (10, 6) -- (10,4) -- (11, 3) -- (11, 1);
    \draw[line width=\lw]
       (ztilde) -- (8, 8) -- (10, 6) -- (10,4) -- (11, 3) -- (11, 1);
    
    \foreach \k in {0,...,\d}{
        \fill (\k,{\d-\k}) circle (\thinpoint);
    }

    \node[left=5pt] at (0,\d) {\Large $u$};
    \node[below=5pt] at (\d,0) {\Large $v$};
    \node[above=4pt] at (x) {\Large $x$};
    \node[below left=2pt] at (q) {\Large $q$};
    \node[below left=2pt] at (q') {\Large $q'$};
    \node[above=4pt] at (x') {\Large $x'$};
    
    \node[above=6pt] at (z) {\Large $z$};
    \node[above=4pt] at (y) {\Large $y$};
    \node[right=4pt] at (ztilde) {\Large $\widetilde{z}$};
    
    \node[below left=2pt] at (s) {\Large $s$};

    \foreach \p in {u, v, x, x', q, q', y, z, ztilde, s}{
        \fill (\p) circle (\thickpoint);
    }

    \node at (1.2, 10) {\Large $Q_1$};
    \node at (5, 6.2) {\Large $Q_2$};
    \node at (10, 1.2) {\Large $Q_3$};
    \node at (7.4, 8.0) {\Large $S$};
    \node at (5.1, 10.8) {\Large $X$};
    \node at (9, 5.5) {\Large $X'$};
\end{tikzpicture}
        \end{adjustbox}
        \caption{\textbf{\textsf{\textcolor{lipicsGray}{Case 2}}} ($s \in Q_3 \setminus \{q'\}$)}
        \label{subfig:case_2}
    \end{subfigure}
    \vspace{0.1cm}
    
    \begin{subfigure}[t]{0.45\textwidth}
        \centering
        \begin{adjustbox}{width=\linewidth}
            \begin{tikzpicture}[>=Stealth, line join=round, font=\small, scale=0.7]

    \def\d{12}
    \def\lw{1.3pt}
    \def\thinpoint{3pt}
    \def\thickpoint{4pt}

    \coordinate (u) at (0, 12);
    \coordinate (v) at (12, 0);
    \coordinate (x) at (10, 12);
    \coordinate (q) at (4, 8);
    \coordinate (q') at (8, 4);
    \coordinate (x') at (11, 6);
    \coordinate (z) at (7, 11);
    \coordinate (y) at (8, 11);
    \coordinate (ztilde) at (8, 10);

    \coordinate (s) at (10, 5);
    
    \draw[step=1cm, gray!25, thin] (0,0) grid (\d+.5,\d+.5);
    
    \newcommand{\hatchtriangle}[2][blue!55]{
        \foreach \c in {-#2,...,#2}{%
            \ifnum\c<0
                \draw[#1] ({(#2+\c)/2},{(#2-\c)/2}) -- ({#2+\c},{#2});
            \else
                \draw[#1] ({(#2+\c)/2},{(#2-\c)/2}) -- ({#2},{#2-\c});
            \fi
        }%
    }
    \hatchtriangle[line width=0.6pt, blue!55, opacity=0.6]{\d};
    
    \draw[gray!60] (0,\d) -- (\d,\d) -- (\d,0);
    
    \draw[thick,->] (0,0) -- (0,\d+1) ;
    \draw[thick,->] (0,0) -- (\d+1,0) ;
    
    \node[left=10pt, rotate=90] at (0,7.35) {\large $j=\mathrm{dist}_G(\cdot,v)$};
    \node[below] at (6,0) {\large $i=\mathrm{dist}_G(\cdot,u)$};

    \draw[line width=10pt, green!55, opacity=0.6, line cap=round]   (0,\d) -- (4,8);   
    \draw[line width=10pt, yellow!75!orange, opacity=0.5, line cap=round] (4,8) -- (8,4);   
    \draw[line width=10pt, violet!40, opacity=0.6, line cap=round]  (8,4) -- (\d,0);  
    \draw[line width=\lw] (0,\d) -- (\d,0);
    
    \draw[line width=10pt, cyan!45, opacity=0.6, line cap=round, line join=round]
       (4,8) -- (5,9) -- (5,10) -- (6,11) -- (9,11) -- (x);
    \draw[line width=\lw]
       (4,8) -- (5,9) -- (5,10) -- (6,11) -- (9,11) -- (x);

    \draw[line width=2.5pt, red] (z)--(y);
    \draw[line width=2.5pt, red, dotted] (y)--(ztilde);

    \draw[line width=10pt, orange!55, opacity=0.6, line cap=round, line join=round]
       (q') -- (9, 5) -- (10, 5) -- (11,6);
    \draw[line width=\lw]
       (q') -- (9, 5) -- (10, 5) -- (11,6);
    
    \draw[line width=10pt, magenta!45, opacity=0.6, line cap=round, line join=round]
       (ztilde) -- (8, 9) -- (10, 7) -- (10,5);
    \draw[line width=\lw]
       (ztilde) -- (8, 9) -- (10, 7) -- (10,5);
    
    \foreach \k in {0,...,\d}{
        \fill (\k,{\d-\k}) circle (\thinpoint);
    }

    \node[left=5pt] at (0,\d) {\Large $u$};
    \node[below=5pt] at (\d,0) {\Large $v$};
    \node[above=4pt] at (x) {\Large $x$};
    \node[below left=2pt] at (q) {\Large $q$};
    \node[below left=2pt] at (q') {\Large $q'$};
    \node[above=4pt] at (x') {\Large $x'$};
    
    \node[above=6pt] at (z) {\Large $z$};
    \node[above=4pt] at (y) {\Large $y$};
    \node[right=4pt] at (ztilde) {\Large $\widetilde{z}$};
    
    \node[below left=2pt] at (s) {\Large $s$};

    \foreach \p in {u, v, x, x', q, q', y, z, ztilde, s}{
        \fill (\p) circle (\thickpoint);
    }

    \node at (1.2, 10) {\Large $Q_1$};
    \node at (5, 6.2) {\Large $Q_2$};
    \node at (10, 1.2) {\Large $Q_3$};
    \node at (8.2, 7.7) {\Large $S$};
    \node at (5.1, 10.8) {\Large $X$};
    \node at (9, 5.5) {\Large $X'$};
\end{tikzpicture}
        \end{adjustbox}
        \caption{\textbf{\textsf{\textcolor{lipicsGray}{Case 3}}} ($s \in X' \setminus \{q'\}$)}
        \label{subfig:case_3}
    \end{subfigure}
    \hfill
    \begin{subfigure}[t]{0.45\textwidth}
        \centering
        \begin{adjustbox}{width=\linewidth}
            \begin{tikzpicture}[>=Stealth, line join=round, font=\small, scale=0.7]

    \def\d{12}
    \def\lw{1.3pt}
    \def\thinpoint{3pt}
    \def\thickpoint{4pt}

    \coordinate (u) at (0, 12);
    \coordinate (v) at (12, 0);
    \coordinate (x) at (10, 12);
    \coordinate (q) at (4, 8);
    \coordinate (q') at (8, 4);
    \coordinate (x') at (11, 6);
    \coordinate (z) at (7, 11);
    \coordinate (y) at (8, 11);
    \coordinate (ztilde) at (8, 10);

    \coordinate (s) at (3, 9);
    
    \draw[step=1cm, gray!25, thin] (0,0) grid (\d+.5,\d+.5);
    
    \newcommand{\hatchtriangle}[2][blue!55]{
        \foreach \c in {-#2,...,#2}{%
            \ifnum\c<0
                \draw[#1] ({(#2+\c)/2},{(#2-\c)/2}) -- ({#2+\c},{#2});
            \else
                \draw[#1] ({(#2+\c)/2},{(#2-\c)/2}) -- ({#2},{#2-\c});
            \fi
        }%
    }
    \hatchtriangle[line width=0.6pt, blue!55, opacity=0.6]{\d};
    
    \draw[gray!60] (0,\d) -- (\d,\d) -- (\d,0);
    
    \draw[thick,->] (0,0) -- (0,\d+1) ;
    \draw[thick,->] (0,0) -- (\d+1,0) ;
    
    \node[left=10pt, rotate=90] at (0,7.35) {\large $j=\mathrm{dist}_G(\cdot,v)$};
    \node[below] at (6,0) {\large $i=\mathrm{dist}_G(\cdot,u)$};

    \draw[line width=10pt, green!55, opacity=0.6, line cap=round]   (0,\d) -- (4,8);   
    \draw[line width=10pt, yellow!75!orange, opacity=0.5, line cap=round] (4,8) -- (8,4);   
    \draw[line width=10pt, violet!40, opacity=0.6, line cap=round]  (8,4) -- (\d,0);  
    \draw[line width=\lw] (0,\d) -- (\d,0);
    
    \draw[line width=10pt, cyan!45, opacity=0.6, line cap=round, line join=round]
       (4,8) -- (5,9) -- (5,10) -- (6,11) -- (9,11) -- (x);
    \draw[line width=\lw]
       (4,8) -- (5,9) -- (5,10) -- (6,11) -- (9,11) -- (x);

    \draw[line width=2.5pt, red] (z)--(y);
    \draw[line width=2.5pt, red, dotted] (y)--(ztilde);

    \draw[line width=10pt, orange!55, opacity=0.6, line cap=round, line join=round]
       (q') -- (9, 5) -- (10, 5) -- (11,6);
    \draw[line width=\lw]
       (q') -- (9, 5) -- (10, 5) -- (11,6);
    
    \draw[line width=10pt, magenta!45, opacity=0.6, line cap=round, line join=round]
       (ztilde) -- (7, 9) -- (s);
    \draw[line width=\lw]
       (ztilde) -- (7, 9) -- (s);
    
    \foreach \k in {0,...,\d}{
        \fill (\k,{\d-\k}) circle (\thinpoint);
    }

    \node[left=5pt] at (0,\d) {\Large $u$};
    \node[below=5pt] at (\d,0) {\Large $v$};
    \node[above=4pt] at (x) {\Large $x$};
    \node[below left=2pt] at (q) {\Large $q$};
    \node[below left=2pt] at (q') {\Large $q'$};
    \node[above=4pt] at (x') {\Large $x'$};
    
    \node[above=6pt] at (z) {\Large $z$};
    \node[above=4pt] at (y) {\Large $y$};
    \node[right=4pt] at (ztilde) {\Large $\widetilde{z}$};
    
    \node[below left=2pt] at (s) {\Large $s$};

    \foreach \p in {u, v, x, x', q, q', y, z, ztilde, s}{
        \fill (\p) circle (\thickpoint);
    }

    \node at (1.2, 10) {\Large $Q_1$};
    \node at (5, 6.2) {\Large $Q_2$};
    \node at (10, 1.2) {\Large $Q_3$};
    \node at (6.9, 8.4) {\Large $S$};
    \node at (5.1, 10.8) {\Large $X$};
    \node at (9, 5.5) {\Large $X'$};
\end{tikzpicture}
        \end{adjustbox}
        \caption{\textbf{\textsf{\textcolor{lipicsGray}{Case 4}}} ($s \in Q_1 \setminus \{q\}$). Note that the path $S$ indicated above is not valid because the $j$-value does not decrease in each edge.}
        \label{subfig:case_4}
    \end{subfigure}

    \vspace{0.1cm}
    \begin{subfigure}[t]{0.45\textwidth}
        \centering
        \begin{adjustbox}{width=\linewidth}
            \begin{tikzpicture}[>=Stealth, line join=round, font=\small, scale=0.7]

    \def\d{12}
    \def\lw{1.3pt}
    \def\thinpoint{3pt}
    \def\thickpoint{4pt}

    \coordinate (u) at (0, 12);
    \coordinate (v) at (12, 0);
    \coordinate (x) at (10, 12);
    \coordinate (q) at (4, 8);
    \coordinate (q') at (8, 4);
    \coordinate (x') at (11, 6);
    \coordinate (z) at (7, 11);
    \coordinate (y) at (8, 11);
    \coordinate (ztilde) at (8, 10);

    \coordinate (s) at (5, 9);
    
    \draw[step=1cm, gray!25, thin] (0,0) grid (\d+.5,\d+.5);
    
    \newcommand{\hatchtriangle}[2][blue!55]{
        \foreach \c in {-#2,...,#2}{%
            \ifnum\c<0
                \draw[#1] ({(#2+\c)/2},{(#2-\c)/2}) -- ({#2+\c},{#2});
            \else
                \draw[#1] ({(#2+\c)/2},{(#2-\c)/2}) -- ({#2},{#2-\c});
            \fi
        }%
    }
    \hatchtriangle[line width=0.6pt, blue!55, opacity=0.6]{\d};
    
    \draw[gray!60] (0,\d) -- (\d,\d) -- (\d,0);
    
    \draw[thick,->] (0,0) -- (0,\d+1) ;
    \draw[thick,->] (0,0) -- (\d+1,0) ;
    
    \node[left=10pt, rotate=90] at (0,7.35) {\large $j=\mathrm{dist}_G(\cdot,v)$};
    \node[below] at (6,0) {\large $i=\mathrm{dist}_G(\cdot,u)$};

    \draw[line width=10pt, green!55, opacity=0.6, line cap=round]   (0,\d) -- (4,8);   
    \draw[line width=10pt, yellow!75!orange, opacity=0.5, line cap=round] (4,8) -- (8,4);   
    \draw[line width=10pt, violet!40, opacity=0.6, line cap=round]  (8,4) -- (\d,0);  
    \draw[line width=\lw] (0,\d) -- (\d,0);
    
    \draw[line width=10pt, cyan!45, opacity=0.6, line cap=round, line join=round]
       (4,8) -- (5,9) -- (5,10) -- (6,11) -- (9,11) -- (x);
    \draw[line width=\lw]
       (4,8) -- (5,9) -- (5,10) -- (6,11) -- (9,11) -- (x);

    \draw[line width=2.5pt, red] (z)--(y);
    \draw[line width=2.5pt, red, dotted] (y)--(ztilde);

    \draw[line width=10pt, orange!55, opacity=0.6, line cap=round, line join=round]
       (q') -- (9, 5) -- (10, 5) -- (11,6);
    \draw[line width=\lw]
       (q') -- (9, 5) -- (10, 5) -- (11,6);
    
    \draw[line width=10pt, magenta!45, opacity=0.6, line cap=round, line join=round]
       (ztilde) -- (7, 10) -- (6, 9) -- (s);
    \draw[line width=\lw]
       (ztilde) -- (7, 10) -- (6, 9) -- (s);
    
    \foreach \k in {0,...,\d}{
        \fill (\k,{\d-\k}) circle (\thinpoint);
    }

    \node[left=5pt] at (0,\d) {\Large $u$};
    \node[below=5pt] at (\d,0) {\Large $v$};
    \node[above=4pt] at (x) {\Large $x$};
    \node[below left=2pt] at (q) {\Large $q$};
    \node[below left=2pt] at (q') {\Large $q'$};
    \node[above=4pt] at (x') {\Large $x'$};
    
    \node[above=6pt] at (z) {\Large $z$};
    \node[above=4pt] at (y) {\Large $y$};
    \node[right=4pt] at (ztilde) {\Large $\widetilde{z}$};
    
    \node[above=3pt, left=4pt] at (s) {\Large $s$};

    \foreach \p in {u, v, x, x', q, q', y, z, ztilde, s}{
        \fill (\p) circle (\thickpoint);
    }

    \node at (1.2, 10) {\Large $Q_1$};
    \node at (5, 6.2) {\Large $Q_2$};
    \node at (10, 1.2) {\Large $Q_3$};
    \node at (6.9, 9.2) {\Large $S$};
    \node at (5.1, 10.8) {\Large $X$};
    \node at (9, 5.5) {\Large $X'$};
\end{tikzpicture}
        \end{adjustbox}
        \caption{\textbf{\textsf{\textcolor{lipicsGray}{Case 5}}} ($s \in X \setminus \{q\}$, $s$ between $q$ and $y$). Note that the path $S$ indicated above is not valid because the $j$-value does not decrease in each edge.}
        \label{subfig:case_5}
    \end{subfigure}
    \hfill
    \begin{subfigure}[t]{0.45\textwidth}
        \centering
        \begin{adjustbox}{width=\linewidth}
            \begin{tikzpicture}[>=Stealth, line join=round, font=\small, scale=0.7]

    \def\d{12}
    \def\lw{1.3pt}
    \def\thinpoint{3pt}
    \def\thickpoint{4pt}

    \coordinate (u) at (0, 12);
    \coordinate (v) at (12, 0);
    \coordinate (x) at (11, 9);
    \coordinate (q) at (4, 8);
    \coordinate (q') at (8, 4);
    \coordinate (x') at (11, 6);
    \coordinate (z) at (7, 11);
    \coordinate (y) at (8, 11);
    \coordinate (ztilde) at (8, 10);
    \coordinate (sast) at (9, 8);
    \coordinate (t) at (11, 1);

    \coordinate (s) at (9, 9);
    
    \draw[step=1cm, gray!25, thin] (0,0) grid (\d+.5,\d+.5);
    
    \newcommand{\hatchtriangle}[2][blue!55]{
        \foreach \c in {-#2,...,#2}{%
            \ifnum\c<0
                \draw[#1] ({(#2+\c)/2},{(#2-\c)/2}) -- ({#2+\c},{#2});
            \else
                \draw[#1] ({(#2+\c)/2},{(#2-\c)/2}) -- ({#2},{#2-\c});
            \fi
        }%
    }
    \hatchtriangle[line width=0.6pt, blue!55, opacity=0.6]{\d};
    
    \draw[gray!60] (0,\d) -- (\d,\d) -- (\d,0);
    
    \draw[thick,->] (0,0) -- (0,\d+1) ;
    \draw[thick,->] (0,0) -- (\d+1,0) ;
    
    \node[left=10pt, rotate=90] at (0,7.35) {\large $j=\mathrm{dist}_G(\cdot,v)$};
    \node[below] at (6,0) {\large $i=\mathrm{dist}_G(\cdot,u)$};

    \draw[line width=10pt, green!55, opacity=0.6, line cap=round]   (0,\d) -- (4,8);   
    \draw[line width=10pt, yellow!75!orange, opacity=0.5, line cap=round] 
    (4,8) -- (8,4);   
    \draw[line width=10pt, violet!40, opacity=0.6, line cap=round]  (8,4) -- (\d,0);  
    \draw[line width=\lw] (0,\d) -- (\d,0);
    
    \draw[line width=10pt, cyan!45, opacity=0.6, line cap=round, line join=round]
       (4,8) -- (5,9) -- (5,10) -- (6,11) -- (8,11) -- (9, 10) -- (9, 8) -- (10, 8) -- (x);
    \draw[line width=\lw]
       (4,8) -- (5,9) -- (5,10) -- (6,11) -- (8,11) -- (9, 10) -- (9, 8) -- (10, 8) -- (x);

    \draw[line width=2.5pt, red] (z)--(y);
    \draw[line width=2.5pt, red, dotted] (y)--(ztilde);

    \draw[line width=10pt, orange!55, opacity=0.6, line cap=round, line join=round]
       (q') -- (9, 5) -- (10, 5) -- (11,6);
    \draw[line width=\lw]
       (q') -- (9, 5) -- (10, 5) -- (11,6);
    
    \draw[line width=10pt, magenta!80!black, opacity=0.5, line cap=round] 
       (ztilde) -- (9, 9) -- (9, 8);
    \draw[line width=\lw]
       (ztilde) -- (9, 9) -- (9, 8);

    \draw[line width=10pt, black!45, opacity=0.6, line cap=round, line join=round]
       (sast) -- (9, 7) -- (10, 6) -- (10, 4) -- (11, 3) -- (t);
    \draw[line width=\lw]
       (sast) -- (9, 7) -- (10, 6) -- (10, 4) -- (11, 3) -- (t);

    \foreach \k in {0,...,\d}{
        \fill (\k,{\d-\k}) circle (\thinpoint);
    }

    \node[left=5pt] at (0,\d) {\Large $u$};
    \node[below=5pt] at (\d,0) {\Large $v$};
    \node[above=4pt] at (x) {\Large $x$};
    \node[below left=2pt] at (q) {\Large $q$};
    \node[below left=2pt] at (q') {\Large $q'$};
    \node[above=4pt] at (x') {\Large $x'$};
    \node[below left=2pt] at (t) {\Large $t$};
    \node[right=4pt] at (s) {\Large $s$};
    \node[left=3pt] at (sast) {\Large $s^\ast$};
    
    \node[above=6pt] at (z) {\Large $z$};
    \node[above=4pt] at (y) {\Large $y$};
    \node[left=4pt] at (ztilde) {\Large $\widetilde{z}$};
    

    \foreach \p in {u, v, x, x', q, q', y, z, ztilde, s, sast, t}{
        \fill (\p) circle (\thickpoint);
    }

    \node at (1.2, 10) {\Large $Q_1$};
    \node at (5, 6.2) {\Large $Q_2$};
    \node at (10, 1.2) {\Large $Q_3$};
    \node at (7.8, 9.2) {\Large $S^\ast$};
    \node at (5.1, 10.8) {\Large $X$};
    \node at (9, 5.5) {\Large $X'$};
    \node at (11.3, 3.7) {\Large $T$};
\end{tikzpicture}
        \end{adjustbox}
        \caption{\textbf{\textsf{\textcolor{lipicsGray}{Case 6}}} ($s \in X \setminus \{q\}$, $s$ between $y$ and $x$). Note that in this case $S$ is the $\wz s$-subpath of $S^\ast$.}
        \label{subfig:case_6}
    \end{subfigure}
    
    \caption{The six cases of where $s$ may lie in our case distinction. We use the same conceptual drawing as in \cref{fig:basic_case}.}
    \label{fig:case_distinction_cases}
\end{figure}

    \proofsubparagraph{Case 1:} $s \in Q_2$ (see \cref{subfig:case_1})\\
    Moving along $S$, the $j$-value of consecutive vertices decreases by one for each consecutive edge.
    Thus for the extended path $y\wz \arr{S} s$, in each step the $j$-value decreases as well, whereas in the path $yz \arr{X} q \arr{Q_2} s = y \arr{R} s$,
    for the first edge (namely $yz$), the $j$-value does not decrease since $yz$ is bad. Hence, the latter path must be at least one edge longer than the former path.
    However, $y \arr{R} s$ must be a shortest path since $R$ is a shortest path between $x$ and $x'$, which yields a contradiction.
    
    \proofsubparagraph{Case 2:}
    $s \in Q_3 \setminus \{q'\}$ (see \cref{subfig:case_2})\\
    Again, along the path $y\wz \arr{S} s$, the $j$-value decreases in each step. In the path $yz \arr{X} q \arr{Q_2} q' \arr{Q_3} s$, this is not the case for the edge $yz$ since it is bad. However, we cannot derive a contradiction as easily as in the previous case.
    
    We claim that $j(\wz) \geq j(q) - 2$.
    To see this, note that $q$ has $(i(q), j(q))$ with $i(q) + j(q)$ minimal among all possible pairs $(i, j)$ such that $i-j = a(q)$. Moreover, $a(z) \leq a(q)$ and $a(q) < a(y)$ since $yz$ is bad.
    The following claim uses the minimality of the pair $(i(q), j(q))$ and allows us to obtain a lower bound on the $j$-value of $z$.

    \begin{claim}
        \label{subclaim:jf_claim}
        Any vertex $w \in V(G)$ with $a(w) \leq a(q)$ satisfies $j(w) \geq j(q)$.
    \end{claim}
    \begin{claimproof}
        Since $a(w) \leq a(q)$, it holds that (i) $i(w) - i(q) \leq j(w) - j(q)$ by definition of the $a$-value. Moreover, by minimality of the pair $(i(q), j(q))$, we have (ii) $i(w) + j(w) \geq i(q) + j(q)$.
        Putting these inequalities together yields
        \begin{align*}
            j(w) + i(q) &= i(w) + j(w) - i(w) + i(q)   \\
            &\geq i(q) + j(q) - i(w) + i(q)  \tag*{(by (ii))} \\
            &\geq i(q) + j(q) - j(w) + j(q). \tag*{(by (i))}
        \end{align*}
        Subtracting $i(q)$ and adding $j(w)$ on both sides yields $2j(w) \geq 2j(q)$.
    \end{claimproof}
    Using \cref{subclaim:jf_claim} and $a(z) \leq a(q)$, we get $j(z) \geq j(q)$. Since $z$ and $\wz$ have distance at most two in $G$ and each edge can change the $j$-value by at most one, it follows that $j(\wz) \geq j(q) - 2$.

    Let $\ell \ceq |V(G) \setminus (Q \cup R)|$ denote the number of vertices in $G$ that are neither in $Q$ nor $R$.
    
    \begin{claim}
        \label{subclaim:number_unused_nodes}
        It holds that $\ell \geq |Q_2| - 2$.
    \end{claim}
    \begin{claimproof}
        If $|Q_2| \leq 2$, then the claim is trivially satisfied. Hence we may assume $|Q_2| \geq 3$ and let $\widetilde{q}$ be the unique vertex along $Q_2$
        at distance two from $q$.
        Then the vertices along $\wz \arr{S} s$ and $\widetilde{q} \arr{Q_2} q'$ (the former subpath may be longer) decrease in $j$-value by exactly one for each edge traversed (this also holds for all of $Q_2$). 
        We have
        \[
            \len(S) = j(\wz) - j(s) > j(\widetilde{q}) - j(q') = |Q_2| - 3
        \]
        because $j(\wz) \geq j(\widetilde{q})$ and $j(s) < j(q')$. Thus $|S| \geq |Q_2| - 1$ and since $S \cap (R \cup Q) = \{s\}$, this shows that $\ell \geq |Q_2| - 2$.
    \end{claimproof}
    
    In $R$, there are $|X| + |X'| - 2$ vertices outside of $Q$.
    Moreover, there are exactly $n-(d+1)$ vertices in $V(G) \setminus Q$. 
    Hence $|X| + |X'| - 2 \leq n - (d+1) - \ell$. Then, using \cref{subclaim:number_unused_nodes}, we have
    \begin{equation}
        \label{equation:bounding_length}
        \begin{aligned}
            \len(R) &= \len(X) + \len(Q_2) + \len(X') \\
                &= |X| + |Q_2| + |X'| - 3 \\
                &\leq |X| + |X'| + 2 + \ell - 3 \\
                &= |X| + |X'| -2 + \ell + 1 \\
                &\leq n - (d+1) - \ell + \ell + 1 \\
                &= n - d. 
        \end{aligned}
    \end{equation}
    But then $\Gdist_G(x,x') < n-d+1$ in contradiction to our initial assumption.
    
    \proofsubparagraph{Case 3:} $s \in X' \setminus \{q'\}$ (see \cref{subfig:case_3})\\
    Since $y$ and $s$ are on the shortest path $R$ between $x$ and $x'$, $y \arr{R} s$ is a shortest $ys$-path.
    In particular, $y\wz\arr{S}s$ cannot have a strictly smaller length than $y \arr{X} q \arr{Q_2} q' \arr{X'} s$ and thus also not strictly smaller than $q \arr{Q_2} q'$.
    Hence, there are at least $|Q_2| - 2$ vertices strictly between $y$ and $s$ (along $S$) that are neither in $Q$ nor $R$.
    We then employ the same argument as in the previous case (compare \cref{equation:bounding_length}) to obtain a contradiction.

    \proofsubparagraph{Case 4:} $s \in Q_1 \setminus \{q\}$ (see \cref{subfig:case_4})\\
    By assumption on the edge $yz$ it holds that $a(y) > a(q)$.
    Thus by \cref{subclaim:crow_foot}, we have $a(\wz) \geq a(y)$.
    Moreover, by the case assumption $a(s) < a(q)$, which, in summary, implies $a(s) < a(q) < a(y) \leq a(\wz)$.
    As a consequence, somewhere along $S$ the $a$-value of two consecutive vertices needs to decrease.
    Formally, there must exist an edge $\alpha \beta$ along $S$ with $\alpha$ closer to $\wz$ than $\beta$ such that $a(\alpha) > a(\beta)$.
    Since each edge along $S$ decreases the $j$-value by exactly one, we have $j(\beta) = j(\alpha) - 1$.
    This yields
    \begin{flalign*}
        \ \quad&\begin{alignedat}{2}
                &&        a(\alpha) &> a(\beta) \\
        \iff\quad && i(\alpha) - j(\alpha) &> i(\beta) - j(\beta) \\
        \iff\quad && i(\alpha) - j(\beta) - 1 &> i(\beta) - j(\beta) \\
        \iff\quad && i(\alpha) - 1 &> i(\beta),
        \end{alignedat}&
    \end{flalign*}
    a contradiction, since adjacent vertices can differ in their distance to $u$ by at most one.
    \bigskip
    
    \noindent For the remaining two cases, note that $s \neq y$. This is because the $j$-value decreases with each edge taken along $S$, so $j(s) \leq j(\wz)$. At the same time, $j(\wz) < j(y)$ by choice of $\wz$.
    
    \proofsubparagraph{Case 5:} $s \in X \setminus \{q\}$, $s$ between $q$ and $y$ (see \cref{subfig:case_5})\\
    Since $s$ lies between $q$ and $y$ it holds that $\Gdist_G(q, s) < \Gdist_G(q, y)$.
    Then we have $\len(z \arr{X} s)\geq j(z)-j(s)$ since the $j$-value can decrease by at most one along each edge.
    This implies $\len(yz \arr{X} s)\geq j(z)-j(s)+1 \geq j(y) - j(s) + 1$ where the last inequality holds by $j(z) \geq j(y)$ since $yz$ is bad.
    Since $R$ is a shortest path from $x$ to $x'$, $yz \arr{X} s$ needs to be a shortest path from $y$ to $s$.
    However, $\len(\wz \arr{S}s)=j(\wz)-j(s)$ by definition of $S$ and with $j(\wz) = j(y) - 1$ it follows that $\len(y\wz \arr{S}s)=j(y)-j(s)$.
    This is a contradiction.

    \proofsubparagraph*{Case 6:} $s \in X \setminus \{q\}$, $s$ between $y$ and $x$ (see \cref{subfig:case_6})\\
    Formally, one has $\Gdist_G(q, s) > \Gdist_G(q, y)$. 
    In contrast to previous cases, we do not care about $s$ (the first meeting point of the shortest path between $\wz$ and $v$ with $X$), but rather the last which we denote by $s^\ast$ in the following.
    Let $S^\ast$ denote the corresponding $\wz s^\ast$-subpath.
    Since the $j$-value decreases along each edge of $S^\ast$ it follows that $s^\ast$ also lies between $y$ and $x$, the reason is similar to 
    \textbf{\textsf{\textcolor{lipicsGray}{Case 4}}}.
    Note that $S^\ast$ need not be a subpath of $X$. However, we have $\len(y\wz \arr{S^\ast} s^\ast) \leq \len(y \arr{X} s^\ast)$ due to decreasing $j$-values.
    
    Note that we can decompose $R$ as shown below, with a (potentially shorter) alternative path $R^\ast$ which follows $S^\ast$ between $s^\ast$ and $y$:
    
    \begin{equation}
        \label{equation:RandRstar}
        \begin{aligned}
        R = x \arr{X} s^\ast &\arr{\mathmakebox[\widthof{$\scriptstyle S^\ast$}]{X}} \hphantom{\wz}yz \arr{X} q \arr{Q_2} q' \arr{X'} x', \\
        R^\ast = x \arr{X} s^\ast &\arr{S^\ast} \wz yz \arr{X} q \arr{Q_2} q' \arr{X'} x'.
        \end{aligned}
    \end{equation}
    Now consider the shortest path between $s^\ast$ and $v$ that is the continuation (starting at $s^\ast$) of the shortest $\wz v$-path that was used to define $S^\ast$ and let $T$ denote the $s^\ast t$-subpath where $t$ is the first vertex in $R \cup Q$ after $s^\ast$.
    It is clear that $t \notin V(X)$. Moreover, $t \notin V(Q_2) \cup V(X')$ as this would yield an even shorter $xx'$-path than we obtained in \textbf{\textsf{\textcolor{lipicsGray}{Case 1}}} and \textbf{\textsf{\textcolor{lipicsGray}{Case 3}}}.
    Similarly, $t \notin V(Q_1)$ by using the same reasoning as in \textbf{\textsf{\textcolor{lipicsGray}{Case 4}}}.
    Hence $t \in V(Q_3)$.
    We construct an $xx'$-path as follows:
    \[
        R' \ceq x \arr{X} s^\ast \arr{T} t \arr{Q_3} q' \arr{X'} x'.
    \]
    Let $U, U'$ denote the $s^\ast q'$-subpaths of $R, R'$, respectively.
    Along $T$ and $S^\ast$ the $j$-value of neighbouring vertices differs by exactly one.
    Hence, we derive the following two (in)equalities for the lengths of $U$ and $U'$ using the $j$-values of intermediate vertices along the paths:
    \begin{align}
        \label{equ:lenU}
        \len(U) &\geq \big(j(y) - j(s^\ast)\big) + \big(j(y) - j(q') + 1\big), \\
        \label{equ:lenUprime}
        \len(U') &= \big(j(s^\ast) - j(t)\big) + \big(j(q') - j(t)\big).
    \end{align}
    The first summand on the right-hand side of \cref{equ:lenU} is due to \cref{equation:RandRstar} since $ \len(s^\ast \arr{X} y) \geq \len(s^\ast \arr{S^\ast} \wz y) = j(y) - j(s^\ast)$.
    Since $U$ is a subpath of a shortest $xx'$-path, it follows that $U$ is a shortest $s^\ast q'$-path and thus $\len(U) \leq \len(U')$.
    Hence, combining \cref{equ:lenU,,equ:lenUprime}, we get
    \begin{flalign*}
        \quad&\begin{alignedat}{2}
            &&  \big(j(s^\ast) - j(t)\big) + \big(j(q') - j(t)\big) &\geq \big(j(y) - j(s^\ast)\big) + \big(j(y) - j(q') + 1\big) \\
            \iff\quad &&  j(s^\ast) - j(t) &\geq j(y) - j(q') + \mfrac{1}{2}.
        \end{alignedat}&
    \end{flalign*}
    As $j$-values are integers, we can round up the right-hand side of the last inequality.
    With $j(y) - j(q') + 1 \geq |Q_2| - 1$ it now follows that
    $\len(T) \geq |Q_2| - 1$ and hence $|T| \geq |Q_2|$. Since the interior of $T$ is disjoint from $R \cup Q$, we have found $|Q_2| - 2$ vertices that are not used by $R$ or $Q$. Then 
    \begin{equation*}
        \begin{aligned}
            |R| &\leq n - (d+1) + |Q_2| - (|T|-2) \\
                &= n - d + 1 + |Q_2| - |T| \\
                &\leq n - d + 1 \\
        \end{aligned}
    \end{equation*}
    and hence $\Gdist_G(x, x') = \len(R) \leq n-d$, a contradiction to our initial assumption.\bigskip
    
    \noindent In all of the above cases, we have obtained a contradiction. In particular, this implies that $X$ cannot have contained a bad edge.
    Applying the same case distinction (up to symmetry of the cases) to a bad edge in $X'$ yields the same conclusion for $X'$.
    In total, this rules out the existence of bad edges altogether, so \cref{equation:vertex_constellation} holds. This concludes the proof.
\end{proof}

\section{\texorpdfstring{Proofs Omitted from \cref{section:max_distortion_bounded_pathlength}}{Proofs Omitted from Section~\ref{section:max_distortion_bounded_pathlength}}}
\label{appendix:max_distortion_bounded_pathlength}
\newcommand{\gammamon}{\ensuremath{\overset{\smash{\raisebox{-0.4ex}{$\scriptscriptstyle\rightarrow$}}}}{\dis}}
\newcommand{\gammaanti}{\ensuremath{\overset{\smash{\raisebox{-0.4ex}{$\scriptscriptstyle\leftarrow$}}}{\dis}}}

\newcommand{\mon}{\ensuremath{\overset{\smash{\raisebox{-0.3ex}{$\scriptscriptstyle\rightarrow$}}}{\pi}}}
\newcommand{\anti}{\ensuremath{\overset{\smash{\raisebox{-0.3ex}{$\scriptscriptstyle\leftarrow$}}}{\pi}}}

\newcommand{\rev}[1]{#1^{\textsf{rev}}}

\subsection{\texorpdfstring{Proof of \cref{proposition:2_approx_real_point_sets}}{Proof of Theorem~\ref{proposition:2_approx_real_point_sets}}}
We extend the definitions of $\dis_{(\infty)}$ and $\Ddist{(\infty)}$ to (pseudo)metric spaces in the obvious way.
Throughout, finite (ordered) subsets $X = \{x_1 < \cdots < x_n\}$ of $\R$ are equipped with the standard metric $\abs{x-y}$,
that is, we view $X$ as a metric space. For $i \leq j \in [n]$ define 
\begin{align*}
    X_{[i,j]} \ceq \{x_i, \hdots, x_j\}, \quad X_{(i,j)} \ceq \{x_{i+1}, \hdots, x_{j-1}\},
\end{align*}
and analogous half-open variants.
Note that, as for path graphs which are special cases of subsets of the real line,
distances in such subsets decompose along the respective order.
The \emph{reverse of $X$}, denoted $\rev{X}$, is defined as $\{-x_n < \cdots < -x_1\}$
and is isometric to $X$ with $x_i \mapsto - x_i$ being an isometry. 
Given another ordered subset $Y = \{y_1 < \cdots < y_n\}$ of $\R$, 
we denote by $\mon$ the monotone bijection defined by
$x_i \mapsto y_i$ for all $i \in [n]$.
Analogously, we write $\anti$ for the antitone bijection defined by 
$x_i \mapsto y_{n+1-i}$ for all $i \in [n]$.
These are the two bijections of \cref{proposition:2_approx_real_point_sets}.
Note that the monotone bijection $X \to Y$
corresponds to the antitone bijection $X \to \rev{Y}$
by chaining it with the described isometry.

The proof of \cref{proposition:2_approx_real_point_sets} relies crucially on the following lemma, which,
intuitively, tells us that
if an alignment significantly reduces the distortion incurred by the monotone alignment
on large distances in the space $X$, then it must act antitonically on certain outer tails of $X$.

\begin{lemma}\label{lemma:auxiliary_monotone_guarantee}
    Let $X = \{x_1 < \cdots < x_n\}$ and $Y = \{y_1 < \cdots < y_n\}$ be ordered subsets of the 
    real line, and $\pi \in \Bij(X,Y)$. If there are $i < j \in [n]$ such that
    \begin{align*}
        |x_i - x_j| - |y_i - y_j| > \mfrac{3}{2} \dis_{(\infty)}(\pi),
    \end{align*}
    then $\pi (x') < \pi (x)$ holds for all $x \in X_{[1,i]}, \, x' \in X_{[j,n]}$. 
\end{lemma}

\begin{proof}
    Write $\dis \ceq \dis_{(\infty)}(\pi)$ and note that the assumption implies
        $2(x_j-x_i) > 3\dis$.
    Define
    \begin{align*}
        L \ceq \min \{k \in [n] \mid \pi(x_k) \in Y_{[i,j]} \}, \quad
        R \ceq \max \{k \in [n] \mid \pi(x_k) \in Y_{[i,j]} \}
    \end{align*}
    and note that $x_R - x_L \leq y_j - y_i + \dis < x_j - x_i$ by
    definition of $\dis$. It follows from $\abs*{Y_{[i,j]}} > \abs*{X_{(i,j)}}$
    that either $L \leq i$ or $R \geq j$ holds and since $x_R - x_L < x_j - x_i$
    the statements cannot be true simultaneously. The resulting cases $L \leq i$
    and $R \geq j$ are easily seen to be equivalent by reversing both spaces.
    We shall therefore only consider the case $L \leq i$ in which, as noted before,
    $R < j$ has to hold. 
    Intuitively, the preimage of $Y_{[i,j]}$ under $\mon$ has diameter $x_j-x_i$, which exceeds
    $y_j-y_i$ by more than $\frac{3}{2}\dis$. By shifting this preimage to the left, the bijection $\pi$
    can reduce its diameter to $x_R-x_L < x_j-x_i$ and thereby incur less distortion than $\mon$. 
    Let $\widehat{x} \in X_{[1,i]}$ such that $\pi (\widehat{x}) \in Y_{[i,j]}$.
    Since $\abs*{X_{[j,n]}} > \abs*{Y_{(j,n]}}$ and no point of $X_{[j,n]}$ maps under $\pi$
    into $Y_{[i,j]}$ (by definition of $R$ and $R < j$), there exists some $\widetilde{x} \in X_{[j,n]}$
    with $\pi (\widetilde{x}) \in Y_{[1,i)}$ and
    \begin{align*}
        \pi (\widetilde{x}) < y_i \leq \pi (\widehat{x}).
    \end{align*}
    Let $x \in X_{[1,i]}$ and assume towards a contradiction that $\pi (x) < \pi (\widetilde{x})$.
    Then $\pi (x) < \pi (\widetilde{x}) < \pi (\widehat{x})$ and thus
    \begin{align*}
        |x - \widehat{x}| + \dis
        &\geq |\pi (x) - \pi (\widehat{x}) | \\
        &= |\pi (x) - \pi (\widetilde{x})| + |\pi (\widetilde{x}) - 
            \pi (\widehat{x}) | \\
        &\geq |x - \widetilde{x}| + |\widetilde{x}- \widehat{x}| - 2\dis.
    \end{align*}
    However, since $x, \widehat{x} \in X_{[1,i]}$ and $\widetilde{x} \in X_{[j,n]}$ we have
    \begin{align*}
        |x - \widetilde{x}| + |\widetilde{x} - \widehat{x}| - |x - \widehat{x}|
        &= 2\widetilde{x} - x - \widehat{x} - \big ( \max \{x,\widehat{x}\} - \min \{x, \widehat{x}\} \big ) \\
        &= 2\big ( \widetilde{x} - \max \{x, \widehat{x}\} \big ) \\
        &\geq 2(x_j-x_i) \\
        &> 3\dis,
    \end{align*}
    a contradiction. Since $x \neq \widetilde{x}$, it follows that $\pi (\widetilde{x}) < \pi (x)$.

    Finally, let $x' \in X_{[j,n]}$ and assume towards a contradiction that $\pi (x) < \pi (x')$.
    Then $\pi (\widetilde{x}) < \pi (x) \leq \pi (x')$ and thus
    \begin{align*}
        |\widetilde{x} - x'| + \dis
        &\geq |\pi(\widetilde{x}) - \pi(x)| + |\pi (x) - \pi (x') | \\
        &\geq |\widetilde{x} - x| + |x - x'| - 2\dis.
    \end{align*}
    However, since $x \in X_{[1,i]}$ and $\widetilde{x}, x' \in X_{[j,n]}$ we have
    \begin{align*}
        |\widetilde{x}- x| + |x - x'| - |\widetilde{x} - x'|
        &= \widetilde{x} - x + x' - x - \big ( \max \{\widetilde{x}, x'\} - \min \{\widetilde{x}, x'\} \big ) \\
        &= 2 \big ( \min \{\widetilde{x}, x'\} - x \big ) \\
        &\geq 2(x_j-x_i) \\
        &> 3\dis,
    \end{align*}
    a contradiction. Since $x \neq x'$, it follows that $\pi(x') < \pi(x)$.
\end{proof}

    We now proceed with the proof of \cref{proposition:2_approx_real_point_sets}.
    Assume towards a contradiction that there exist ordered subsets of the real line
    $X = \{x_1 < \cdots < x_n\}$ and $Y = \{y_1 < \cdots < y_n\}$ with
    \begin{align}\label{eq:sim_failure_assumption}
        \gammamon > \mfrac{7}{4}\dis \quad\text{and}\quad \gammaanti > \mfrac{7}{4}\dis
    \end{align}
    where $\dis \ceq \Ddist{(\infty)}(X, Y)$ and 
    $\gammamon \ceq \dis_{(\infty)}(\mon)$, $\gammaanti \ceq  \dis_{(\infty)}(\anti)$ denote the distortion of the monotone and antitone alignments ($\mon$ and $\anti$), respectively.
    Note that $n > 1$, because otherwise there only exists a single alignment.
    Note further that $\dis > 0$, because an isometry is necessarily monotone or antitone. 
    Fix an optimal $\pi \in \Bij(X, Y)$, that is $\dis_{(\infty)}(\pi) = \dis$.
    We may assume without loss of generality that $\diam(X) \leq \diam(Y)$, otherwise
    simply swap $X$ and $Y$. Since
    every alignment can have distortion at most $\diam(Y)$, we have $\diam(Y) \geq \frac{7}{4}\dis > \frac{3}{2}\dis$.

    Note that $\pi^{-1} (y_1) \neq \pi^{-1} (y_n)$ since $n > 1$ and $\pi$ is a bijection. 
    If necessary, we reverse $X$ and assume henceforth that
    $\pi^{-1} (y_1) < \pi^{-1} (y_n)$. 
    Reversing either of the spaces interchanges the distortions $\gammamon$ and $\gammaanti$, 
    and hence preserves the simultaneous failure assumption \cref{eq:sim_failure_assumption}. 

    We shall now apply \cref{lemma:auxiliary_monotone_guarantee} to $\mon$ and $\anti$ in a suitable
    fashion. Consider first the monotone alignment $\mon$. Suppose $\gammamon$ is realised by 
    an expansion of a distance in $X$, i.e. that there exist $x,x' \in X$ with $|x-x'|=
    |\mon (x) - \mon (x')|-\gammamon$. Applying \cref{lemma:auxiliary_monotone_guarantee} to
    the alignment $\pi^{-1} \colon Y \to X$ shows that $\pi^{-1}$ acts antitonically on $y_1, y_n$,
    i.e. $\pi^{-1}(y_n) < \pi^{-1}(y_1)$, contradicting our assumption $\pi^{-1} (y_1) < \pi^{-1} (y_n)$.
    Hence $\gammamon$ is only realised by contractions of distances in $X$
    and in particular there are $i < j \in [n]$ such that 
    \begin{align}\label{eq:aux_mon_dis}
        |x_i-x_j| = |\mon (x_i) - \mon (x_j)| + \gammamon.
    \end{align}
    Applying \cref{lemma:auxiliary_monotone_guarantee} shows
    \begin{align}\label{eq:cons_mon_aux_lemma}
        \pi (x') < \pi (x) \quad \text{ for all } x \in X_{[1,i]}, \, x' \in X_{[j,n]}.
    \end{align}
    In particular, $\pi^{-1}(y_1) \notin X_{[1,i]}$ and $ \pi^{-1}(y_n) \notin X_{[j,n]}$,
    because otherwise there must exist points smaller than $y_1$ or larger than $y_n$ in $Y$.
    Overall we thus know
    \begin{align}\label{eq:aux_X_order}
        x_i < \pi^{-1} (y_1) < \pi^{-1} (y_n) < x_j.
    \end{align}

    Consider next the antitone alignment $\anti$. Suppose $\gammaanti$ is attained by a contraction
    of a distance in $X$, \ie{} that there exist $x,x' \in X$ with $|x-x'| = \left|\anti (x)- 
    \anti (x')\right| + \gammaanti$. Since $\anti$ is monotone as a map $X \to \rev{Y}$, we can apply
    \cref{lemma:auxiliary_monotone_guarantee} to $\pi \colon X \to \rev{Y}$ to obtain
    that $\pi (x_1) < \pi (x_n)$, contradicting
    \cref{eq:cons_mon_aux_lemma}. Hence $\gammaanti$ is only attained by expansions of distances
    in $X$. Consequently, there are $ k < \ell \in [n]$ with 
    \begin{align}\label{eq:aux_anti_dis}
        |x_k -x_\ell| = |\anti (x_k)- \anti (x_\ell) | - \gammaanti.
    \end{align}
    Since $\anti^{-1}$ is monotone as a map $Y \to \rev{X}$, 
    applying \cref{lemma:auxiliary_monotone_guarantee} to $\pi^{-1} \colon Y \to \rev{X}$ shows
    \begin{align*}
        \pi^{-1} (y) < \pi^{-1} (y') \quad
        \text{ for all } y \in Y_{[1,n+1-\ell]}, \, y' \in Y_{[n+1-k,n]}.
    \end{align*}
    In particular, $\pi (x_1) \notin Y_{[n+1-k,n]}$, as otherwise every point in
    $\pi^{-1}(Y_{[1, n+1-\ell]})$ must be strictly smaller than $x_1$, which is impossible. 
    Similarly, $\pi (x_n) \notin Y_{[1,n+1-\ell]}$. Together with \cref{eq:cons_mon_aux_lemma}
    this implies 
    \begin{align}\label{eq:aux_Y_order}
        \anti (x_\ell) = y_{n+1-\ell} < \pi (x_n) < \pi (x_1) < y_{n+1-k} = \anti (x_k).
    \end{align}

    The key connection between the monotone and antitone alignments is made by the following claim,
    stating that $X_{[k,\ell]} \subseteq X_{(i,j)}$ and $Y_{[i,j]} \subseteq Y_{(n+1-\ell, n+1-k)}$
    cannot hold simultaneously.

    \begin{claim}\label{claim:aux_mon_and_anti}
        At least one of the following inequalities holds:
        \begin{align*}
            k \leq i, \quad \ell \geq j, \quad i \leq n+1-\ell, \quad j \geq n+1-k.
        \end{align*}
    \end{claim}

    \begin{claimproof}
        Note that we have
        $\abs*{Y_{[i,j]}} > \abs*{X_{(i,j)}}$ and 
        $\abs*{X_{[k, \ell]}}
            > \abs*{Y_{(n+1-\ell, n+1-k)}}$.
        Hence, if $X_{[k,\ell]} \subseteq X_{(i,j)}$ and
        $Y_{[i,j]} \subseteq Y_{(n+1-\ell, n+1-k)}$ hold simultaneously, it follows that
        \begin{align*}
            \abs*{Y_{[i,j]}} > \abs*{X_{(i,j)}} \geq \abs*{X_{[k,\ell]}} >
            \abs*{Y_{(n+1-\ell, n+1-k)}} \geq \abs*{Y_{[i,j]}},
        \end{align*}
        a contradiction. 
    \end{claimproof}
    Combining \cref{eq:aux_mon_dis,,eq:aux_anti_dis}, we have
    \[
        \gammamon + \gammaanti 
        =
        (x_j - x_i) - (y_j - y_i) + (y_{n+1-k} - y_{n+1-\ell}) - (x_\ell - x_k).
    \] 
    Based on the four cases outlined in \cref{claim:aux_mon_and_anti}, we bound
    $\gammamon + \gammaanti$ from above by regrouping according to the case at hand,
    dropping the two non-positive terms, and bounding each of the remaining two summands.
    
    \begin{description}
        \item \textbf{\textsf{\textcolor{lipicsGray}{Case 1}}} ($k \leq i$): Then $x_k \leq x_i$ and 
            \begin{align*}
                \gammamon + \gammaanti
                &= (x_j - x_\ell) + (x_k - x_i) - (y_j - y_i) + (y_{n+1-k} - y_{n+1-\ell}) \\
                &\leq (x_j - x_\ell) + (y_{n+1-k} - y_{n+1-\ell}) \\
                &\leq x_n - x_\ell + \diam(Y).
            \end{align*}
        \item \textbf{\textsf{\textcolor{lipicsGray}{Case 2}}} ($\ell \geq j$): Then $x_\ell \geq x_j$ and
            \begin{align*}
                \gammamon + \gammaanti
                &= (x_j - x_\ell) + (x_k - x_i) - (y_j - y_i) + (y_{n+1-k} - y_{n+1-\ell}) \\
                &\leq (x_k-x_i) + (y_{n+1-k} - y_{n+1-\ell}) \\
                &\leq x_k - x_1 + \diam(Y).
            \end{align*}
        \item \textbf{\textsf{\textcolor{lipicsGray}{Case 3}}} ($i \leq n+1-\ell$): Then $y_i \leq y_{n+1-\ell}$ and
            \begin{align*}
                \gammamon + \gammaanti
                &= (x_j - x_i) + (y_{n+1-k} - y_j) + (y_i - y_{n+1-\ell}) - (x_\ell - x_k) \\
                &\leq (x_j - x_i) + (y_{n+1-k} - y_j) \\
                &\leq \diam (X) + y_n - y_j \\
                &\leq \diam(Y) + y_n - y_j.
            \end{align*}
        \item \textbf{\textsf{\textcolor{lipicsGray}{Case 4}}} ($j \geq n+1-k$): Then $y_j \geq y_{n+1-k}$ and 
            \begin{align*}
                \gammamon + \gammaanti
                &= (x_j - x_i) + (y_{n+1-k} - y_j) + (y_i - y_{n+1-\ell}) - (x_\ell - x_k) \\
                &\leq (x_j - x_i) + (y_i - y_{n+1-\ell}) \\
                &\leq \diam (X) + y_i-y_1 \\
                &\leq \diam(Y) + y_i - y_1.
            \end{align*}
    \end{description}

    \noindent The final two claims show that in each of the four cases, 
    the right-hand side can be bounded from above by $\frac{7}{2}\dis$.
    This implies that either $\gammamon \leq \frac{7}{4}\dis$
    or $\gammaanti \leq \frac{7}{4}\dis$ which contradicts our initial assumption
    stated in \cref{eq:sim_failure_assumption}.

    \begin{claim}\label{claim:mon_anti_diam_bound}
        $\diam(Y) \leq 2\dis$.
    \end{claim}

    \begin{claimproof}
        Recall from \cref{eq:aux_X_order} and \cref{eq:aux_Y_order} that
        \begin{align*}
            x_1 < \pi^{-1} (y_1) < \pi^{-1} (y_n) < x_n 
            \quad\text{and}\quad
            y_1 < \pi (x_n) < \pi (x_1) < y_n.
        \end{align*}
        Since $\abs{y_1-y_n}=\diam(Y)$,
        we have $\left|\pi^{-1}(y_1)-\pi^{-1} (y_n)\right| \geq \diam(Y)-\dis$ and thus
        \begin{align*}
            \left|x_1- \pi^{-1} (y_1)\right| + \left|\pi^{-1}(y_n)- x_n\right|
            &= \diam(X) - \left|\pi^{-1}(y_1) - \pi^{-1} (y_n)\right| \\
            &\leq \dis - \big (\diam(Y)-\diam(X)\big ).
        \end{align*}
        Similarly, $\left|\pi (x_n) - \pi (x_1)\right| \geq \diam(X) - \dis$ because
        $\abs{x_1 - x_n}=\diam(X)$. Finally,
        \begin{flalign*}
            \quad&\begin{alignedat}{2}
                &&  \left|x_1 - \pi^{-1} (y_1) \right| + \left|\pi^{-1}(y_n) - x_n\right| 
                &\geq \left| \pi (x_1) - y_1\right| + \left|y_n - \pi (x_n)\right| - 2\dis \\
                \implies\quad && \dis - \big (\diam(Y)-\diam(X)\big ) &\geq 
                \diam(Y) + \left|\pi (x_n)- \pi (x_1) \right| - 2\dis \\
                \implies\quad && 3\dis + \diam(X) - \left|\pi (x_n) - \pi (x_1) \right| &\geq 2 \diam (Y) \\
                \implies\quad && 4 \dis &\geq 2\diam(Y),
            \end{alignedat}&
        \end{flalign*}
        which proves the claim.
    \end{claimproof}

    \begin{claim}\label{claim:mon_anti_tail_bounds}
        Each of $ x_n - x_\ell,\, x_k-x_1,\, y_n - y_j,\, y_i - y_1$ is 
        bounded from above by $\frac{3}{2}\dis$.
    \end{claim}

    \begin{claimproof}
        For $x \in X_{[1,i]}$ we can make use of the following two bounds:
        \begin{align*}
            \left|y_1 - \pi (x) \right| &\leq \left|\pi^{-1} (y_1) - x\right| + \dis, \\
            \left|y_1 - \pi (x) \right| &=\diam(Y) - \left|\pi (x) - y_n\right| \\
            &\leq \diam(Y) - \left|x - \pi^{-1} (y_n )\right| + \dis.
        \end{align*}
        By \cref{eq:aux_X_order}, we can write $\left|x - \pi^{-1} (y_n)\right| =
         \left|x - \pi^{-1} (y_1)\right| 
        + \left|\pi^{-1}(y_1) - \pi^{-1}(y_n)\right|$.
        Plugging this into the sum of the above upper bounds yields
        \begin{align*}
            2\left|y_1 - \pi (x) \right| 
            \leq \diam(Y) - \left|\pi^{-1}(y_1)- \pi^{-1}(y_n) \right| + 2 \dis
            \leq 3\dis.
        \end{align*}
        Since $\abs*{X_{[1,i]}}=i > i-1 = \abs*{Y_{[1,i)}}$, $\pi$ has to map some $x \in X_{[1,i]}$ into
        $Y_{[i,n]}$ which yields $y_i - y_1 \leq \left|y_1 - \pi (x)\right| \leq \frac{3}{2}\dis$.
        
        The same calculation for $x' \in X_{[j,n]}$ paired with $y_n$ shows $\left|\pi (x') - y_n\right| \leq
        \frac{3}{2}\dis$ and, for a similar cardinality reason, $\pi$ has to map some
        $x' \in X_{[j,n]}$ into $Y_{[1,j]}$. Consequently, $y_n - y_j \leq \left|\pi (x') - y_n\right|
        \leq \frac{3}{2}\dis$. 

        The remaining two upper bounds are proven analogously by considering $x_1, x_n$
        paired with $y \in Y_{[1,n+1-\ell]}$ or $y' \in Y_{[n+1-k, n]}$ respectively.
    \end{claimproof}

    \noindent In each of the four cases we obtain $\gammamon + \gammaanti \leq 2\dis + \frac{3}{2}\dis = \frac{7}{2}\dis$. This concludes the proof of \cref{proposition:2_approx_real_point_sets}.

\subsection{Applying Theorem 11 to Graphs of Bounded Pathlength}
To derive \cref{lemma:max_distortion_caterpillar_approx} and \cref{theorem:pathlength_approx}
from \cref{proposition:2_approx_real_point_sets}, we use the following corollary.
Note that equipping multisets over $\R$ with the standard metric $\abs{x-y}$ results in pseudometric spaces.
For two pseudometric spaces $\X, \Y$, let $\Delta^\ast(\X, \Y)$ denote the 
optimal maximum distortion between $\X$ and $\Y$.

\begin{corollary}\label{corr:multiset_to_set}
     Let $\X = \multiset{x_1, \hdots, x_n} $ and $\Y =\multiset{y_1, \hdots, y_n}$ 
     be multisets over $\mathbb{Z}$ enumerated in nondecreasing order. 
     The better of the two bijections $\pi(x_i) = y_i$ and $\pi(x_i) = y_{n+1-i}$ for $i \in [n]$
     has a maximum additive distortion of at most $\frac{7}{4}$ times the optimum.
\end{corollary}

\begin{proof}
    We begin by shifting each point within $\X$ and $\Y$ by a small amount. In particular, this shifts repeated values apart and thus turns  multisets into sets 
    while preserving the original distances up to an error depending only on the shift amounts. 
    Formally, for $\eps \in (0,1)$ we define
    \begin{align*}
        \X^\eps &\ceq \left\{ x_j + j \cdot \mfrac{\eps}{n} \midd j \in [n] \right\} \quad\text{and}\quad        \Y^\eps \ceq \left\{ y_j + j \cdot \mfrac{\eps}{n} \midd j \in [n] \right\}.
    \end{align*}
    Let $x_1' < \cdots < x_n'$ and $y_1' < \cdots < y_n'$ be the (now increasingly) ordered elements of $\X^\eps$ and $\Y^\eps$, respectively.
    Using the identity on $[n]$, we can naturally
    identify bijections $\X \to \Y$ with bijections $\X^\eps \to \Y^\eps$. Note that this 
    identification preserves monotonicity and antitonicity. 
    Since no single point is shifted by more than $\eps$, the maximum distortion of any fixed bijection
    differs by at most $2\eps$ between its set version $\X^\eps \to \Y^\eps$ and its multiset version
    $\X \to \Y$. 
    
    Invoking \cref{proposition:2_approx_real_point_sets}, we obtain a bijection
    $\X^\eps \to \Y^\eps$ of maximum distortion $\Delta_\eps \leq \frac{7}{4}\Delta^\ast (\X^\eps, \Y^\eps)$.
    Let $\Delta$ denote the distortion of the multiset version of that bijection and observe that
    \begin{align*}
        \Delta &\leq \Delta_\eps + 2\eps
        \leq \mfrac{7}{4}\Delta^\ast\big(\X^\eps, \Y^\eps\big) + 2\eps
        \leq \mfrac{7}{4}\big( \Delta^\ast(\X, \Y) + 2\eps \big) + 2\eps
        = \mfrac{7}{4}\Delta^\ast (\X, \Y) + \mfrac{11}{2}\eps.
    \end{align*}
    Hence, for $\eps < \frac{1}{22}$ we have $\frac{11}{2}\eps < \frac{1}{4}$ and
    \begin{align*}
        \Delta < \mfrac{7}{4}\Delta^\ast (\X, \Y) + \mfrac{1}{4} 
        \leq \left\lfloor \mfrac{7}{4}\Delta^\ast (\X, \Y) \right\rfloor  + 1 
    \end{align*}
    since $\Delta^\ast(\X, \Y)$ is an integer.
    We can conclude $\Delta \leq \frac{7}{4}\Delta^\ast(\X,\Y)$ because 
    $\Delta$ is an integer as well.
\end{proof}

Recall that for a caterpillar $T$, the \emph{spine} of $T$ is the path $v_1, \ldots, v_k$ that remains 
after removing all its leaves.
In the following, we assume $T \neq K_1$ and $T \neq K_2$ so that a spine is always well-defined.
Let $S(T)$ denote the set $\{v_1, \ldots, v_k\}$ of spine vertices of $T$ and
$L(T) \ceq V(T) \setminus S(T)$ the set of leaves of $T$.
Further, for a vertex $v \in V(T)$, let $s_T(v) \in [k]$ denote
the index of the unique closest spine vertex, \ie{} the unique neighbour of a leaf, or the vertex 
itself if it is a spine vertex. 
We define $B_T(j) \ceq N_T[v_j] \setminus (N_T(v_j) \cap S(T))$ as the \emph{bag} of the $j$-th spine vertex for each $j \in [k]$.
Moreover, let $b_T(j) \ceq |B_T(j)|$ denote the \emph{size} of the $j$-th bag. 
Setting the length of edges between $S(T)$ and $L(T)$ to zero (from one), we obtain a multiset $M(T)$ over $\{1, \ldots, k\}$ containing the value $i$ exactly $b_T(i)$ times for each $i \in [k]$.
We view $M(T)$ as pseudometric space with metric $d_T$ defined by
\[
    d_T(x, x') \ceq |s_T(x) - s_T(x')|
\]
for all $x, x' \in V(T)$.
Using the natural identification of vertices of $T$ with the points in the multiset $M(T)$, we may write $|x - x'|$ instead of $|s_T(x) - s_T(x')|$.
Note that the resulting multisets always have a support consisting of \emph{consecutive} integers.
For two caterpillars $T_1, T_2$, an alignment $\pi$ between $M(T_1)$ and $M(T_2)$, that is, a bijection between the vertex sets, is said to be \emph{monotone} if for all $x, x' \in V(T_1)$ with $s_{T_1}(x) < s_{T_1}(x')$ it holds that $s_{T_2}(\pi(x)) \leq s_{T_2}(\pi(x'))$, and \emph{antitone} if for all $x, x' \in V(T_1)$ with $s_{T_1}(x) < s_{T_1}(x')$ it holds that $s_{T_2}(\pi(x)) \geq s_{T_2}(\pi(x'))$.

For two pseudometric spaces $\X, \Y$ and an algorithm $\mathcal A$ computing a bijection $\pi_{\mathcal A}$ between $\X$ and $\Y$, we denote by $\Delta^{\mathcal A}(\X, \Y)$ 
the maximum distortion of $\pi_{\mathcal A}$.

\begin{proof}[Proof of \cref{lemma:max_distortion_caterpillar_approx}]
    Let $\mathcal B$ be the algorithm that, given two caterpillars $T_1$ and $T_2$ on the same number of vertices, first computes $M_1 \ceq M(T_1)$ and $M_2 \ceq M(T_2)$ and then runs the algorithm $\mathcal A$ of \cref{corr:multiset_to_set} on input $M_1$ and $M_2$ and returns the corresponding bijection between $V(T_1)$ and $V(T_2)$.
    Identifying the elements of $M_i$ with $V(T_i)$ in the obvious way, for all bijections $\pi \in \Bij(V(T_1), V(T_2))$, it holds that
    \begin{equation}
        \label{equ:caterpillar_to_multiset}
        \Big|
        \underbrace{\big|\Gdist_{T_1}(x, x') - \Gdist_{T_2}(\pi(x), \pi(x')) \big|}_{\text{caterpillar distortion}} - \underbrace{\big| |x - x'| - |\pi(x) - \pi(x')| \big|}_{\text{multiset distortion}} \Big| \leq 2.
    \end{equation}
    Thus, we have

    \begin{align*}
        \Delta^{\mathcal B}(T_1, T_2)
        &\leq \Delta^{\mathcal A}(M_1, M_2) + 2 \tag*{\cref{equ:caterpillar_to_multiset}}\\
        &\leq \mfrac{7}{4} \Delta^\ast(M_1, M_2) + 2 \tag*{\cref{corr:multiset_to_set}}\\
        &\leq \mfrac{7}{4} \big(\Delta^\ast(T_1, T_2) + 2\big) + 2 \\
        &=\mfrac74\Delta^\ast(T_1, T_2) + \mfrac{11}{2}.
    \end{align*}
    Consequently, algorithm $\mathcal B$ provides a $\left(\frac{7}{4}, \frac{11}{2}\right)$-approximation for 
    caterpillars. 
    
    Regarding the runtime, the spine of a caterpillar can be found in linear time, and thus, a monotone and an antitone alignment between $M_1$ and $M_2$ can be computed in linear time.
    Moreover, the maximum distortion of these bijections 
    can be evaluated in linear time by a single pass over all $x \in V(T_1)$ along the bags corresponding to the spine.
    This is because for a monotone alignment $\pi$ between $M_1$ and $M_2$, it holds that
    \[
        \dis_{(\infty)}(\pi) = \left(\max_{x \in V(T_1)} s_{T_1}(x) - s_{T_2}(\pi(x))\right) - \left(\min_{x \in V(T_1)} s_{T_1}(x) - s_{T_2}(\pi(x)) \right).
    \]
    Intuitively, the quantity $s_{T_1}(x) - s_{T_2}(\pi(x))$ can be understood as the \enquote{tilt} of the \enquote{assignment edge} $x\pi(x)$. 
    Combining the maximum forward and the maximum backward tilt yields the maximum distortion value since the alignment is monotone.
\end{proof}

\cref{example:caterpillar_spaces_32} gives an infinite family of pairs of multisets of unbounded 
diameter for which the monotone and antitone bijections can be off by a factor of $\frac{3}{2}$ from the 
optimal distortion. 
The following example, using multisets, shows that the factor $\frac{7}{4}$ in \cref{proposition:2_approx_real_point_sets} is optimal.
Note that these multisets do not 
directly correspond to caterpillars:
By definition, the outmost bags of multisets obtained from caterpillars must always contain at least two elements (due to the \enquote{contraction} of edges incident to leaves). 
Because of this, we cannot rule out the possibility of an approximation algorithm with multiplicative approximation guarantee of $\frac{3}{2}$ for caterpillars yet. Any example of bounded diameter with a larger multiplicative gap between optimal bijections and the monotone and antitone bijections could still be solved to optimality in polynomial time by solving integer programs in which the number of integer variables is 
bounded by a constant \cite{lenstra_ip_fixed_number_variables}.

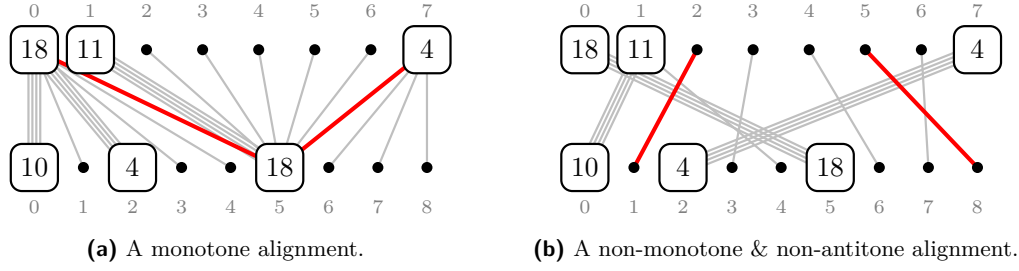
\begin{figure}[!htbp]
    \captionsetup{singlelinecheck = true, justification=justified}
    \centering
    \begin{subfigure}[t]{0.48\textwidth}
        \centering
        \begin{adjustbox}{height=3cm}
\def\hLength{5}   
\def\vGap{1.5}       
\begin{tikzpicture}[
    scale=1,
    box/.style={
        draw,
        fill=white,
        rounded corners,
        minimum width=0.6cm,
        minimum height=0.6cm,
        thick
    }
]

\foreach \i/\sz in {0/18,1/11,2/1,3/1,4/1,5/1,6/1,7/4} {
  \pgfmathsetmacro{\xpos}{\i*\hLength/7}
  \coordinate (c\i) at (\xpos,\vGap);
}
 
\foreach \i/\sz in {0/10,1/1,2/4,3/1,4/1,5/18,6/1,7/1,8/1} {
  \pgfmathsetmacro{\ypos}{\i*\hLength/8}
  \coordinate (d\i) at (\ypos,0);
}
 
\foreach \i/\j/\k in {
    0/0/4,
    0/1/1,
    0/2/4,
    0/3/1,
    0/4/1,
    0/5/1,
    1/5/4,
    2/5/1,
    3/5/1,
    4/5/1,
    5/5/1,
    6/5/1,
    7/5/1,
    7/6/1,
    7/7/1,
    7/8/1
} {
    \path (c\i);
    \pgfgetlastxy{\cx}{\cy}
    \path (d\j);
    \pgfgetlastxy{\dx}{\dy}
 
    \pgfmathsetmacro{\vx}{\dx-\cx}
    \pgfmathsetmacro{\vy}{\dy-\cy}
    \pgfmathsetmacro{\len}{veclen(\vx,\vy)}
 
    \pgfmathsetmacro{\ux}{-\vy/\len}
    \pgfmathsetmacro{\uy}{\vx/\len}
 
    \foreach \t in {1,...,\k} {
        \pgfmathsetmacro{\off}{(\t-(\k+1)/2)*0.05}
 
        \draw[thick, gray!50]
            ($ (c\i) + (\off*\ux,\off*\uy) $)
            --
            ($ (d\j) + (\off*\ux,\off*\uy) $);
    }
}

\draw[line width=1.5pt, red] (c0) -- (d5);
\draw[line width=1.5pt, red] (c7) -- (d5);
 
\foreach \i/\sz in {0/18,1/11,2/1,3/1,4/1,5/1,6/1,7/4} {
  \ifnum\sz>1
    \node[box] at (c\i) {\sz};
  \else
    \fill (c\i) circle (2pt);
  \fi
    \node at ($(c\i)+(0,0.5)$) {\textcolor{gray}{\scriptsize \i}};
}
 
\foreach \i/\sz in {0/10,1/1,2/4,3/1,4/1,5/18,6/1,7/1,8/1} {
  \ifnum\sz>1
    \node[box] at (d\i) {\sz};
  \else
    \fill (d\i) circle (2pt);
  \fi
    \node at ($(d\i)-(0,0.5)$) {\textcolor{gray}{\scriptsize \i}};
}
 
\end{tikzpicture}
        \end{adjustbox}
        \caption{A monotone alignment.}
        \label{subfig:monotone_74}
    \end{subfigure}
    \hfill
    \begin{subfigure}[t]{0.48\textwidth}
        \centering
        \begin{adjustbox}{height=3cm}
\def\hLength{5}   
\def\vGap{1.5}       

\begin{tikzpicture}[
  box/.style={
    draw,
    fill=white,
    rounded corners,
    minimum width=0.6cm,
    minimum height=0.6cm,
    thick
  }
]
 
 
\foreach \i/\sz in {0/18,1/11,2/1,3/1,4/1,5/1,6/1,7/4} {
  \pgfmathsetmacro{\xpos}{\i*\hLength/7}
  \coordinate (c\i) at (\xpos,\vGap);
}
 
\foreach \i/\sz in {0/10,1/1,2/4,3/1,4/1,5/18,6/1,7/1,8/1} {
  \pgfmathsetmacro{\ypos}{\i*\hLength/8}
  \coordinate (d\i) at (\ypos,0);
}
 
\foreach \i/\j/\k in {
    0/5/4,
    1/0/4,
    1/4/1,
    2/1/1,
    3/3/1,
    4/6/1,
    5/8/1,
    6/7/1,
    7/2/4
} {
    \path (c\i);
    \pgfgetlastxy{\cx}{\cy}
    \path (d\j);
    \pgfgetlastxy{\dx}{\dy}
 
    \pgfmathsetmacro{\vx}{\dx-\cx}
    \pgfmathsetmacro{\vy}{\dy-\cy}
    \pgfmathsetmacro{\len}{veclen(\vx,\vy)}
 
    \pgfmathsetmacro{\ux}{-\vy/\len}
    \pgfmathsetmacro{\uy}{\vx/\len}
 
    \foreach \t in {1,...,\k} {
        \pgfmathsetmacro{\off}{(\t-(\k+1)/2)*0.05}
 
        \draw[thick, gray!50]
            ($ (c\i) + (\off*\ux,\off*\uy) $)
            --
            ($ (d\j) + (\off*\ux,\off*\uy) $);
    }
}
 
\draw[line width=1.5pt, red] (c2) -- (d1);
\draw[line width=1.5pt, red] (c5) -- (d8);
 
\foreach \i/\sz in {0/18,1/11,2/1,3/1,4/1,5/1,6/1,7/4} {
    \node at ($(c\i)+(0,0.5)$) {\textcolor{gray}{\scriptsize \i}};
  \ifnum\sz>1
    \node[box] at (c\i) {\sz};
  \else
    \fill (c\i) circle (2pt);
  \fi
}
 
\foreach \i/\sz in {0/10,1/1,2/4,3/1,4/1,5/18,6/1,7/1,8/1} {
  \ifnum\sz>1
    \node[box] at (d\i) {\sz};
  \else
    \fill (d\i) circle (2pt);
  \fi
    \node at ($(d\i)-(0,0.5)$) {\textcolor{gray}{\scriptsize \i}};
}
 
\end{tikzpicture}
        \end{adjustbox}
        \caption{A non-monotone \& non-antitone alignment.}
        \label{subfig:non_monotone_74}
    \end{subfigure}
    
    \caption{Two alignments between the same two multisets given over $\{0, \ldots, 8\}$. 
    Each multiset is visualised horizontally, with the multiplicities of each integer denoted by black dots and boxes: Dots correspond to multiplicity of one, a box labelled $k$ denotes a multiplicity of $k$.
    Both multisets have $38$ elements.
    Alignments are visualised by grey edges; an edge connects two bags $B, B'$ if there is an $x \in B$ with $\pi(x) \in B'$. 
    Red edges indicate pairs of maximum distortion under the respective alignment.
    }
    \label{fig:non_monotone_caterpillar_space_74}
\end{figure}

\begin{example}
    \label{example:caterpillar-74-example}
    Consider the multisets depicted in \cref{fig:non_monotone_caterpillar_space_74}, where
    we use the same conceptual visualisation as in \cref{fig:non_monotone_caterpillar_space}.
    We show that the bijection of \cref{subfig:non_monotone_74} achieves a distortion of $4$,
    while both monotone and antitone alignments incur a distortion of $7$.
    Employing the translation of multisets into subsets from 
    \cref{corr:multiset_to_set}, it follows that for every 
    $\eps > 0$ there are pairs of subsets of $\R$ for which the distortion of the optimal bijection
    differs by a factor of at least $\frac{7}{4}-\eps$ from the distortion of the better of the monotone and 
    antitone bijections.
    
    In the following, we refer to maximal collections of points with the same value as bags. The monotone alignment (\cref{subfig:monotone_74}) incurs a distortion of $7$ since the left- and 
    rightmost bags of the top space both contain a point that is 
    mapped into the bag of size $18$ in the bottom space. Similarly, any antitone alignment
    incurs a distortion of $7$, because it spreads the two leftmost bags of the top space
    over the whole bottom space, so to a distance of $8$. 
    In contrast, the non-monotone and non-antitone alignment (\cref{subfig:non_monotone_74}) identifies
    the two bags containing $18$ points each and relies on an intertwined structure to distribute the
    remaining bags. Any two points in the top space at distance at least $4+i, \, i \in [4]$, have 
    their images in the bottom space at distance at least $i$ and vice versa. Hence, the maximum 
    distortion of the alignment is at most $4$. Suppose there exists an alignment with maximum distortion
    at most $3$. Then the image of the first bag of the top space can span across at most four consecutive bags,
    implying that some point of that bag is mapped into the bag with $18$ vertices
    in the bottom space. Every preimage of the first bag of the bottom space must then avoid
    the first two bags of the top space, which is impossible because this leaves only nine
    available points for the preimages. 
\end{example}

Recall that the length of a path decomposition of a graph $G$ is the largest diameter of its bags measured with respect to the shortest-path metric of $G$.
The pathlength of $G$ is the minimum length taken over all path decompositions of $G$.
Deciding whether a graph has pathlength at most two is already \NP-complete \cite[Theorem 3]{np_hardness_path_length} and thus computing a path decomposition of optimal length is \NP-hard.
However, a path decomposition of length at most twice the optimum can be computed in time $\O(n^3)$ \cite[Algorithm 6, Theorem 3]{Dragan2017_pathlength_approx}. 

\begin{proof}[Proof of \cref{theorem:max_distortion_approx_2_4}]
    We start with the proof of (1). Given two connected $n$-vertex graphs $G_1, G_2$ of pathlength at most $k$, first we invoke the algorithm of Dragan, Köhler, and Leitert \cite[Algorithm 6, Theorem 3]{Dragan2017_pathlength_approx} for computing a path decomposition of at most twice the optimal length in time $\O(n^3)$.
    For a given graph $F$ of pathlength at most $k$, their algorithm considers for each start vertex $s \in V(F)$ a path decomposition obtained from a breadth-first search \enquote{layering} starting at $s$, and, among these path decompositions, returns one of smallest length.
    We work with the corresponding layering directly.
    As Dragan and Köhler describe in \cite[Section 3.1]{dragan2025graphparameterscoarselyequivalent_pathlength},
    this layering can be turned into a \enquote{canonical} caterpillar $T$ satisfying
    \begin{equation}
        \label{equation:canonical_diff}
        \Gdist_{T}(x, x') - 2 \leq \Gdist_F(x, x') \leq \Gdist_T(x, x') + 2k
    \end{equation}
    for all $x, x' \in V(F)$ by \cite[Lemma 2]{dragan2025graphparameterscoarselyequivalent_pathlength} where we, again slightly abusing notation, identify the vertices of $T$ with the corresponding vertices of $F$.

    Let $T_1, T_2$ be the caterpillars obtained from $G_1, G_2$ in this way.
    For $x, x' \in V(G_1)$ and $\pi \in \Bij(V(G_1), V(G_2))$, we have
    \begin{equation}
        \label{equ:graph_caterpillar_dist_approx}
        \Big| \big|\Gdist_{T_1}(x, x') -\Gdist_{T_2}(\pi(x), \pi(x'))\big| -
        \big|\Gdist_{G_1}(x, x') -\Gdist_{G_2}(\pi(x), \pi(x'))\big| \Big| \leq  2k + 2
    \end{equation}
    by applying \cref{equation:canonical_diff} twice.
    As a second step, we run the algorithm $\mathcal B$ of \cref{lemma:max_distortion_caterpillar_approx} on $T_1$ and $T_2$.
    Let $\mathcal C$ be the algorithm obtained from combining these two steps.
    Then
    \begin{align*}
        \Delta^{\mathcal C}(G_1, G_2) &\leq \Delta^{\mathcal B}(T_1, T_2) + 2k + 2 \tag*{\cref{equ:graph_caterpillar_dist_approx}}\\
        &\leq \mfrac{7}{4}\Delta^\ast(T_1, T_2) + \mfrac{11}{2} + 2k + 2 \tag*{\cref{lemma:max_distortion_caterpillar_approx}}\\
        &\leq \mfrac{7}{4}\big(\Delta^\ast(G_1, G_2) + 2k + 2 \big) + \mfrac{11}{2} + 2k + 2 \\
        &= \mfrac{7}{4}\Delta^\ast(G_1, G_2) + \mfrac{11}{2}k + 11,
    \end{align*}
    so $\mathcal C$ is a $\big(\frac{7}{4}, \frac{11}{2}k + 11\big)$-approximation algorithm. 

    We now prove (2) by transferring \cite[Theorem 6]{approximating_the_distortion} to our setting.
    Let $G, H$ be the input graphs where we assume \WLOG{} that $H$ has pathlength at most $k$.
    We first compute a layering of $H$ and a corresponding canonical caterpillar $T_H$ as before. Let $M \ceq M(T_H)$ be the corresponding multiset.
    Combining \cref{equ:caterpillar_to_multiset,equation:canonical_diff}, for all $\pi \in \Bij(V(G), V(H))$, we have
    \begin{equation}
        \label{equ:graph_multiset_approx}
        \Big| \big| \Gdist_G(x, x') - \Gdist_H(\pi(x), \pi(x')) \big| - \big| \Gdist_G(x, x') - |\pi(x) - \pi(x')| 
        \big| \Big| \leq 2k + 2
    \end{equation}
    where we again identify vertices of the graph $H$ with their corresponding multiset elements.
    
    We extend the notion of maximum distortion between graphs and multisets in the obvious way.
    To compute an approximately optimal alignment between $G$ and $M$,
    we closely 
    follow the algorithm \textsc{Intervals} of Hall and Papadimitriou 
    \cite[Section 3.2]{approximating_the_distortion} (we denote their algorithm by $\mathcal D$).
    Combined with the previous step, this yields an algorithm $\mathcal E$ with a total running time of $\O(n^3)$.

    Let $w$ be a point in the leftmost bag of $M$. 
    The algorithm $\mathcal D$ considers for each $x \in V(G)$ the partial map $x \mapsto w$ and greedily extends it by mapping the vertices $x' \in V(G) \setminus \{x\}$ ordered increasingly by $\Gdist_G(x, x')$ and matching each to an unmatched point in the leftmost bag of $M$ (that still contains unmatched points).
    This yields, for each $x \in V(G)$, a bijection $\pi_x$. Our algorithm $\mathcal E$ returns a bijection of least maximum distortion among these.

    Combining everything, we have
    \begin{align*}
        \Delta^{\mathcal E}(G, H) &\leq \Delta^{\mathcal D}(G, M) + 2k + 2 \tag*{\cref{equ:graph_multiset_approx}}\\
        &\leq 5\Delta^\ast(G, M) + 2k + 2 \tag*{\cite[Theorem 6]{approximating_the_distortion}}\\
        &\leq 5\big(\Delta^\ast(G, H) + 2k + 2\big) + 2k + 2 \\
        &= 5\Delta^\ast(G, H) + 12k + 12.
    \end{align*}
 
    We note that the second inequality requires extending the result \cite[Theorem 6]{approximating_the_distortion} to the setting where one input is a graph and the other a multiset. This is achieved by shifting the points of $M$ apart as in the 
    proof of \cref{corr:multiset_to_set}, yielding a set $M^\eps$. Since distances in $M$
    and $M^\eps$ differ by at most $\eps$, the distortion of every bijection with
    respect to $M$ and $M^\eps$ differ by at most $\eps$. By running $\mathcal D$
    on $M$ and breaking ties inside bags according to the shifts (equivalently according to the given enumerations) one obtains
    \begin{align*}
        \Delta^{\mathcal D}(G, M) 
        \leq \Delta^{\mathcal D}(G, M^\eps) + \eps
        \leq 5 \Delta^\ast (G, M^\eps) + \eps
        \leq 5 \Delta^\ast (G, M) + 6\eps.
    \end{align*}
    Both $\Delta^{\mathcal D} (G, M)$ and $\Delta^\ast (G, M)$ are integers, meaning for $\eps < \frac{1}{6}$ it holds that
    $\Delta^{\mathcal D}(G, M) \leq 5 \Delta^\ast (G, M)$. 
    
    By precomputing $D_G$ and $d_T$ (the metric on $M$), we can compute $\pi_x$ and its distortion in time $\O(n^2)$; hence, computing an approximately optimal alignment between $G$ and $M$ can be done in time $\O(n^3)$. As all previous steps also require time $\O(n^3)$ this yields the desired time bound.
\end{proof}

\section{\texorpdfstring{Proofs Omitted from \cref{section:beyond_max_distortion}}{Proofs Omitted from Section~\ref{section:beyond_max_distortion}}}
\label{appendix:beyond_max_distortion}
\subsection{Local Distortion}
\begin{proof}[Proof of \cref{theorem:local_dist_hardness_hamiltonian}]
    We reduce from \hamcycle{} which remains \NP-complete even for $3$-regular graphs \cite{hamcycle_threeregular_np_complete}. Given a $3$-regular graph $G$ on $n$ vertices as input for \hamcycle{}, the reduction uses the $(n+1)$-vertex graphs $\widehat{C_n}$ and $\widehat{G}$ as inputs for $\Ddist{p}$ with threshold 1.
    Recall that $b(p) = \max\{\frac{1}{p}, 1-\frac{1}{p}\}$ for $1 \leq p < \infty$ and $b(\infty) = 1$.

    Assume that $G$ admits a Hamiltonian cycle. Then there exists a bijection $\pi: V(C_n) \to V(G)$ that aligns the cycle $C_n$ perfectly with the Hamiltonian cycle in $G$. We extend $\pi$ to a bijection $\widehat{\pi}$ between $\widehat{C_n}$ and $\widehat{G}$ by mapping $c_{C_n}$ to $c_G$. Each vertex of $G$ has three neighbours in $V(G)$ in $\widehat{G}$, two of which are matched correctly by $\widehat{\pi}$ as they are part of the Hamiltonian cycle. Hence, every vertex of $G$ has a mismatch of one. 
    Using \cite[Lemma 14]{gervens2022graphsimilaritybasedmatrix} together with \cref{observation:diameter_two}, it follows that $\Ddist{p}(\widehat{C_n}, \widehat{G}) \leq 1$.

    Now assume that $G$ has no Hamiltonian cycle. 
    Any bijection $\pi' \in \Bij(V(\widehat{C_n}), V(\widehat{G}))$ with $\pi'(c_{C_n}) \neq c_G$ has $\dis_{p}(\pi') \geq (n-4)^{b(p)} \geq 3^{b(p)}$ given that $n \geq 8$ due to the degree mismatch between $c_{C_n}$ and $\pi'(c_{C_n})$.
    Hence, we can restrict to the graphs $C_n, G$ and consider any bijection $\sigma \in \Bij(V(C_n), V(G))$ with $\widehat{\sigma} = \pi'$ instead. 
    Since $G$ is not Hamiltonian, there must be an edge $vv' \in E(C_n)$ such that $\sigma(v)\sigma(v')\notin E(G)$. Moreover, only one of the edges incident to $\sigma(v) \in V(G)$ can be matched correctly. It follows that $\sigma(v)$ has at least three edge mismatches under $\sigma$. Hence, by \cite[Lemma 13]{gervens2022graphsimilaritybasedmatrix} and \cref{observation:diameter_two} it holds $\dis_{p}(\pi') \geq 3^{b(p)}$. Thus $\Ddist{p}(\widehat{C_n}, \widehat{G}) \geq 3^{b(p)}$.
\end{proof}

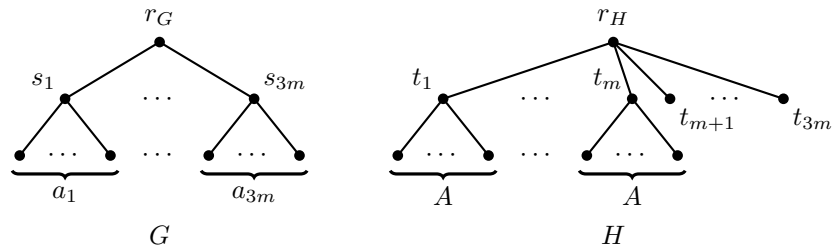
\begin{figure}[b!]
    \captionsetup{singlelinecheck = true, justification=justified}
    \centering
    
    \begin{tikzpicture}[
            box/.style={
                draw,
                fill=white,
                rounded corners,
                minimum width=1.6cm,
                minimum height=0.8cm,
                thick
            },
            dot/.style={
                draw,
                fill=black,
            }
        ]
    
        \newcommand{\branchleaves}[4]{%
            \coordinate (#1root) at (#2, #3);
            \coordinate (#1nstart) at (#2-0.6,#3-0.75);
            \coordinate (#1nend)   at (#2+0.6,#3-0.75);
    
            \fill (#1nstart) circle (2pt);
            \fill (#1nend)   circle (2pt);
        
            \coordinate (#1mid) at ($(#1nstart)!0.5!(#1nend)$);
            \node at (#1mid) {$\cdots$};
        
            \draw[thick] (#1root) -- (#1nstart);
            \draw[thick] (#1root) -- (#1nend);
        
            \draw[decorate,decoration={brace},very thick]
                ($(#1nend)+(0.1,-0.2)$)
                --
                ($(#1nstart)+(-0.1,-0.2)$)
                node[midway,below=2pt] {#4};
        }
    
        \newcommand{\kpath}[5]{%
    
            \begin{scope}[shift={(#2,#3)}]
        
                \coordinate (#1start) at (0,0);
                \coordinate (#1end) at (0,-#4);

                \def\pad{0.2}
        
                \draw[rounded corners, fill=white]
                    (-\pad,\pad) rectangle (\pad,-#4-\pad);
                    
                \fill (#1start) circle (2pt);
                \fill (#1end)   circle (2pt);
        
                \node[rotate=90] (mid) at ($(#1start)!0.5!(#1end)$) {$\cdots$};
                \node[right=4pt] at (mid) {#5};
        
        
            \end{scope}
        }
    
        \coordinate (G) at (0,0);
        \coordinate (H) at (6,0);
        \begin{scope}[shift=(G)]
    
            \node at (1.5, -2.8) {$G$};
        
            \coordinate (rG) at (1.5,-0.25);
            \fill (rG) circle (2pt);
            \node[above=3pt] at (rG) {$r_G$};
            \draw[thick] (0.25,-1) -- (rG);
            \draw[thick] (2.75,-1) -- (rG);
            \branchleaves{a1}{0.25}{-1}{$a_1$};
            \node at (1.5,-1) {$\cdots$};
            \node at (1.5,-1.75) {$\cdots$};
            \branchleaves{a3m}{2.75}{-1}{$a_{3m}$};
            \node[above left] at (0.25, -1) {$s_1$};
            \node[above right] at (2.75, -1) {$s_{3m}$};
            \fill (0.25, -1) circle (2pt);
            \fill (2.75, -1) circle (2pt);
    
        \end{scope}

        \begin{scope}[shift=(H)]
    
            \node at (1.5, -2.8) {$H$};
            
            \coordinate (rH) at (1.5,-0.25);
            \coordinate (t1) at (-0.75, -1);
            \coordinate (tm) at (1.75, -1);
            \coordinate (tm1) at (2.25, -1);
            \coordinate (t3m) at (3.75, -1);
            \fill (rH) circle (2pt);
            \node[above=3pt] at (rH) {$r_H$};
            \draw[thick] (t1) -- (rH);
            \draw[thick] (tm) -- (rH);
            \branchleaves{a1}{-0.75}{-1}{$A$};
            \node at (0.5,-1) {$\cdots$};
            \node at (0.5,-1.75) {$\cdots$};
            \branchleaves{a3m}{1.75}{-1}{$A$};
            \node[above left] at (t1) {$t_1$};
            \node[above left] at (tm) {$t_{m}$};
            \node[below right] at (tm1) {$t_{m+1}$};
            \node[below right] at (t3m) {$t_{3m}$};
            \fill (t1) circle (2pt);
            \fill (tm) circle (2pt);
            \fill (tm1) circle (2pt);
            \fill (t3m) circle (2pt);
            
            \node at (3,-1) {$\cdots$};

            \draw[thick] (tm1) -- (rH);
            \draw[thick] (t3m) -- (rH);
        \end{scope}
    \end{tikzpicture}
    
    \caption{The trees in the construction of the proof of \cref{theorem:local_dist_mult_hardness_trees}.}
    \label{figure:local_hardness_trees}
\end{figure}

\begin{proof}[Proof of \cref{theorem:local_dist_mult_hardness_trees}]
    We again reduce from \threepartition{}.
    Note that we may \WLOG{} assume that all instances have $m \geq 3$.
    Further, by strong \NP-hardness we can fix some $d \geq 1$ and restrict our attention to instances 
    satisfying $A \leq m^d$, ensuring that our reduction can be computed in polynomial time. 
    For an instance $\mathcal I = (A, a_1, \ldots, a_{3m})$ we
    construct graphs $G$ and $H$ as follows.
    For $G$, start with the disjoint union of $3m$ stars, where the $i$-th star has $a_i$ leaves. The centre of the $i$-th star will be denoted by $s_i$.
    Then $G$ is obtained from the disjoint union of these stars by adding a fresh vertex $r_G$, in the following also referred to as the root of $G$, which is connected to each $s_i$ for $i \in [3m]$.
    The graph $H$ is constructed from $m$ disjoint copies of stars with $A$ leaves each, whose centres we will refer to as $t_1, \ldots, t_m$ and call them the roots of the bags of $H$. Moreover, we add $2m$ isolated vertices which are denoted $t_{m+1}, \ldots, t_{3m}$. To complete the construction of $H$, again add a root vertex $r_H$ which is connected to each $t_j$ for $j \in [3m]$.
    The constructed graphs are trees of diameter four.

    For $i \in [3m]$, we refer to the leaves attached to $s_i$ as the \emph{leaves of $s_i$} and write $L(s_i)$ for the corresponding set.
    Similarly, for $j \in [m]$, the \emph{leaves of $t_j$} are the set $L(t_j)$ of leaves attached to $t_j$. We write $L[v] \ceq \{v\} \cup L(v)$ where $v = s_i$ or $v = t_j$.
    This notation naturally extends to sets.
    Finally, we let $S \ceq \{ s_i \mid i \in [3m]\}$, $T \ceq \{ t_j \mid j \in [m]\}$, and $U \ceq \{ t_j \mid m < j \leq 3m\}$.

    We claim that $\Ddist{1}(G, H) \leq \frac{3}{2}A + 4$ if $\mathcal I$ is a YES-instance and $\Ddist{1}(G, H) \geq \frac{7}{4}A$ otherwise.
    For notational convenience, for a vertex $v \in V(G)$ and a bijection $\pi \in \Bij(V(G), V(H))$, define
    \begin{equation}
        \label{equation:dis_1_pi_v}
        \dis_{1}(\pi, v) \ceq \sum_{v' \in V(G)} \dis(\pi, v, v')
    \end{equation}
    as the sum of distortions incurred relative to the vertex $v$.
    By taking the maximum over all choices of $v$, we have the identity
    $\dis_1(\pi) = \max_{v \in V(G)} \dis_1(\pi, v)$.

    \yescase Suppose $\mathcal I$ is a YES-instance.
    Let $\sigma: [3m] \to [m]$ be a witnessing assignment and define a bijection $\pi \in \Bij(V(G), V(H))$ satisfying the following properties:
    \begin{itemize}
        \item $\pi(r_G) = r_H$,
        \item $\pi(S) = T \cup U$, and thus in particular also $\pi(L(S)) = L(T)$.
        \item For each $j \in [m]$ and $\sigma^{-1}(j) = \{i_1, i_2, i_3\}$, $\pi(L(s_{i_1}) \cup L(s_{i_2}) \cup L(s_{i_3})) = L(t_j)$, and $\pi(s_i) = t_j$ for some $i \in \sigma^{-1}(j)$.
    \end{itemize}
    \begin{claim}
        For $\pi$ specified as above, it holds that $\dis_1(\pi) \leq \frac{3}{2}A + 4$.
    \end{claim}
    \begin{claimproof}
        Let $v \in V(G)$.
        If $v = r_G$, then $\dis_1(\pi, r_G) = 0$ since $\pi$ preserves the level of each vertex in the two trees.
        Next, assume $v \in S$ and let $i \in [3m]$ such that $v = s_i$.
        
        If $\pi(v) \in T$, that is, $v$ is mapped to the root of a bag in $H$, then $v$ incurs no distortion to any of its leaves in $G$ because $H[\pi(L[v])]$ is a subtree of $H$.
        Clearly, there is no distortion to any vertex of $S \cup \{r_G\}$.
        Let $j \neq k \in [3m]$ be distinct from $i$ such that $\sigma(i) = \sigma(j) = \sigma(k)$.
        Then $v$ incurs distortion of two to each $v' \in L(s_j) \cup L(s_k)$ since for each such $v'$ it holds that $\Gdist_G(v, v') = 3$ but $\Gdist_H(\pi(v), \pi(v')) = 1$.
        Hence, the total amount of distortion is at most $2\cdot |L(s_j) \cup L(s_k)| = 2(a_j + a_k) = 2(A - a_i) < \frac{3}{2}A$ where the last inequality holds since $a_i > \frac{A}{4}$.
        
        If $\pi(v) \in U$, then $v$ only incurs distortion of two to every vertex in $L(v)$, hence at most $2\cdot |L(s_i)| = 2a_i < A$ where the last inequality holds by $a_i < \frac{A}{2}$.

        Finally, let $v$ be a leaf vertex and let $i \in [3m]$ such that $v \in L(s_i)$.
        Again, let $j \neq k \in [3m]$ be distinct from $i$ such that $\sigma(i) = \sigma(j) = \sigma(k)$.
        One easily verifies that $v$ incurs a distortion of two to every vertex $v' \in L(s_j) \cup L(s_k)$.
        Moreover, the only other vertices whose distances to $v$ may change under $\pi$ are in $\{s_i, s_j, s_k\}$.
        If $\pi(v)\pi(s_i) \in E(H)$, then the distances to all three vertices are preserved under $\pi$. Otherwise, the distance between $v$ and $s_i$ and between $v$ and one of $s_j$ and $s_k$ change by two, respectively.
        In total, again using $a_i > \frac{A}{4}$, this yields a total sum of at most $\frac{3}{2}A + 4$.
    \end{claimproof}
    
    \nocase Suppose $\mathcal I$ is a NO-instance and assume towards a contradiction that there is a bijection $\pi \in \Bij(V(G), V(H))$ with $\dis_1(\pi) < \frac{7}{4}A$.
    We start by showing that under this assumption, any such bijection necessarily needs to preserve layers, \ie{} map $r_G$ to $r_H$ and $S$ to $T \cup U$.

    \begin{claim}
        If $\pi(r_G) \neq r_H$ or $\pi(S) \neq T \cup U$, then
        $\dis_1(\pi) \geq 2A$.
    \end{claim}
    \begin{claimproof}
        First suppose that $\pi(r_G) = x$ for some $x \neq r_H$.
        Note that $\Gdist_G(r_G, v) \leq 2$ for all $v \in V(G)$; however, the set $\{ y \in V(H) \mid \Gdist_H(x, y) > 2 \}$ has size at least $(m-1)A \geq 2A$ for $m \geq 3$. So assume $\pi(r_G) = r_H$ and $\pi(S) \neq T \cup U$.
        This in particular implies that there is a $v \in L(S)$ with $\pi(v) \in T \cup U$.
        Let $i \in [3m]$ such that $v \in L(s_i)$.
        However, vertices in $L(S)$ have even distances between each other in $G$, whereas $\pi(v)$ has odd distances to all vertices in $L(T)$.
        Hence, the parity of distances to $v$ changes for at least $(mA - 1) - (3m-1)$ vertices under $\pi$.
        Finally, note that $v$ only has exactly $a_i$ vertices at distance two from it in $G$ and hence $(3m-1) - a_i$ vertices necessarily change distance to $v$ by at least one.
        In total, this yields $\dis_1(\pi, v) \geq mA - 1 - a_i \geq (m-1)A \geq 2A$ since $a_i < \frac{A}{2}$ and $m \geq 3$.
    \end{claimproof}

    From the above, we may, from now on and until the end of this proof, assume that $\pi(S)  = T\cup U$ and hence $\pi(L(S)) = L(T)$.

    Suppose there is $i \in [3m]$ such that $\pi(L(s_i)) \not\subseteq L(t_j)$ for all $j \in [m]$, that is, $L(s_i)$ is split across multiple leaf sets of $H$.
    Then there exists a vertex $v \in L(s_i)$
    and $j \in [m]$ such that $\pi(v) \in L(t_j)$ with the property that at least half of the vertices of $L(s_i)$ are mapped to $L(T) \setminus L(t_j)$ under $\pi$.
    Then $v$ incurs a distortion of two to every vertex that is mapped into $L(t_j)$ but not contained in $L(s_i)$ in $G$. Since at most $\frac{a_i}{2}$ vertices of $L(s_i)$ are mapped into $L(t_j)$, the number of such vertices is at least $A - \frac{a_i}{2} > \frac{3}{4} A$ since $a_i < \frac{A}{2}$.
    Additionally, $v$ incurs a distortion of two to every vertex in $L(s_i)$ that is mapped to $L(T) \setminus L(t_j)$. The number of such vertices is at least $\frac{a_i}{2} > \frac{A}{8}$ since 
    $a_i > \frac{A}{4}$.
    Thus $\dis_1(\pi) \geq \dis_1(\pi, v) > 2\big(\frac{3}{4}A + \frac{1}{8}A\big) = \frac{7}{4}A$.

    Hence, we may assume that, for each $i \in [3m]$, there is a $j \in [m]$ such that $\pi(L(s_i)) \subseteq L(t_j)$.
    Let $\sigma: [3m] \to [m]$ be the corresponding assignment function.
    Then $\sum_{i \in \sigma^{-1}(j)} a_i \leq A$ for each $j \in [m]$, which implies that $\mathcal I$ is a YES-instance, a contradiction.
    Hence $\Ddist{1}(G, H) > \frac{7}{4}A$.\bigskip

    \noindent The above gap in $\dis_1$-values between YES- and NO-instances shows that there cannot exist a polynomial-time approximation algorithm with a multiplicative error guarantee of $\frac{7}{6}-\eps$ since for each fixed $\eps > 0$, for sufficiently large $A$ we have
    \begin{equation}
        \label{equation:local_tree_eps}
        \frac{\frac{7}{4}A}{\frac{3}{2}A + 4} > \mfrac{7}{6} - \eps.
    \end{equation}
    Hence, a $\big(\frac{7}{6}-\eps\big)$-approximation algorithm would decide \threepartition{}, in contradiction to the assumption $\P \neq \NP$.
    Note that while $A$ is part of the input instance, we can always multiply the values of $\mathcal I$ 
    by a constant factor (depending on $\eps$) to ensure that \cref{equation:local_tree_eps} holds.

    For the additive approximation hardness, we again replicate the leaves in both trees $c = m^b$ times 
    (compare \cref{appendix:hardness_max_distortion}) for an integer $b > 0$. Recall that $A \leq m^d$.
    The resulting graphs have $n \ceq 1 + 3m + cmA = \O(m^{b+d+1})$ vertices and for $A \geq 32$ 
    (which we may assume \WLOG{} by simply considering the instance $32\mathcal I$)
    and $c \geq 1$, the gap between YES- and NO-instances is at least 
    \begin{align*}
        \mfrac{7}{4}cA - \left(\mfrac{3}{2}cA + 4\right)
        = \mfrac{1}{4}cA - 4 \geq \mfrac{1}{8}cA
        = \Omega\big(m^b\big)
        = \Omega\big(n^{b/(b+d+1)}\big).
    \end{align*}
    Since $\lim_{b \to \infty} \frac{b}{b + d + 1} = 1$, for every $\eps > 0$, we can choose $b$
    sufficiently large to ensure $\frac{b}{b+d+1} > 1-\eps$, implying that the gap lies in $\omega(n^{1-\eps})$.
    Hence, approximating $\Ddist{1}$ with additive error guarantee $\O(n^{1-\eps})$ remains \NP-hard.
\end{proof}

\subsection{Global Distortion}
We start with the following simple lower bound on the global distortion.
\begin{claim}
    \label{claim:lower_bound_distortion_sum}
    Let $1 \leq p < \infty$ and let $G, H$ be connected graphs on the same number of vertices, $\pi \in \Bij(V(G), V(H))$ with $\pi(v) = w$.
    Then $\dis_{(p)}(\pi) \geq |\deg_G(v) - \deg_H(w)|^{1/p}$.
\end{claim}
\begin{claimproof}
    We may assume \WLOG{} that $\deg_G(v) \geq \deg_H(w)$, otherwise consider $\pi^{-1}$.
    We have 
    \begin{align*}
        \dis_{(p)}(\pi) &= \left(\sum_{v', v'' \in V(G)} \big|\Gdist_G(v', v'') - \Gdist_H(\pi(v'), \pi(v''))\big|^p \right)^{1/p} \\
        &\geq \left( \sum_{v' \in V(G)} \big|\Gdist_G(v, v') - \Gdist_H(w, \pi(v'))\big|^p \right)^{1/p} \\
        &\geq \left( \sum_{\substack{v' \in N_G(v)\\\pi(v') \notin N_H(w)}} \big|\Gdist_G(v, v') - \Gdist_H(w, \pi(v'))\big|^p \right)^{1/p} \geq \big(\deg_G(v) - \deg_H(w) \big)^{1/p}
    \end{align*}
    which proves the claim.
\end{claimproof}
We also note that for graphs $G, H$, and all $\pi \in \Bij(V(G), V(H))$ one has $\dis_{(1)}(\pi) \leq \big(\dis_{(p)}(\pi)\big)^p$. Equality holds if $\dis(\pi, v, v') \in \{0, 1\}$ for all $v, v' \in V(G)$.

\begin{proof}[Proof of
    \cref{theorem:additive_inapprox_global_distortion}]
    The proof is by reduction from \textsc{Hamiltonian Cycle}, which is already \NP-complete for 3-regular graphs \cite{hamcycle_threeregular_np_complete}. 
    Given a 3-regular graph $F$ on $n$ vertices, we construct two graphs $G, H$ from $C_n$ and $F$ as follows.
    
    Let $X \in \{C_n, F\}$ and assume $V(X) = \{v_1, \ldots, v_n\}$.
    Then we construct a new graph by replacing each vertex $v \in V(X)$ by a $k$-clique, the vertex set of which is denoted $V_1(v)$.
    Next, introduce a fresh vertex $c_X$ that is connected to all vertices of the $k$-cliques.
    If $vv' \in E(X)$, then each pair of vertices in $V_1(v) \times V_1(v')$ is joined by an edge.
    Next, add a fresh $2n(n+4)k^2$-clique and partition its vertex set into $n$ blocks of size $2(n+4)k^2$ each (the blocks are denoted $L_1, \ldots, L_n$ in the following) and introduce an edge for each pair of vertices in $L_i \times V_1(v_i)$, where $i \in [n]$. We let $V_2(v_i) \ceq L_i$ and, from now on, write $V_2(v_i)$ instead of $L_i$.
    \cref{figure:global_additive_hardness} shows the above transformation on an example input graph.

    \begin{figure}[b!]
    \captionsetup{singlelinecheck = true, justification=justified}
    \centering
    
    \begin{tikzpicture}[
        node/.style={
            draw,
            circle,
            fill=white,
        },
        s/.style={
            inner sep=1pt
        },
        m/.style={
            inner sep=5pt
        },
    ]
    
        \coordinate (offset1) at (0, 0.5);
        \begin{scope}[scale=0.7, shift=(offset1)]
            \coordinate (b1) at (0.2, 0);  
            \coordinate (b2) at (0.2, 2);  
            \coordinate (b3) at (1.4, 1);  
            \coordinate (b4) at (2.6, 1);  
            \coordinate (b5) at (3.8, 2);  
            \coordinate (b6) at (3.8, 0);
            
            \draw[line width=1pt]
                (b1) -- (b2)
                (b1) -- (b3)
                (b1) -- (b6)
                (b2) -- (b3)
                (b2) -- (b5)
                (b3) -- (b4)
                (b4) -- (b5)
                (b4) -- (b6)
                (b5) -- (b6)
            ;
    
            \fill (b1) circle (3pt);
            \fill (b2) circle (3pt);
            \fill (b3) circle (3pt);
            \fill (b4) circle (3pt);
            \fill (b5) circle (3pt);
            \fill (b6) circle (3pt);

            \node at (2, -0.75) {$F$};
        \end{scope}

        \node at (4.75, 1.25) {\Large$\xlongrightarrow{}$};
    
        \coordinate (offset2) at (7, 0);
        \begin{scope}[shift=(offset2)]
            \coordinate (b1) at (0.2, 0);  
            \coordinate (b2) at (0.2, 2);  
            \coordinate (b3) at (1.4, 1);  
            \coordinate (b4) at (2.6, 1);  
            \coordinate (b5) at (3.8, 2);  
            \coordinate (b6) at (3.8, 0);  
            
            \coordinate (c1) at (-1, 3.25);
            \coordinate (c2) at (0.2, 3.25);
            \coordinate (c3) at (1.4, 3.25);
            \coordinate (c4) at (2.6, 3.25);
            \coordinate (c5) at (3.8, 3.25);
            \coordinate (c6) at (5, 3.25);
            
            \coordinate (c) at (2, -1);

            \draw[line width=1.2pt]
                (b1) -- (b2)
                (b1) -- (b3)
                (b1) -- (b6)
                (b2) -- (b3)
                (b2) -- (b5)
                (b3) -- (b4)
                (b4) -- (b5)
                (b4) -- (b6)
                (b5) -- (b6)
            
                (b1) to[bend left=5] (c1)
                (b2) -- (c2)
                (b3) -- (c3)
                (b4) -- (c4)
                (b5) -- (c5)
                (b6) to[bend right=5] (c6)
                
                (b1) to[bend right=10] (c)
                (b2) to[bend right=10] (c)
                (b3) to[bend right=5] (c)
                (b4) to[bend left=5] (c)
                (b5) to[bend left=10] (c)
                (b6) to[bend left=10] (c)
            ;

            \fill (c) circle (2.5pt);
            \node[below=2pt] at (c) {$c_F$};
            
            \node[node,s] at (b1) {$K_k$};
            \node[node,s] at (b2) {$K_k$};
            \node[node,s] at (b3) {$K_k$};
            \node[node,s] at (b4) {$K_k$};
            \node[node,s] at (b5) {$K_k$};
            \node[node,s] at (b6) {$K_k$};
            
            \node[node,m] at (c1) {$K_\ell$};
            \node[node,m] at (c2) {$K_\ell$};
            \node[node,m] at (c3) {$K_\ell$};
            \node[node,m] at (c4) {$K_\ell$};
            \node[node,m] at (c5) {$K_\ell$};
            \node[node,m] at (c6) {$K_\ell$};
            
            \node[
                draw,
                rounded corners,
                fit=(c1)(c6),
                inner sep=20pt,
                label=above:{$K_{n\ell}$}
            ] {};

        \end{scope}
    
    \end{tikzpicture}
    
    \caption{Transforming an example graph $F$ as described in the construction for \cref{theorem:additive_inapprox_global_distortion}. Here $n=6$ and $\ell \ceq 20k^2$. Cliques are indicated by large circles and boxes. Edges between cliques indicate that all pairs of vertices are connected by edges.}
    \label{figure:global_additive_hardness}
\end{figure}
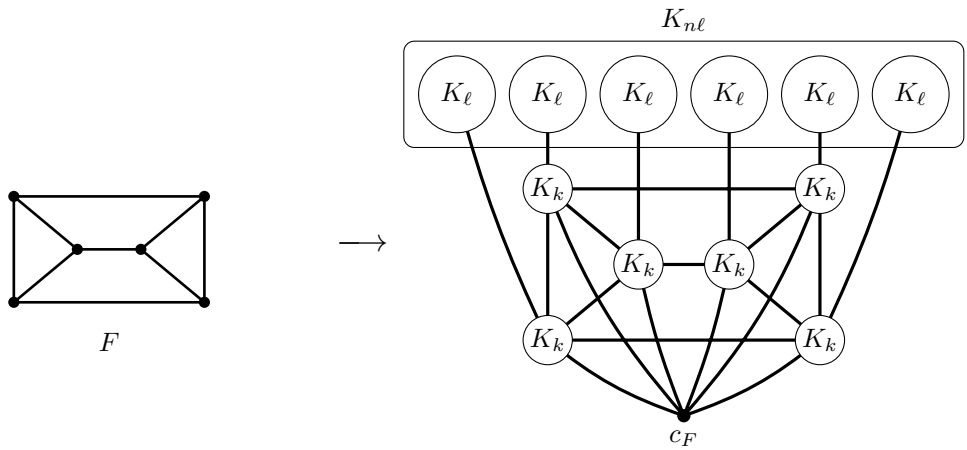
    
    The resulting graphs have $1 + nk + 2n(n+4)k^2$ vertices each, which is polynomial in the input size, given that $k = \poly(n)$. We will choose the correct value of $k$ later in the proof.

    For convenience, we denote the vertex sets of $G$ and $H$ by $V$ and $W$, respectively, and partition the vertex sets based on the distance vertices have to the vertex $c_G$ added in the first augmentation step.
    Specifically, $V = V_0 \uplus V_1 \uplus V_2$ where $V_0 = \{c_G\}$, $V_1 = \bigcup_{i \in [n]} V_1(v_i)$ contains the vertices of the $k$-cliques corresponding to the original vertices of $C_n$, and $V_2$ contains the remaining $2n(n+4)k^2$ vertices. 
    In particular, we have $V_2(v) = N_G(V_1(v)) \cap V_2$.
    We define the same partition on $W$ in the obvious way.

    Before proving correctness, we briefly explain the idea behind our construction: If $F$ admits a Hamiltonian cycle, then the \enquote{thickened} cycle (where each vertex has been replaced by a $k$-clique) can be aligned with the thickened version of $F$ such that each clique has distortion to only one other clique.
    If $F$ does not admit a Hamiltonian cycle, then the large clique in $V_2$ and $W_2$ forces any reasonable alignment to align the cycle with the graph (as if they were not thickened, \ie{} vertices were truly vertices and not $k$-cliques).
    But then there are two cliques with mismatched adjacencies to three other cliques.
    The key point is that, since each clique has $k$ vertices, the sum of distortions will scale quadratically in $k$ since each pair of such vertices from distinct cliques incurs a distortion of one.
    
    It remains to prove correctness.
    We claim that $\Ddist{(p)}(G, H) \leq (nk^2)^{1/p}$ if $F$ admits a Hamiltonian cycle and $\Ddist{(p)}(G, H) \geq ((n+4)k^2)^{1/p}$ otherwise.
    Since both graphs have diameter two, by \cref{observation:diameter_two}, the only nonzero distortions under a bijection $\pi$ have value one and correspond exactly to edge mismatches under $\pi$. Hence, in the following it suffices to count mismatched edges under $\pi$ and then take the $p$-th root to obtain $\dis_{(p)}(\pi)$.
    
    \yescase Suppose $F$ has a Hamiltonian cycle. Then there is a bijection $\sigma: V(C_n) \to V(F)$ aligning the $n$-cycle with $F$ such that every vertex $w \in V(F)$ is incident to two mapped edges of $C_n$. 
    We extend this mapping in the obvious way, mapping $V_1(v)$ to $W_1(\sigma(v))$ and $V_2(v)$ to $W_2(\sigma(v))$, to obtain a sum of distortions of exactly $nk^2$ since each pairwise distortion is either zero or one and the only nonzero distortions are incurred by pairs of vertices coming from cliques corresponding to vertices that are non-adjacent in $\widehat{C_n}$ but whose images are adjacent in $\widehat{F}$.
    Here we use the fact that all pairs of distinct vertices in $V_2$ have a distance of one and, moreover, if $v, v' \in V(C_n)$ are distinct and $x \in V_1(v)$ and $x' \in V_2(v')$, then $\Gdist_G(x, x') = 2$. In particular, no other vertex pair incurs any distortion. It follows that $\dis_{(p)}(\pi) = (nk^2)^{1/p}$.

    \nocase Now suppose that $F$ has no Hamiltonian cycle and assume towards a contradiction that there is a $\pi: V(G) \to V(H)$ with $\dis_{(p)}(\pi) < ((n+4)k^2)^{1/p}$.
    By \cref{observation:diameter_two}, our assumption is equivalent to $\dis_{(1)}(\pi) < (n+4)k^2$.
    We start by showing that $\pi$ must preserve the partitioning of the vertex sets described above.
    \begin{claim}
        \label{subclaim:layer_structure}
        For each $i \in \{0, 1, 2\}$, it holds that $\pi(V_i) = W_i$.
    \end{claim}
    \begin{claimproof}
        We start by observing that
        \[
            \deg_G(v) = \begin{cases}
                nk, & \text{if $v \in V_0$}, \\
                3k + 2(n+4)k^2, & \text{if $v \in V_1$}, \\
                2n(n+4)k^2 -1 + k, & \text{if $v \in V_2$},
            \end{cases}
        \]
        and
        \[
            \deg_H(w) = \begin{cases}
                nk, & \text{if $w \in W_0$}, \\
                4k + 2(n+4)k^2, & \text{if $w \in W_1$}, \\
                2n(n+4)k^2 - 1 + k, & \text{if $w \in W_2$}.
            \end{cases}
        \]
        Hence, if $n, k \geq 2$ and there exist $v \in V_i$ and $w \in W_j$ such that $i \neq j$ and
        $\pi(v) = w$, then by \cref{claim:lower_bound_distortion_sum} we have $\dis_{(1)}(\pi) 
        \geq (n+4)k^2$.
    \end{claimproof}
    Next, we show that the layering structure guaranteed by \cref{subclaim:layer_structure} forces each clique in $V_1$ to be mapped to some clique in $W_1$.
    \begin{claim}
        \label{subclaim:layer_one_bijection}
        For each $v \in V(C_n)$, there is a $w \in V(F)$ such that $\pi(V_1(v)) = W_1(w)$.
    \end{claim}
    \begin{claimproof}
        Suppose for contradiction that there is $v \in V(C_n)$ such that there are distinct $w, w' \in V(F)$ and distinct $x, x' \in V_1(v)$ with $\pi(x) \in W_1(w)$ and $\pi(x') \in W_1(w')$.
        Let $M_x \ceq \{ z \in V_2(v) \mid \pi(z) \in W_2(w) \}$, and similarly, $M_{x'} \ceq \{z \in V_2(v) \mid \pi(z) \in W_2(w')\}$.
        Since $|M_x| + |M_{x'}| \leq 2(n+4)k^2$ it follows that at least one of the two sets contains at most $(n+4)k^2$ vertices.
        \WLOGcap{} we may assume this to be set $M_x$.
        Then $x$ has $2(n+4)k^2$ neighbours in $V_2(v)$ in $G$ but at most $(n+4)k^2$ of these vertices remain neighbours after mapping under $\pi$, all others must be at distance two in $H$ since $\diam(H) = 2$.
        Hence, $\dis_{(1)}(\pi) \geq 2(n+4)k^2 - (n+4)k^2 = (n+4)k^2$ in contradiction to $\dis_{(1)}(\pi) < (n+4)k^2$.
    \end{claimproof}
    Let $\sigma: V(C_n) \to V(F)$ be a bijection induced by \cref{subclaim:layer_one_bijection}.
    Then, for $C_n$ and $F$, we have $\dis_{(1)}(\sigma) < n+4$ as otherwise $\dis_{(1)}(\pi) \geq (n+4)k^2$ because each edge mismatch of $\sigma$ is scaled by $k^2$ when replacing vertices by $k$-cliques.

    We now argue that $\dis_{(1)}(\sigma) \geq n+4$ which then yields a contradiction.
    Since $F$ does not contain a Hamiltonian cycle,
    $\sigma$ does not align $C_n$ perfectly with $F$. Specifically, there must be an edge $vv' \in E(C_n)$ with $\sigma(v)\sigma(v') \notin E(F)$, which implies that $v$ and $v'$ \enquote{lose} one neighbour and \enquote{gain} two new neighbours under $\sigma$, so 
    $\dis_1(\sigma, v), \dis_1(\sigma, v') \geq 3$ borrowing the notation introduced in \cref{equation:dis_1_pi_v}.
    For the remaining vertices $v'' \in V(C_n) \setminus \{v, v'\}$, since $\deg_{C_n}(v'') = 2$ but $\deg_F(w) = 3$ for all $w \in V(F)$ it follows that $\dis_1(\sigma, v'') \geq 1$.
    Note that $\dis_{(1)}(\sigma) = \sum_{v \in V(C_n)} \dis_1(\sigma, v)$ as both sides sum over all ordered pairs of vertices.
    Thus, it follows that
    $\dis_{(1)}(\sigma) \geq (n - 2) \cdot 1 + 2 \cdot 3 = n+4$.
    \bigskip
    
    The above shows that $\Ddist{(p)}(G, H) \leq (nk^2)^{1/p}$ if $F$ is a YES-instance and, conversely, $\Ddist{(p)}(G, H) \geq ((n+4)k^2)^{1/p}$ if $F$ is a NO-instance.
    To analyse the additive gap,
    let 
    \[
        D(n, k) \ceq ((n+4)k^2)^{1/p} - (nk^2)^{1/p}
    \]
    denote the minimum guaranteed difference of $\Ddist{(p)}$-values between YES- and NO-instances (parametrised by $n$ and $k$),
    and similarly, write 
    $N(n, k) \ceq 1 + nk + 2n(n+4)k^2$
    for the total number of vertices in the graphs constructed in the reduction.
    We prove the following relation between these two quantities which yields the inapproximability guarantee of the theorem.
    \begin{claim}
        \label{subclaim:beta}
        For every $0 < \beta < \frac{1}{p}$, there are $\alpha \in \N$ and $n_0 \in \N$ such that
        \begin{equation}
            \label{equation:to_show_ineq}
            D(n, n^\alpha) > N(n, n^\alpha)^{\beta}
        \end{equation}
        for all $n \geq n_0$.
    \end{claim}
    \begin{claimproof}
        We first derive a lower bound on $D(n, k)$. We have
            $D(n, k) = k^{2/p}\left( (n+4)^{1/p} - n^{1/p} \right)$.
        Applying the mean value theorem to the function $x \mapsto x^{1/p}$, there exists a $\zeta_n \in (n,n+4)$ such that 
        \[
            (n+4)^{1/p}-n^{1/p}
            =
            \mfrac{4}{p}\zeta_n^{1/p-1}.
        \]
        Since $\frac{1}{p}-1 \le 0$, the map $x \mapsto x^{1/p-1}$ is non-increasing on $(0,\infty)$. 
        Because $\zeta_n < n+4$, we obtain
        $\zeta_n^{1/p-1} \ge (n+4)^{1/p-1}$.
        Thus, we have
        \[
            D(n, k) \geq \mfrac{4}{p}k^{2/p}(n+4)^{1/p-1}.
        \]
        We choose $k \ceq n^\alpha$ and define $f: \N \to \Rnn,\, n \mapsto D(n, n^\alpha)$.
        For $n \geq 4$, it holds that $(n+4)^{1/p-1} \geq (2n)^{1/p-1}$ and thus
        \[
            f(n) \geq \mfrac{4}{p}2^{1/p-1}n^{\frac{2\alpha + 1}{p} - 1} = \Omega\left( n^{\frac{2\alpha +1}{p} -1 } \right).
        \]
        
        Next, observe that $g(n) \ceq N(n, n^\alpha)$ is dominated by the term $2n(n+4)n^{2\alpha}$, so there is a constant $C > 0$ such that $g(n) \leq Cn^{2+2\alpha}$, and thus $g(n) = \O(n^{2+2\alpha})$.
        To prove \cref{equation:to_show_ineq}, it suffices to prove that
        \[
            n^{\frac{2\alpha + 1}{p} - 1} > n^{(2\alpha+2)\beta}
        \]
        for all $n \geq n_0$
        where the multiplicative constants of the bounds on the functions $f$ and $g$ are omitted.
        It remains to show
        \begin{equation}
            \label{equation:sufficient_inequ}
            \mfrac{2\alpha + 1}{p} - 1 > (2\alpha + 2)\beta.
        \end{equation}
        Rearranging the above inequality for $\beta$, we obtain
        $\beta < \frac{2\alpha + 1 - p}{2p(1+\alpha)}$.
        Observe that the function $h: \Rnn \to \R,\, \alpha \mapsto \frac{2\alpha + 1 - p}{2p(1+\alpha)}$ is strictly increasing and $\lim_{\alpha \to \infty}h(\alpha) = \frac{1}{p}$.
        Hence, for any $0 < \beta < \frac{1}{p}$, there is an $\alpha \in \N$ with $\beta < h(\alpha)$.
        This proves \cref{equation:sufficient_inequ} and thus also \cref{equation:to_show_ineq} for a sufficiently large $n_0$.
    \end{claimproof}
    To obtain the inapproximability guarantee of the theorem, for each fixed $0<\eps< \frac{1}{p}$, let $\beta \ceq \frac{1}{p}-\frac{\eps}{2}$ and choose $k \ceq n^\alpha$ (for the $\alpha \in \N$ guaranteed by \cref{subclaim:beta}) in the construction of the graphs $G$ and $H$.
    Then, our reduction yields an additive gap of at least $N^{1/p-\eps/2} = \omega(N^{1/p-\eps})$ between YES- and NO-instances, where $N \ceq N(n, n^\alpha)$ is the order of the two input graphs for $\GSim{\Ddist{(p)}}$. Note that for $\eps \geq \frac{1}{p}$ the statement holds trivially. 
\end{proof}

\begin{proof}[Proof of
    \cref{theorem:global_distortion_into_path_hardness}]
    Membership in \NP{} is immediate since any bijection that certifies the inequality $\Ddist{(1)}(G, P) \leq k$ for an instance $(G, k)$ can be used as a certificate and only requires polynomially many bits.
    To show \NP-hardness, we give a reduction from \hampath{} \cite[Problem GT39]{garey1979computers}, which remains \NP-hard for graphs of diameter two.
    To see this, we give a simple reduction from arbitrary graphs to the case of graphs of diameter two.
    Given a graph $G$ on $n$ vertices, let $\widetilde{G}$ be the graph obtained from $G$ by adding a vertex $x$ that is connected to all vertices of $G$ and attaching one fresh leaf $y$ to $x$.
    The resulting graph has $n+2$ vertices, can be computed in polynomial time, and any Hamiltonian path in $G$ can be extended to a Hamiltonian path in $\widetilde{G}$ by appending first $x$ and then $y$. Conversely, any Hamiltonian path in $\widetilde{G}$ must either end or start with $y$ since it is a leaf. Since $x$ is the unique neighbour of $y$ it must be the neighbour in the Hamiltonian path, so in particular, the subpath excluding $x$ and $y$ yields a Hamiltonian path in $G$.

    Now let $G$ be an arbitrary graph of diameter two, let $n \ceq |G|$, $m \ceq |E(G)|$ and denote by $P$ the path on $n$ vertices with vertex set $V(P) = [n]$.
    For a bijection $\pi \in \Bij(V(G), V(P))$, we define $\ell_\pi(v, v') \ceq \Gdist_P(\pi(v), \pi(v'))$ for all $v, v' \in V(G)$.
    Then $\dis(\pi, v, v') = |\Gdist_G(v, v') - \ell_\pi(v, v')|$.
    Note that for distinct $v, v' \in V(G)$, it holds that $\dis(\pi, v, v') = |1 - \ell_\pi(v, v')| = \ell_\pi(v, v') - 1$ if $vv' \in E(G)$ and $\dis(\pi, v, v') = |2 - \ell_\pi(v, v')|$ otherwise.
    Rewriting the identity for non-edges using indicator variables for the values $\ell_\pi(v, v')$, we have
    \begin{equation}
        \label{equation:non_edge_global_path}
        \dis(\pi, v, v') = \ell_\pi(v, v') - 2 + 2 \cdot 
        \indicator[\ell_\pi(v, v') = 1].
    \end{equation}
    Finally, for a given $\pi$, define the set
    \[
        M(\pi) \ceq \{ i \in [n-1] \mid \pi^{-1}(i)\pi^{-1}(i+1) \notin E(G) \}
    \]
    of indices of the edges in $P$ that are not covered by an edge from $G$ under $\pi$.
    In particular, $|M(\pi)|$ is the number of non-edges in $G$ that are mapped onto edges of $P$.
    Thus, it holds that
    \begin{equation}
        \label{equation:size_of_M_global_path}
        2 \cdot |M(\pi)| = \sum_{\substack{v \neq v' \in V(G)\\vv' \notin E(G)}} 
        \indicator[\ell_\pi(v, v') = 1]
    \end{equation}
    where the factor of two appears since each (unordered) non-edge is counted twice on the right-hand side.
    Summing over all (ordered) pairs of vertices and grouping them according to their (non-)adjacency we have
    \begin{align*}
        \dis_{(1)}(\pi) 
        &= \sum_{v, v' \in V(G)} \dis(\pi, v, v') \\
        &= \sum_{\substack{v \neq v' \in V(G)\\vv' \in E(G)}} (\ell_\pi(v, v') - 1)
        + \sum_{\substack{v \neq v' \in V(G)\\vv' \notin E(G)}} \big( \ell_\pi(v, v') - 2 + 2 \cdot
        \indicator[\ell_\pi(v, v') = 1]
        \big) 
        \\
        &= \left(\sum_{v \neq v' \in V(G)} \ell_\pi(v, v') \right)
        - 2m - 4\cdot\left( \binom{n}{2} - m \right) + 4 \cdot |M(\pi)| 
        \\ 
        &= 2 \cdot \left(\sum_{1 \leq i < j \leq n} (j-i) \right)
        + 2m - 4 \cdot \binom{n}{2} + 4 \cdot |M(\pi)| \\ 
        &= 
        \underbrace{2 \cdot \left(\frac{n(n^2 - 1)}{6} + m - n(n-1) \right)}_{=: f(G)} 
        +\ 4 \cdot |M(\pi)|.
    \end{align*}
    Note that $f(G)$ depends on $G$ but is independent of the bijection $\pi$.
    In particular, if $G$ admits a Hamiltonian path, then $\pi$ can be chosen such that $|M(\pi)| = 0$, so $\dis_{(1)}(\pi) = f(G)$. Otherwise, $|M(\pi)| \geq 1$ and thus $\dis_{(1)}(\pi) \geq f(G) + 4$ for all $\pi$.
\end{proof}

\begin{proof}[Proof of \cref{theorem:additive_approx_global_distortion}]
    Let $\eps>0$ be fixed and let $G_1,G_2$ be connected graphs on $n$ vertices.
    Choose $\eps' \ceq \frac{1}{8}\eps$ and $r \ceq \eps' n$.
    For $i\in\{1,2\}$, let $B_i\subseteq V(G_i)$ be an inclusion-wise maximal set of vertices with pairwise distance greater than $r$. Note that such sets can be computed greedily in polynomial time.
    By maximality, every vertex has distance at most $r$ to some vertex of $B_i$.
    The balls of radius $\frac{r}{2}$ around the vertices of $B_i$ are pairwise disjoint.
    By connectivity, each ball contains at least $\frac{r}{2}$ vertices, unless $B_i$ is a single vertex.
    Hence $|B_i| \leq \max\big\{1, \frac{2n}{r} \big\} \leq \left\lceil \frac{2}{\eps'} \right\rceil$.

    Assign every vertex $v \in V(G_i)$ to a nearest vertex $b_i(v)\in B_i$ (we break ties arbitrarily), 
    and define the simplified pseudometric $d'_{G_i}(v,v') \ceq \Gdist_{G_i}(b_i(v), b_i(v'))$.
    Intuitively, this identifies all vertices assigned to the same representative.
    For $i \in \{1, 2\}$ and $v, v' \in V(G_i)$, the triangle inequality gives
    \[
        \big|\Gdist_{G_i}(v, v') - d'_{G_i}(v, v') \big| \leq \Gdist_{G_i}(v,b_i(v)) + \Gdist_{G_i}(v',b_i(v')) \leq 2\eps' n,
    \]
    and thus the distortion of every vertex pair under any bijection can be approximated with additive error at most $4\eps' n$.
    Summing over all $n^2$ vertex pairs gives an additive error of at most $4\eps' n^3$.
    Let $\pi^\ast$ denote an optimal bijection between $G_1$ and $G_2$ and $\pi'$ an optimal bijection between the simplified spaces (with metrics $d'_{G_1}$ and $d'_{G_2}$). Then
    \begin{align*}
        &\ \sum_{v,v'\in V(G_1)}\big|\Gdist_{G_1}(v, v') - \Gdist_{G_2}(\pi'(v), \pi'(v')) \big|\\
        \leq& \ \sum_{v,v'\in V(G_1)}\big|d'_{G_1}(v, v') - d'_{G_2}(\pi'(v), \pi'(v')) \big|+4\eps' n^3\\
        \leq& \ \sum_{v,v'\in V(G_1)}\big|d'_{G_1}(v, v') - d'_{G_2}(\pi^\ast(v), \pi^\ast(v')) \big|+4\eps' n^3\\
        \leq& \ \sum_{v,v'\in V(G_1)}\big|\Gdist_{G_1}(v, v') - \Gdist_{G_2}(\pi^\ast(v), \pi^\ast(v')) \big|+8\eps' n^3\\
        =& \ \Ddist{(1)}(G_1,G_2)+8\eps' n^3.
    \end{align*}
    By choice of $\eps'$, this yields an additive error guarantee of $\eps n^3$.
    
    It therefore suffices to compute an optimal bijection between the simplified pseudometric spaces $(V(G_1), d'_{G_1})$ and $(V(G_2), d'_{G_2})$.
    This reduces to finding a nonnegative integer matrix $M$ specifying how many vertices from each equivalence class of $d'_{G_1}$ are mapped to each equivalence class of $d'_{G_2}$, subject to the prescribed class sizes, so as to minimise the total distortion.
    Since $M$ has only $\O\big(\frac{1}{\eps^2}\big)$ entries whose values lie in $\{0, \ldots, n\}$, an optimal matrix can be found by brute-force enumeration in time $n^{\O(1/\eps^2)}$.

    The same approach can be used to obtain additive $\eps n$-approximations for $\GSim{\Ddist{(\infty)}}$ and $\GSim{\Mdist{\GH}}$ that run in time $2^{\O(1/\eps^2)} \cdot \poly(n)$.
    To compute an optimal bijection between the simplified pseudometric spaces, we enumerate all $2^{\O(1/\eps^2)}$ possible supports for the alignment between the classes. 
    Regarding $\Ddist{(\infty)}$, for any chosen support, we check whether it is realisable as a bijection by solving a bipartite flow problem which accounts for the numbers of vertices in the respective classes. 
    For $\Mdist{\GH}$, every support that covers all classes of both sides is realisable, so the feasibility check can be omitted. 
\end{proof}

\section{\texorpdfstring{Proofs Omitted from \cref{section:gw}}{Proofs Omitted from Section~\ref{section:gw}}}
\label{appendix:gw_hardness}
Before giving the proof of \cref{theorem:glob_gw_dist}, we note that
Mémoli \cite{memoli_gw_original} also introduced a variant of the Gromov--Wasserstein distance for 
$p=\infty$. 

\subsection{\texorpdfstring{$\infty$-Gromov--Wasserstein}{∞-Gromov--Wasserstein}}
\label{appendix:gw_hardness_infty}

For simplicity, we restrict the discussion to finite weighted metric spaces $\X = (X, d_{\X},
\mu_\X)$ and $\Y = (Y, d_{\Y}, \mu_\Y)$ with $\mu_\X(X) = \mu_\Y(Y)$. 
The \emph{support} of a coupling $\lambda \in \mathcal C(\mu_\X, \mu_\Y)$
is $\supp(\lambda) \ceq \big\{(x,y) \;\big|\; \lambda\big(\{(x,y)\}\big) > 0 \big\}$ and
\begin{equation}\label{eq:GW_infty}
    \Mdist{(\infty)}^{\GW}(\X, \Y) \ceq
    \mfrac{1}{2}\cdot \inf_{\lambda \in \mathcal C(\mu_\X, \mu_\Y)} 
    \underbrace{\max_{(x,y), (x',y') \in \supp(\lambda)} \abs*{d_\X(x,x')-d_\Y(y,y')}}_{\eqqcolon\, \dis_{(\infty)}
    (\lambda)}.
\end{equation}
Since $\dis_{(\infty)}(\lambda)$ depends only on $\supp(\lambda)$ and there are only finitely many
supports, the infimum in \cref{eq:GW_infty} is attained. 

Analogously to \cref{subsection:gw_derivation}, one can derive a graph metric $\Ddist{(\infty)}^{\GW}$ from 
$\Ddist{(\infty)}$ by replacing bijections with doubly stochastic matrices $Q \in \DS(V(G), V(H))$ and defining
$\dis_{(\infty)}(Q) \ceq \max \big\{\abs{D_{G}(v,v')-D_{H}(w,w')} \;\big|\; (v,w), (v',w') \in \supp(Q)\big\}$.
However, it is easy to see that $\Ddist{(\infty)}$ and $\Ddist{(\infty)}^{\GW}$ are equal: By Birkhoff's Theorem
\cite{birkhoff_theorem}
every doubly stochastic matrix is a convex combination of permutation matrices whose supports are contained
in $\supp(Q)$. 
Moreover, when considering graphs with uniform probability measures, the metrics 
$\Ddist{(\infty)}^{\GW}$ and $\Mdist{(\infty)}^{\GW}$ coincide up to a factor of $2$.

\begin{proposition}
    \label{proposition:gw_max_distortion_relation}
    For connected graphs $G, H$ of the same order, let $\X \ceq (V(G), \Gdist_G, \mu_G)$ and $\Y \ceq (V(H), \Gdist_H, \mu_H)$ be metric spaces where $\mu_G$ and $\mu_H$ are uniform probability measures on $V(G)$ and $V(H)$, respectively. 
    Then
    \begin{align*}
        2 \cdot \Mdist{(\infty)}^{\GW} (\X, \Y) 
        = \Ddist{(\infty)}^{\GW} ( G, H )
        = \Ddist{(\infty)} ( G, H  ).
    \end{align*}
\end{proposition}

\begin{proof}
    Recall that $\Ddist{(\infty)}^{\GW} ( G, H) = \Ddist{(\infty)} ( G, H )$.
    We prove $2 \cdot \Mdist{(\infty)}^{\GW} (\X, \Y) = \Ddist{(\infty)} ( G, H )$.
    Let $V$ and $W$ denote the vertex sets of $G$ and $H$ respectively
    and $n$ the order of the graphs. 
    For $\pi \in \Bij(V,W)$ let $\lambda_\pi\big((v,w)\big)\ceq\frac{1}{n}$ if $\pi(v)=w$
    and $0$ otherwise.
    Note that $\lambda_\pi \in \mathcal C(\mu_G, \mu_H)$ and $\supp(\lambda_\pi) =
    \{ (v, \pi(v)) \mid v \in V\}$, implying 
    \begin{align*}
        2 \cdot \Mdist{(\infty)}^{\GW} (\X, \Y)
        \leq \dis_{(\infty)}(\lambda_\pi)
        = \dis_{(\infty)}(\pi).
    \end{align*} 
    Minimising over all $\pi$ yields 
    $2 \cdot \Mdist{(\infty)}^{\GW} (\X, \Y) \leq \Ddist{(\infty)} ( G, H )$.

    Conversely, consider a $\lambda \in \mathcal C(\mu_G, \mu_H)$, let
    $S \ceq \supp(\lambda) \subseteq V \times W$, and define
    $$N(A) \ceq \{w \in W \mid (v,w) \in S \text{ for some } v \in A \}$$
    for every $A \subseteq V$. Then
    \begin{align*}
        \frac{\abs{A}}{n} = \mu_G(A) = \lambda (A \times W) = \lambda\big(A \times N(A) \big)
        \leq \lambda\big(V \times N(A) \big) = \mu_H\big(N(A)\big) = \frac{\abs{N(A)}}{n}.
    \end{align*}
    Hence, $\abs{N(A)} \geq \abs{A}$ holds for every $A \subseteq V$ and, by Hall's Marriage Theorem
    \cite[Theorem 2.1.2]{diestel2025graph},
    there is a $V$-perfect matching for the bipartite graph $(V \cup W, S)$. As $\abs{V}=\abs{W}$,
    such a matching defines a bijection $\pi_\lambda \colon V \to W$ with
    $\{(v,\pi_\lambda(v)) \mid v \in V\} \subseteq S$. Therefore 
    \begin{align*}
        \Ddist{(\infty)} (G, H)
        \leq \dis_{(\infty)}(\pi_\lambda) \leq \dis_{(\infty)}(\lambda).
    \end{align*}
    Taking the infimum over all $\lambda$ gives
    $\Ddist{(\infty)}( G, H ) \leq 2 \cdot \Mdist{(\infty)}^{\GW} (\X, \Y)$.
\end{proof}

As such, the hardness results of \cref{theorem:max_distortion_hardness} for 
$\GSim{\Ddist{(\infty)}}$
carry over to
$\GSim{\Mdist{(\infty)}^{\GW}}$ by endowing the constructed graphs with uniform measures.
Moreover, $\GSim{\Mdist{(\infty)}^{\GW}}$ is contained in $\NP$ for rational inputs:
A certificate for $\Mdist{(\infty)}^{\GW} (\X, \Y) \leq q$ is a set $S \subseteq X \times Y$
such that $\abs*{d_\X(x,x')-d_\Y(y,y')} \leq 2q$ for all $(x,y),(x',y') \in S$ and some $\lambda \in
\mathcal C(\mu_\X, \mu_\Y)$ satisfies $\supp(\lambda) \subseteq S$. The first condition can be checked 
directly and the second can be recast as a max-flow problem, thus verifying a certificate takes
polynomial time. 

\subsection{\texorpdfstring{Proof of \cref{theorem:glob_gw_dist}}{Proof of Theorem~\ref{theorem:glob_gw_dist}}}
    We start by proving membership in \NP{} for the decision variant of $\Ddist{(p)}^{\GW}$:
    Given graphs $G$ and $H$ on the same number of vertices and a rational value $q$, decide whether $\Ddist{(p)}^{\GW}(G, H) \leq q$.
     
    If $p = \infty$, we use $\Ddist{(\infty)}^{\GW} = \Ddist{(\infty)}$ (\cref{proposition:gw_max_distortion_relation}), so a bijection can be used as a certificate.
    If $p$ is an integer, the problem can be written as a quadratic program with integer coefficients.
    Thus, by \cite{vavasis_quadratic_programming_in_np}, the program has a rational optimal solution $Q^\ast \in \DS(V(G), V(H))$ where each entry of $Q^\ast$ only requires polynomial bitlength. Hence, the encoding of a witness only requires polynomially many bits and can be verified in polynomial time.

    We now turn to the \NP-hardness statement.
    The case $p = \infty$ follows from \cref{proposition:gw_max_distortion_relation} and \cref{theorem:max_distortion_hardness}.
    For the remainder of the proof, fix some $1 \leq p < \infty$.
    The reduction is again from \threepartition{}.
    Given an instance $\mathcal I = (A, a_1, \ldots, a_{3m})$ we construct the graphs $G'$ and $H'$ where $G'$ consists of $3m$ paths denoted $P^1, \ldots, P^{3m}$, where the $i$-th path is on $a_i$ vertices. Similarly, $H'$ consists of $m$ cycles, denoted $C^1, \ldots, C^m$, on $A$ vertices each. We denote the vertex sets of $G'$ and $H'$ by $V_2$ and $W_2$, respectively.
    Then, $G$ and $H$ are obtained by introducing a fresh vertex ($r_G$ and $r_H$, respectively) that is connected to all other vertices by an edge. We let $V_1 \ceq \{r_G\}$ and $W_1 \ceq \{r_H\}$.
    Finally, we add to both graphs a large clique on $M$ vertices (with vertex sets $V_0$ and $W_0$, respectively) and connect a single vertex of each (denoted $z_G$ and $z_H$) to $r_G$ and $r_H$ each by an edge. The concrete value of $M$ will be determined later in the proof and depends on the input instance.
    We refer to $V_0, V_1, V_2$ and $W_0, W_1, W_2$ as the \emph{layers} of $G$ and $H$ and let $V$ and $W$ denote the vertex sets, and $N$ the order of the two resulting graphs. Note that this construction is computable in polynomial time given that $A, M = \poly(m)$ and recall that we can safely assume 
    $A = \poly(m)$ by strong \NP-hardness of \threepartition{}.
    We claim that $\Ddist{(p)}^{\GW}(G, H) \leq (6m)^{1/p}$ if and only if $\mathcal I$ is a YES-instance.
    
    The idea of our construction is as follows: If $\mathcal I$ is a YES-instance, we may choose a bijection that maps triples of paths of $G[V_2]$ into cycles of $H[W_2]$ only incurring distortion of one at the endvertices of the paths. Intuitively, this amount of distortion is unavoidable due to degree differences in the two graphs.
    If $\mathcal I$ is a NO-instance, then we need to reason about all $Q \in \DS(V, W)$. Since each pair of vertices is considered in global distortion, the large clique of size $M$ essentially forces any $Q$ with \enquote{small} $\dis_{(p)}(Q)$ to respect the layer structure of the two graphs which allows us to argue that some path needs to be spread fractionally across multiple cycles since $\mathcal I$ is a NO-instance, \ie{} there is no way to map triples of paths into the cycles without \enquote{running out of space} in at least one cycle. This split incurs an additional amount of distortion which is sufficient to separate YES- from NO-instances based on the value of $\Ddist{(p)}^{\GW}(G, H)$.
    
    In the following, we assume that instances satisfy $A > m$ and $a_i \geq 4$ for all $i \in [3m]$. These assumptions can always be achieved by scaling the entries of a given instance.
    The reduction remains polynomial-time computable.

    \yescase If $\mathcal I$ is a YES-instance, then there is a witnessing assignment $\sigma: [3m] \to [m]$.
    We define a bijection $\pi: V \to W$ by mapping the vertices of $P^i$ into $C^j$ (mapping each edge of the path onto an edge of the cycle) whenever $\sigma(i) = j$. 
    This defines a bijection between $V_2$ and $W_2$ which we extend by mapping $V_1$ to $W_1$ (\ie{} $r_G$ to $r_H$) and mapping $V_0$ to $W_0$ such that $\pi(z_G) = z_H$.
    One easily checks that the only vertices incurring a distortion of one each are the endvertices of the $3m$ paths (in $G[V_2]$). All other vertices incur no distortion whatsoever.
    There are a total of $6m$ such endvertices, each having a distortion of one to exactly one other endvertex (the endvertex that is mapped next to it in some cycle).
    Hence $\dis_{(1)}(\pi) = 6m$, so in particular $\Ddist{(p)}^{\GW}(G, H) \leq \dis_{(p)}(\pi) = (6m)^{1/p}$.
    Formally, we identify $\pi$ with the unique permutation matrix $P \in \DS(V, W)$ satisfying $P(v, w) = 1$ 
    if and only if $\pi(v) = w$.

    \nocase Now assume that $\mathcal I$ is a NO-instance
    and assume towards a contradiction that $\dis_{(p)}(Q) \leq (6m)^{1/p}$ for some $Q \in \DS(V, W)$.
    In the following, we iteratively show restrictions on the structure of $Q$ until we have reached a contradiction.
    We start by showing that $Q$ needs to essentially preserve the layer structure, \ie{} map $V_i$ to $W_i$ for $i \in \{0, 1, 2\}$. Otherwise the distortion sum is already too large.
    Let
    \[
        \eta \ceq \sum_{i \in \{0, 1, 2\}} Q(V_i, W \setminus W_i)
    \]
    be the amount of mass that is mapped between distinct layers. 

    \begin{lemma}[\pc{Robin Westermann}, private communication]
        \label{subclaim:bounding_inter_layer_mapping}
        It holds that $\eta \leq \frac{12m}{N}$.
    \end{lemma}
    \begin{proof}
        Since $V_0$ and $W_0$ contain the vast share of vertices in the two graphs, $Q$ must transport much of the mass between these sets. Formally,
       \begin{equation}
            \label{equ:dist_sum_layer_bound}
            \begin{aligned}
                Q(V_0, W_0) &= Q(V_0, W) - Q(V_0, W \setminus W_0) \\
                &\geq Q(V_0, W) - Q(V, W \setminus W_0) \\
                &= M - (N-M) \\
                &\geq \mfrac{N}{2} + 2
            \end{aligned}
        \end{equation}
        where the last inequality holds if 
        $M \geq \frac{3}{4}N + 1$ (we will later choose $M$ accordingly).
        Hence
        \begin{equation}
            \label{equ:dist_layer_zero_without_connector}
            \begin{aligned}
                Q(V_0 \setminus \{z_G\}, W_0 \setminus \{z_H\}) &= \underbrace{Q(V_0, W_0)}_{\geq \frac{N}{2} + 2} - \underbrace{Q(z_G, W_0)}_{\leq 1} - \underbrace{Q(V_0, z_H)}_{\leq 1} + \underbrace{Q(z_G, z_H)}_{\geq 0} \\
                &\geq \mfrac{N}{2}.
            \end{aligned}
        \end{equation}
        Thus, we have
        \begin{align*}
            \left(\dis_{(p)}(Q)\right)^p 
            &= \sum_{\substack{v, v' \in V\\w, w' \in W}} |\Gdist_G(v, v') - \Gdist_H(w, w')|^p Q(v, w)Q(v',w') \\
            &\geq \sum_{i \in \{0, 1, 2\}} \sum_{\substack{v \in V_i,\, v' \in V_0 \setminus \{z_G\} \\ w \in W \setminus W_i,\, w' \in W_0 \setminus \{z_H\}}} |\Gdist_G(v, v') - \Gdist_H(w, w')|^p Q(v, w)Q(v',w') \\
            &\geq \sum_{i \in \{0, 1, 2\}} \sum_{\substack{v \in V_i,\, v' \in V_0 \setminus \{z_G\}\\ w \in W \setminus W_i,\, w' \in W_0 \setminus\{z_H\}}} Q(v, w)Q(v',w') \\
            &= \sum_{i \in \{0, 1, 2\}} Q(V_i, W \setminus W_i) Q(V_0 \setminus \{z_G\}, W_0 \setminus \{z_H\}) \\
            &\geq \frac{N}{2} \sum_{i \in \{0, 1, 2\}} Q(V_i, W \setminus W_i) \\
            &= \frac{N}{2} \cdot \eta
        \end{align*}
        where the last inequality holds by \cref{equ:dist_layer_zero_without_connector}.
        Moreover, by assumption we have $\big(\dis_{(p)}(Q)\big)^p \leq 6m$, so in particular $\eta \leq \frac{12m}{N}$ which concludes the proof.
    \end{proof}
    Using \cref{subclaim:bounding_inter_layer_mapping}, by choosing $M$ sufficiently large, we can ensure that the total mass that $Q$ maps between distinct layers can be bounded from above by any small constant.
    Hence, in the following analysis, we are able to \enquote{focus} on the intra-layer mapping of the mass. The inter-layer mapping is accounted for with error terms that depend on $\eta$.

    Going forward, our goal is to find an amount of distortion that yields $\dis_{(p)}(Q) > (6m)^{1/p}$.
    Roughly speaking, this distortion is split into two \enquote{disjoint} parts. 
    The first part is the distortion incurred by the endvertices of the paths in $G[V_2]$ and comes from the fact that these vertices have degree two in $G$, but in contrast, all vertices in $W_2$ have degree three.
    Thus, under the simplifying assumption that no inter-layer mapping occurs at all, each endvertex of a path yields a distortion of at least one.
    Note that this first part of distortion is inherent in the construction and was the part we also found in the case of a YES-instance. 
    However, since we now have to consider arbitrary $Q$, tracking this distortion down is more involved.

    For the second part of distortion, we explicitly use the fact that $\mathcal I$ is a NO-instance and argue that an edge along one of the paths in $G[V_2]$ must be split up fractionally across multiple cycles in $H$ by $Q$.
    Our argument yields an amount of distortion bounded away from zero.
    Together with the first part this is sufficient to obtain a contradiction.

    We start with the first part.
    Let $S$ denote the set of endvertices of the paths in $G[V_2]$ and recall that $|S| = 6m$.
    Define the quantity
    \begin{equation}
        \label{equation:first_part_contribution} 
        d_1(Q) \ceq \sum_{s \in S}\sum_{v \in V_2 \setminus N_G[s]} \sum_{w, w' \in W_2} |2 - \Gdist_H(w, w')|^p Q(s, w)Q(v, w').
    \end{equation}
    Intuitively, $d_1(Q)$ captures the amount of distortion incurred due to the degree difference of the vertices of $S$ in $G$ to those in $W_2$ in $H$.
    Each summand corresponds to the distortion of a pair $(s, v)$ where $s$ and $v$ are distinct and non-adjacent in $G$.
    \begin{lemma}
        \label{subclaim:degree_diff_distortion}
        It holds that $d_1(Q) \geq 6m - 3\eta$.
    \end{lemma}
    \begin{proof}
        For each $s \in S$, let $d_1(Q, s)$ denote the sum obtained by fixing the argument of the outer summation of \cref{equation:first_part_contribution} to be $s$. 
        Then, we have

        \begin{align*}
            d_1(Q, s) &= \sum_{v \in V_2 \setminus N_G[s]} \sum_{w, w' \in W_2} |2 - \Gdist_H(w, w')|^p Q(s, w)Q(v, w') \\
            &=\sum_{w \in W_2} Q(s, w) \left( \sum_{v \in V_2 \setminus N_G[s]} \sum_{w' \in W_2} |2 - \Gdist_H(w, w')|^p Q(v, w') \right) \\
            &=\sum_{w \in W_2} Q(s, w) \Biggl(
            \sum_{v \in V_2 \setminus N_G[s]}
            \Bigl(
            \sum_{w' \in W_2 \setminus N_H[w]}
            \underbrace{|2 - \Gdist_H(w, w')|^p}_{=0}Q(v, w') \\
            &\hspace{4.8cm}
            + \sum_{w' \in W_2 \cap N_H[w]}
            \underbrace{|2 - \Gdist_H(w, w')|^p}_{\geq 1}Q(v, w')
            \Bigr)
            \Biggr) \\
            &\geq \sum_{w \in W_2} Q(s, w) \sum_{v \in V_2 \setminus N_G[s]} 
            \left( \sum_{w' \in N_H[w] \cap W_2} Q(v, w')  \right) \\
            &= \sum_{w \in W_2} Q(s, w) \underbrace{Q(V_2 \setminus N_G[s], N_H[w] \cap W_2)}_{\eqqcolon (\ast)}.
        \end{align*}
        Assuming no inter-layer mapping ($\eta = 0$),
        a simple counting argument would already show that $(\ast) \geq 1$:
        We have $|V_2 \setminus N_G[s]| = |V_2| - 2$ but $|W_2 \setminus N_H[w]| = |V_2| - 3$, and hence at least one unit of mass must be transported to $N_H[w] \cap W_2$.
        Since $d_1(Q) = \sum_{s \in S} d_1(Q, s)$, this would immediately yield $d_1(Q) \geq 6m$.
        
        We now drop the assumption of no inter-layer mapping.
        To this end, let $\alpha \ceq Q(S, W_0 \cup W_1)$ and $\beta \ceq Q(V_0 \cup V_1, W_2)$ denote parts of the mass that is mapped between layers.
        Moreover, for $w \in W_2$, define $\beta(w) \ceq Q(V_0 \cup V_1, N_H[w] \cap W_2)$. Note that each $w' \in W_2$ appears in precisely three closed neighbourhoods of the form $N_H[w] \cap W_2$, so $\sum_{w \in W_2} \beta(w) = 3\beta$. 

        We rewrite
        \begin{align*}
            \underbrace{Q(V_2 \setminus N_G[s], N_H[w] \cap W_2)}_{(\ast)} &= 
            Q(V_2, N_H[w] \cap W_2) - \underbrace{Q(N_G[s] \cap V_2, N_H[w] \cap W_2)}_{\leq 2} \\
            &\geq \underbrace{Q(V, N_H[w] \cap W_2)}_{=3} - \underbrace{Q(V_0 \cup V_1, N_H[w] \cap W_2)}_{= \beta(w)} - \,2 \\
            &= 1 - \beta(w).
        \end{align*}
        Substituting this into the previous inequalities, we obtain
        \begin{align*}
            d_1(Q) &= \sum_{s \in S} d_1(Q, s) \\
                &\geq \sum_{s \in S}\sum_{w \in W_2}Q(s, w) \underbrace{Q(V_2 \setminus N_G[s], N_H[w] \cap W_2)}_{(\ast)} \\
                &\geq \sum_{s \in S} \sum_{w \in W_2} Q(s, w)(1 - \beta(w)) \\
                &= \sum_{s \in S} \sum_{w \in W_2} Q(s, w) - \sum_{w \in W_2} \beta(w) \underbrace{\sum_{s \in S} Q(s, w)}_{\leq 1} \\
                &\geq Q(S, W_2)  - \sum_{w \in W_2} \beta(w) \\
                &= Q(S, W_2) - 3\beta \\
                &= Q(S, W) - Q(S, W_0 \cup W_1) - 3\beta \\
                &\geq Q(S, W) - \alpha - 3\beta \\
                &= 6m - \alpha - 3\beta.
        \end{align*}
        By definition of $\eta$ it follows that $\alpha + \beta \leq \eta$ and thus $6m - \alpha - 3\beta \geq 6m - 3\eta$.
    \end{proof}

    We now turn to the second part of distortion.
    For each $i \in [3m]$, define $\Pi_i \ceq \{(v, v') \in V(P^i)^2 \mid v = v' \text{ or } vv' \in E(P^i) \}$.
    Define the sum
    \begin{equation}
        \label{equ:d2Q}
        d_2(Q) \ceq \sum_{i \in [3m]} \sum_{(v, v') \in \Pi_i} \sum_{w, w' \in W} |\Gdist_G(v, v') - \Gdist_H(w, w')|^p Q(v, w)Q(v', w').
    \end{equation}
    Intuitively, this captures the amount of distortion of vertices in $V_2$, either to themselves or to their neighbours along their respective paths.
    Hence, the distortion captured by $d_2(Q)$ is entirely disjoint from that of $d_1(Q)$.

    \begin{lemma}
        \label{subclaim:d2Q}
        It holds that $d_2(Q) \geq \frac{1}{4} - 2A\eta$.
    \end{lemma}

    \begin{proof}
        For each $i \in [3m]$, let $d_2(Q, i)$ denote the sum obtained by fixing the argument of the outer summation of \cref{equ:d2Q} to be $i$.
        Clearly $d_2(Q, i) \leq d_2(Q)$ for each $i \in [3m]$.
    
        We now develop a lower bound for $d_2(Q, i)$ that holds for at least one such $i$.
        To this end, we start by developing the necessary formalism.
        For $v \in V_2$ and $j \in [m]$, define $\eps^j_v \ceq Q(v, V(C^j))$.
        Moreover, for $i \in [3m]$ we let $\eps^{i,j} \ceq \sum_{v \in V(P^i)} \eps^j_v$.
        Intuitively, $\eps^j_v$ captures the amount of mass of $v$ that is mapped into $C^j$ and $\eps^{i,j}$ captures the total mass of the path $P^i$ that is mapped into $C^j$.
        
        Let $v, v'$ be neighbours in some path $P^i$ and let $j \in [m]$ be arbitrary.
        In the following, we develop an approximation of the distortion incurred by the pair $\{v, v'\}$ under $Q$ that is relative to the cycle $C^j$.
        For this, we use the fact that $\eps^j_v$- and $\eps^j_{v'}$-fractions are mapped into $C^j$ and that fractions of at least $(1-\eps^j_v - \eta)$ and $(1 - \eps^j_{v'} - \eta)$ are mapped into $W_2 \setminus V(C^j)$. 
        For simplicity, we disregard all other terms of distortion occurring within $C^j$ and $W \setminus V(C^j)$, respectively.
        We thus define the following lower bound:
        \begin{equation*}
            \label{equ:dis_pair_bound_orig}
                \dis(Q, v, v', j) \ceq \eps^j_v\big(1 - \eps^j_{v'} - \eta\big) + \big(1 - \eps^j_v - \eta\big) \eps^j_{v'} + \eps^j_v\big(1-\eps^j_v - \eta\big) +  \eps^j_{v'}\big(1-\eps^j_{v'} - \eta\big).
        \end{equation*}
        Here, the first two quantities account for the distortion of one between the fractions of $v$ inside $C^j$ and $v'$ in $W_2 \setminus C^j$ (and vice versa), respectively, and the third and fourth quantities account for \emph{half} of the distortion (of two) incurred by each $v$ and $v'$ by being split up fractionally across $C^j$ and $W_2 \setminus C^j$: This is because vertices in distinct cycles are at distance two.
        
        Note that we only consider half of the latter quantities to prevent double counting of these contributions as we will sum up $\dis(Q, \cdot, \cdot, j)$ along an entire path $P^i$ (each interior vertex of a path appears in exactly two edges) to obtain a lower bound on $d_2(Q, i)$.
    
        First, we simplify $\dis(Q, v, v', j)$ by separating the contribution of $\eta$, that is
        \begin{equation}
            \begin{aligned}
                \label{equ:dis_pair_bound}
                \dis(Q, v, v', j) &= \eps^j_v\big(1 - \eps^j_{v'} - \eta\big) + \big(1 - \eps^j_v - \eta\big) \eps^j_{v'} + \eps^j_v\big(1-\eps^j_v - \eta\big) +  \eps^j_{v'}\big(1-\eps^j_{v'} - \eta\big) \\
                 &= 2\eps^j_v + 2\eps^j_{v'} - 2\eps^j_v\eps^j_{v'} - 
                 \big(\eps^j_v\big)^2 - \big(\eps^j_{v'}\big)^2 - \underbrace{\big(2\eps^j_v \eta + 2\eps^j_{v'}\eta\big)}_{\leq 4\eta} \\
                &\geq \big(\eps^j_v + \eps^j_{v'}\big)\big(2 - \eps^j_v - \eps^j_{v'}\big) - 4\eta.
            \end{aligned}
        \end{equation}
        Note that we can make $\eta$ arbitrarily small using \cref{subclaim:bounding_inter_layer_mapping}.
        Letting
        \[
            \dis(Q, i, j) \ceq \sum_{vv' \in E(P^i)} \dis(Q, v, v', j),
        \]
        for $i \in [3m]$ and $j \in [m]$, we have
        \begin{equation}
            \label{equation:d2Q_lb_dis_Qij}
            d_2(Q) \geq d_2(Q, i) \geq \dis(Q, i, j)
        \end{equation}
        for each $i \in [3m]$ and $j \in [m]$. It hence suffices to bound one such $\dis(Q, i, j)$ from below.
        
        \begin{claim}
            \label{subclaim:finding_a_quarter}
            There exist $i \in [3m]$ and $j \in [m]$ such that $\dis(Q, i, j) \geq \mfrac{1}{4} - 2A\eta$.
        \end{claim}
        \begin{claimproof}
            First, suppose there are $v, v' \in V(P^i)$ for some $i \in [3m]$ with $vv' \in E(P^i)$ and there is $j \in [m]$ such that $\eps^j_v \leq \frac{1}{2}$ and $\eps^j_{v'} \geq \frac{1}{2}$ or vice versa. Then \cref{equ:dis_pair_bound} yields that 
            \begin{align}
                \label{equation:bounding_single_edge}    
                \dis(Q, v, v', j) &\geq \big(\eps^j_v + \eps^j_{v'}\big)\big(2 - \eps^j_v - \eps^j_{v'}\big) - 4\eta \geq \mfrac{1}{4} - 4\eta
            \end{align}
            as both factors are bounded from below by $\frac{1}{2}$.
            Moreover, for every edge $vv' \in E(P^i)$, we have $\dis(Q, v, v', j) \geq -4\eta$, which, together with \cref{equation:bounding_single_edge}, implies
            \[
                \dis(Q, i, j) \geq \mfrac{1}{4} - |E(P^i)|\cdot 4\eta \geq \mfrac{1}{4} - 2A\eta
            \]
            since $|E(P^i)| = a_i - 1 <  \frac{A}{2}$.
        
            In the following, assume that either $\eps^j_v, \eps^j_{v'} \leq \frac{1}{2}$ or $\eps^j_v, \eps^j_{v'} > \frac{1}{2}$ for all pairs $v, v'$ of adjacent vertices in the paths $P^i$ and all cycles $C^j$.
            In particular, this implies that all vertices belonging to the same path either all map at most half their mass into some cycle $C^j$ or all vertices map strictly more than half their mass into $C^j$.
            Since $\frac{A}{4} < a_i$, with $A \geq m+1$ it follows that $4a_i \geq A+1 \geq m+2$ and thus $\frac{a_i - \eta}{m} > \frac{1}{4}$ since $\eta < \frac{1}{2}$ (using \cref{subclaim:bounding_inter_layer_mapping}).
            Hence, there is a $j \in [m]$ such that $\eps^{i, j} > \frac{1}{4}$. 
            
            For the first case, we assume that $\eps_v^j \leq \frac{1}{2}$ for all $v \in V(P^i)$.
            Then we have
            \begin{align*}
                \dis(Q, i, j) &= \sum_{vv' \in E(P^i)} \dis(Q, v, v', j) \\
                &\geq \sum_{vv' \in E(P^i)}
                \Big( \big(\eps^j_v + \eps^j_{v'}\big)\underbrace{\big(2 - \eps^j_v - \eps^j_{v'}\big)}_{\geq 1} - \ 4\eta \Big) \tag*{\cref{equ:dis_pair_bound}}\\
                &\geq \sum_{vv' \in E(P^i)}  \big(\eps^j_v + \eps^j_{v'} - 4\eta\big) \\
                &\geq \eps^{i, j} - 2A\eta \tag*{$|E(P^i)| \leq \frac{A}{2}$}\\
                &> \mfrac{1}{4} - 2A\eta.
            \end{align*}
            If the above holds for a single $i \in [3m]$, we are done.
            If not, that is, the second case holds, then there is a $j \in [m]$ such that $\eps_v^j > \frac{1}{2}$ for all $v \in V(P^i)$.
            Since these inequalities are strict, there is a \emph{unique} $j \in [m]$ with this property, denoted $j(i)$ in the following.
            Then
            \begin{align*}
                \dis(Q, i, j(i)) &= \sum_{vv' \in E(P^i)} \dis(Q, v, v', j(i)) \\
                &\geq \sum_{vv' \in E(P^i)} 
                \Big( \underbrace{\big(\eps^{j(i)}_v + \eps^{j(i)}_{v'}\big)}_{\geq 1}\big(2 - \eps^{j(i)}_v - \eps^{j(i)}_{v'}\big) - \ 4\eta \Big) \tag*{\cref{equ:dis_pair_bound}}\\
                &\geq \sum_{vv' \in E(P^i)} 
                \Big( \big(2 - \eps^{j(i)}_v - \eps^{j(i)}_{v'}\big) - 4\eta \Big) \\
                &\geq \sum_{v \in V(P^i)} \big(1-\eps^{j(i)}_v - 4\eta\big)  \\
                &\geq |P^i| - \eps^{i, j(i)} - 2A\eta. \tag*{$|V(P^i)| \leq \frac{A}{2}$}
            \end{align*}
            We claim that the last quantity can be bounded from below by $\frac{1}{4} - 2A\eta$ for some $i \in [3m]$.
            Suppose for contradiction that
            for each $i \in [3m]$ it holds 
            $|P^i| - \eps^{i, j(i)} < \frac{1}{4}$, or equivalently $\eps^{i,j(i)} > |P^i| - \frac{1}{4}$.
            Define the mapping $\sigma: [3m] \to [m],\, i \mapsto j(i)$.
            Then, for each $j \in [m]$ we have
            \[
                A \geq \sum_{i \in \sigma^{-1}(j)} \eps^{i, j(i)} > \left(\sum_{i \in \sigma^{-1}(j)} |P^i|\right) - \mfrac{|\sigma^{-1}(j)|}{4}.
            \]
            Since $\frac{A}{4} < |P^i| < \frac{A}{2}$, it follows that $|\sigma^{-1}(j)| \leq 3$.
            As the domain of $\sigma$ has $3m$ elements, it follows that $|\sigma^{-1}(j)| = 3$ for all $j \in [m]$.
            But then
            \[
                A \geq \sum_{i \in \sigma^{-1}(j)} \eps^{i, j(i)} > \left(\sum_{i \in \sigma^{-1}(j)} |P^i|\right) - \mfrac{3}{4}
            \]
            it follows that $\sum_{i \in \sigma^{-1}(j)} a_i < A+1$. Together with $\sum_{i \in [3m]} a_i = mA$, this implies that $\sigma$ is a valid partition function witnessing that $\mathcal I$ is a YES-instance. This is a contradiction.
            Hence, there exist $i \in [3m]$ and $j \in [m]$ with $\dis(Q, i, j) \geq \frac{1}{4} - 2A\eta$.
        \end{claimproof}
        With \cref{subclaim:finding_a_quarter}, the lemma follows.
    \end{proof}
        
    \noindent Putting everything together, we have
    \begin{align*}
        \dis_{(p)}(Q) &\geq \Big( d_1(Q) + d_2(Q) \Big)^{1/p} \\
        &\geq \left( 6m - 3\eta + \mfrac{1}{4} - 2A\eta \right)^{1/p} \tag*{\cref{subclaim:degree_diff_distortion,subclaim:d2Q}}\\
        &\geq \left( 6m + \mfrac{1}{4} - 3A\eta \right)^{1/p} \\
        &\geq \left( 6m + \mfrac{1}{4} - \mfrac{36mA}{N} \right)^{1/p} \tag*{\cref{subclaim:bounding_inter_layer_mapping}} \\
        &> (6m)^{1/p}
    \end{align*}
    where the last inequality holds for $N > 144mA$.
    We now conclude the proof by choosing $M \ceq 144mA + 1$ which is also sufficient for \cref{equ:dist_sum_layer_bound} to remain valid since $N = M + 1 + mA$.
    Note that by the strong \NP-hardness of \threepartition{}, $M$ remains polynomial in the input size.
    As a final remark, we note that the value $(6m)^{1/p}$ may be irrational. However, as we find a nonzero gap between the values resulting from YES- and NO-instances, we may choose a rational threshold (of polynomial precision). This proves \cref{theorem:glob_gw_dist}.

\end{document}